\documentclass[a4paper,english,pagebackref]{article} 

\newif\iflong
\newif\ifshort

\longtrue

\iflong
\else
\shorttrue
\fi

\usepackage{hyperref}
\usepackage{amsmath,amssymb,xspace}
\usepackage[capitalize]{cleveref}
\usepackage[disable,textsize=scriptsize]{todonotes}
\usepackage{paralist}

\usepackage{tikz}
\usetikzlibrary{decorations.pathmorphing}

\usepackage{color}
\usepackage{graphicx}
\definecolor{mygrey}{rgb}{0.9,0.9,0.9}
\definecolor{darkgreen}{RGB}{0,100,0}

\usepackage{amsthm}
\usepackage{authblk}
\usepackage[numbers,sort]{natbib}
\usepackage{mathtools}
\usepackage[margin=3cm]{geometry}
\usepackage{url}
\renewcommand{\citet}[1]{\cite{#1}}

\theoremstyle{plain}
\newtheorem{theorem}{Theorem}
\newtheorem{proposition}[theorem]{Proposition}

\newtheorem{observation}[theorem]{Observation}
\newtheorem{step}{Step}
\newtheorem{lem}[theorem]{Lemma}
\crefname{step}{Step}{Steps}
\crefname{lem}{Lemma}{Lemmas}
\crefname{observation}{Observation}{Observations}
\crefname{proposition}{Proposition}{Propositions}
\crefname{prop}{Property}{Properties}
\creflabelformat{prop}{(#2#1#3)}

\newcommand{\yes}{\textsc{yes}\xspace}

\renewcommand{\P}{\ensuremath{\mathcal{P}}}

\newcommand{\FPT}{FPT\xspace}
\newcommand{\W}[1]{W[\ensuremath{#1}]\xspace}

\newcommand{\poly}{\ensuremath{\operatorname{poly}}}

\DeclareMathOperator{\tw}{tw}

\DeclareMathOperator{\dist}{dist}
\DeclareMathOperator{\diam}{diam}

\newcommand{\prob}[5]{%
  \begingroup
  \par\medskip
  \noindent \textsc{#1}\nopagebreak[4]
  \par\noindent\hangindent=\parindent\textit{#2}  #3
  \par\noindent\hangindent=\parindent\textit{#4}  #5
  \par  \medskip
  \endgroup
}

\newcommand{\decprob}[3]{\prob{#1}{Input:}{#2}{Question:}{#3}}

\newcommand{\msetsc}{\textsc{Minimum Shared Edges}}

\newcommand{\MSE}{\textsc{MSE}\xspace}
\newcommand{\MSEl}{\textsc{Minimum Shared Edges}\xspace}
\newcommand{\acrmse}[1]{\textsc{MSE(#1)}}

\newcommand{\AMSEl}{\textsc{Almost Minimum Shared Edges}\xspace}
\newcommand{\AMSE}{\textsc{AMSE}\xspace}
\newcommand{\MCC}{\textsc{MCC}\xspace}
\newcommand{\MCCl}{\textsc{Multicolored Clique}\xspace}

\newcommand{\mvtsc}{\textsc{Minimum Vulnerability}}

\newcommand{\sctsc}{\textsc{Set Cover}}
\newcommand{\acrsc}{\textsc{SC}}

\newcommand{\Sids}{Sidon set\xspace}
\newcommand{\Sidss}{Sidon sets\xspace}

\title{The Parameterized Complexity of the Minimum Shared Edges Problem\footnote{Till Fluschnik, Stefan Kratsch, and Manuel Sorge gratefully acknowledge support by the DFG, projects DAMM (NI 369/13-2), PRE-
MOD (KR 4286/2-1), and DAPA (NI 369/12-2), respectively. An extended abstract appeared in: Proceedings of the 35th IARCS Annual Conference on Foundations of Software Technology and Theoretical Computer Science (FSTTCS 2015), volume 45, pages 448--462.}}

\author[1]{Till Fluschnik}
\author[2]{Stefan Kratsch}
\author[1]{Rolf Niedermeier}
\author[1]{Manuel~Sorge}

\affil[1]{\small{Institut f\"{u}r Softwaretechnik und Theoretische Informatik, TU~Berlin, Germany, \texttt{\{till.fluschnik, rolf.niedermeier, manuel.sorge\}@tu-berlin.de}}}
\affil[2]{\small{Institut f\"{u}r Informatik, Universit\"{a}t Bonn, Germany, \texttt{kratsch@cs.uni-bonn.de}}}

\begin{document}

\maketitle

\begin{abstract}
We study the NP-complete \textsc{Minimum Shared Edges (MSE)} problem. Given an undirected graph, a source and a sink vertex, and two integers $p$ and $k$, the question is whether there are $p$~paths in the graph connecting the source with the sink and sharing at most $k$ edges. Herein, an edge is shared if it appears in at least two paths. 
We show that MSE is %
W[1]-hard when parameterized by the treewidth of the input graph and the number~$k$ of shared edges combined. We show that MSE is fixed-parameter tractable with respect to~$p$, but does not admit a polynomial-size kernel (unless $\text{NP}\subseteq\text{coNP/poly}$). In the proof of the fixed-parameter tractability of MSE parameterized by~$p$, we employ the treewidth reduction technique due to Marx, O'Sullivan, and Razgon [ACM TALG 2013].

\noindent
\textbf{Key words:} Fixed-parameter tractability; W-hardness; kernelization; 
tree decompositions of graphs; treewidth reduction technique; VIP routing.
\end{abstract}


\section{Introduction}
We consider the parameterized complexity of the following basic routing problem.

\decprob{\MSEl (\MSE)}{A graph $G=(V,E)$, $s,t\in V$, $p\in\mathbb{N}$ and $k\in\mathbb{N}_0$.}{Is there a $(p, s, t)$-routing in $G$ in which at most $k$ edges are shared?}%
\noindent Herein, a \emph{$(p, s, t)$-routing} is a set of $s$-$t$ paths with cardinality~$p$, and an edge is called \emph{shared} if it is contained in at least two of the paths in the routing. If $s$ and $t$ are understood from the context, we simplify notation and speak of a $p$-routing and call the paths it contains \emph{routes}. %
\MSEl is polynomial-time solvable with~$k=0$, while it becomes NP-hard for general values of~$k$~\cite{Flu15}. 

\looseness=-1 \MSEl has two natural applications. One is to route an important person which is under threat of attack from $s$ to $t$ in a street network. In order to confound attackers, $p-1$ additional, empty convoys are routed, and guards are placed on streets that are shared by routes. \MSEl then minimizes the costs to place guards~\cite{OmranSZ13}.

A second application arises from finding a resilient way of communication between two servers $s$ and $t$ in an interconnection network, assuming that $p-1$ faulty connections may be present that block or alter the communicated information. Finding $p$ edge-disjoint paths ensures at least one piece of information arrives unscathed. When this is not possible, and if we can ensure that a link is not faulty by expending some fixed cost per link, then \MSEl is the problem of finding a resilient way of communication that minimizes the overall costs~\cite{Pel96}.

We study \MSEl{} from a parameterized complexity perspective, that is, for certain parameters $\ell$ of the inputs (of size~$r$), we identify algorithms with running time $f(\ell) \cdot \poly(r)$ or we prove that such algorithms are unlikely to exist. There are two natural parameters for \MSEl: the number~$p$ of routes and the number~$k$ of shared edges. Both of them can be reasonably assumed to be small in applications. As we will see, there is also a connection between $p$ and the treewidth~$\tw$ of~$G$.

\subparagraph*{Related Work.}
Omran et al.~\cite{OmranSZ13} introduced \MSEl{} on directed graphs and showed NP-hardness by a reduction from {\sc Set Cover}. The reduction also implies W[2]-hardness with respect to the number~$k$ of shared arcs in this directed case. Undirected \MSEl admits an XP-algorithm with respect to treewidth, more specifically, it can be solved in $O((n + m)\cdot (p+1)^{2^{\omega\cdot(\omega+1)/2}})$ time%
~\cite{YeLLZ13}. 

Assadi et al.~\cite{AssadiENYZ12} introduced a generalization of directed \MSEl, called \mvtsc{}, which additionally considers arc weights (the cost of sharing an arc), arc capacities (an upper bound on the number of routes supported by an arc) and a share-threshold for each arc (the threshold of routes, possibly other than two, after which the arc becomes shared). Directed \mvtsc{} admits an XP-algorithm with respect to the number~$p$ of routes~\cite{AssadiENYZ12}. Undirected \mvtsc{} is NP-hard even on bipartite series-parallel graphs, but admits a pseudo-polynomial-time algorithm on bounded treewidth graphs~\cite{AokiHHIKZ14}. Furthermore, \mvtsc{} is fixed-parameter tractable with respect to~$p$ on chordal graphs~\cite{AokiHHIKZ14}.
Recently, Fluschnik and Sorge~\cite{FluschnikSorge16} showed that \MSE remains NP-hard on planar graphs of maximum degree four.

There are also several results regarding approximation algorithms and lower bounds~\cite{AssadiENYZ12,OmranSZ13}; however, our focus is on exact algorithms.

\subparagraph*{Our Contributions.} 
First we show that \MSEl{} is NP-complete and W[2]-hard with respect to the number~$k$ of shared edges (\cref{sec:msenphard}).
We then prove two main results, namely, that \MSEl is fixed-parameter tractable (FPT) with respect to the number~$p$ of routes (\cref{sec:twred}) and that it is W[1]-hard with respect to the treewidth~$\tw{}$ and the number~$k$ of shared edges combined (\cref{sec:whard-tw}). 
Moreover, complementing the fixed-parameter tractability result with respect to~$p$, we show that there is no polynomial-size problem kernel with respect to~$p$ (\cref{sec:nopoly}).

The FPT result with respect to~$p$ is obtained by modifying the input graph so that the resulting graph has treewidth bounded by some (exponential) function of~$p$ using the treewidth reduction technique~\cite{MarxOR13} (see \cref{sec:twred}). Then we apply a dynamic program which also is an FPT algorithm with respect to $p$ and~$\tw$~(\cref{sec:dp}). For this purpose, we design a new dynamic program rather than using the ones from the literature~\cite{AokiHHIKZ14,YeLLZ13}. In comparison, ours yields an improved running time in the FPT algorithm with respect to~$p$, that is, the dependence is doubly exponential on~$p$ rather than triply exponential. Our result complements the known FPT algorithm for undirected \mvtsc{} on chordal graphs, parameterized by~$p$. Treewidth reduction has lately also found application in a wide varity of problems, for example, in graph coloring~\cite{ECM13}, graph partitioning~\cite{BESS15}, and arc routing~\cite{GJS14}.

As mentioned, our second main result is that \MSEl is W[1]-hard with respect to the treewidth~$\tw$ and the number~$k$ of shared edges combined. This provides a corresponding lower bound for the known polynomial-time algorithms on constant treewidth graphs for \MSEl and for the more general undirected \mvtsc{}~\cite{AokiHHIKZ14,YeLLZ13}. More precisely, the exponents in the running time depend on~$\tw$ and our result shows that removing this dependence is impossible unless FPT${} = {}$W[1]. Interestingly, the known dynamic programs on tree decompositions keep track of the number of routes over certain separators in their tables. Our hardness result shows that information of this sort is crucial, that is, it is unlikely that there are dynamic programs with table entries relying only on information bounded by the treewidth.

%


\section{Preliminaries}\label{sec:prelim}
We use standard notation from parameterized complexity~\cite{DowneyF13,FG06,Nie06,CyganFKLMPPS15} and graph theory~\cite{Diestel10,west}.

\subparagraph*{Graphs and Tree Decompositions.} Unless stated otherwise, all graphs are without parallel edges or loops. When it is not ambiguous, we use $n$ for the number of vertices of a graph and $m$ for the number of edges. 

Let $G=(V,E)$ be an undirected graph. We write $V(G)$ for the vertex set of graph~$G$ and $E(G)$ for the edge set of graph~$G$. We define the size of graph~$G$ as $|G|:=|V(G)|+|E(G)|$. \ifshort{} For a vertex set~$W\subseteq V(G)$, we denote by~$G[W]$ the subgraph of~$G$ \emph{induced} by the vertex set~$W$. For an edge set $F\subseteq E(G)$, we denote by $G[F]$ the subgraph of~$G$ \emph{induced} by the edge set $F$.\fi{}\iflong{} For a vertex set $W\subseteq V(G)$, we denote by~$G[W]$ the subgraph of $G$ with vertex set $\{v\in V(G)\mid v\in W\}$ and edge set $\{\{v,w\}\in E(G)\mid v,w\in W\}$. We say that $G[W]$ is the subgraph of~$G$ \emph{induced} by the vertex set $W$. For an edge set $F\subseteq E(G)$, we denote by $G[F]$ the subgraph of $G$ with vertex set $\{v\in V(G)\mid (e\in F)\wedge (v\in e)\}$ and edge set $\{e\in E(G)\mid e\in F\}$. We say that $G[F]$ is the subgraph of~$G$ \emph{induced} by the edge set $F$. \fi{} \iflong{}For an edge $e\in E$, we denote by $G/\{e\}$ the contraction of edge~$e$ in~$G$, and we denote by $G\backslash \{e\}$ the deletion of edge~$e$ in~$G$ (we write $G/e$ and $G\backslash e$ for short). Consequently, for a set of edges $F\subseteq E$ we write $G/F$ and $G\backslash F$ for the contraction and the deletion of the edges in~$F$, respectively.\fi{}\ifshort{}We write $G/F$ and $G\backslash F$ for the contraction and the deletion of the edges in~$F$, respectively (we write $G/e$ and $G\backslash e$ for short if $F=\{e\}$).\fi{} We write $\Delta(G)$ to denote the maximum degree of graph $G$ and $\diam(G)$ to denote the diameter of~$G$. 

A \emph{tree decomposition} of a graph $G$ is a tuple $\mathbb{T}:=(T,(B_\alpha)_{\alpha\in V(T)})$ of a tree $T$ and family $(B_\alpha)_{\alpha\in V(T)}$ of sets $B_\alpha\subseteq V(G)$, called \emph{bags}, such that~$V(G)=\bigcup_{\alpha\in V(T)} B_\alpha$,%
\begin{compactenum}[(i)]
\item for every edge $e\in E(G)$ there exists an~$\alpha\in V(T)$ such that $e\subseteq B_\alpha$ and
\item for each $v\in V(G)$, the graph induced by the node set $\{\alpha\in V(T)\mid v\in B_\alpha\}$ is a tree.
\end{compactenum}
The \emph{width} $\omega$ of a tree decomposition $\mathbb{T}$ of a graph $G$ is defined as $\omega(\mathbb{T}):=\max \{|B_\alpha|-1\mid \alpha\in V(T)\}$. The \emph{treewidth}~$\tw(G)$ of a graph~$G$ is the minimum width over all tree decompositions of~$G$. 
A tree decomposition $\mathbb{T}=(T,(B_\alpha)_{\alpha\in V(T)})$ is a \emph{nice tree decomposition with introduce edge nodes} if the following conditions hold.
\begin{compactenum}[(i)]
\item The tree~$T$ is rooted and binary.
\item For all edges in $E(G)$ there is exactly one \emph{introduce edge node} in $\mathbb{T}$, where an introduce edge node is a node~$\alpha$ in the tree
  decomposition~$\mathbb{T}$ of~$G$ labeled with an edge~$\{v,w\}\in
  E(G)$ with~$v,w\in B_\alpha$ that has exactly one child
  node~$\alpha'$; furthermore~$B_\alpha=B_{\alpha'}$.
\item Each node
  $\alpha\in V(T)$ is of one of the following types:
  \begin{compactitem}
  \item introduce edge node;
  \item \emph{leaf node}: $\alpha$ is a leaf of $T$ and $B_\alpha=\emptyset$;
  \item \emph{introduce vertex node}: $\alpha$ is an inner node of $T$ with
    exactly one child node $\beta\in V(T)$; furthermore~$B_\beta\subseteq
    B_\alpha$ and~$|B_\alpha\backslash B_\beta|=1$;
  \item \emph{forget node}: $\alpha$ is an inner node of $T$ with exactly one
    child node $\beta\in V(T)$; furthermore~$B_\alpha\subseteq B_\beta$
    and~$|B_\beta\backslash B_\alpha|=1$;
  \item \emph{join node}: $\alpha$ is an inner node of $T$ with exactly two
    child nodes $\beta,\gamma\in V(T)$; furthermore~$B_\alpha = B_\beta =
    B_\gamma$.
  \end{compactitem}
\end{compactenum}
A given tree decomposition can be modified in linear time to fulfill the above constraints; moreover, the number of nodes in such a tree decomposition of width~$\omega$ is $O(\omega\cdot n)$~\cite{Kloks94,CyganNPPRW11}.%

\iflong{}\subparagraph*{Flows, Cuts, and Paths.}\fi{}
\ifshort{}\subparagraph*{Cuts and Paths.}\fi{}

Let $G$ be an undirected, connected graph. A \emph{cut} $C\subseteq E$ is a set of edges such that the graph~$G\backslash C$ is not connected. 
Let $s,t\in V(G)$~be two vertices in~$G$. An \emph{$s$-$t$~cut}~$C$ is a cut such that the vertices~$s$ and~$t$ are not connected in~$G\backslash C$. A \emph{minimum~$s$-$t$~cut} is an $s$-$t$~cut~$C$ such that $|C|=\min |C'|$, where the minimum is taken over all $s$-$t$~cuts~$C'$ in~$G$. An $s$-$t$~cut~$C$ in $G$ is \emph{minimal} if for all edges~$e\in C$ it holds that $C\backslash\{e\}$ is not an $s$-$t$~cut in~$G$. 

A \emph{path} is a connected graph with exactly two vertices of degree one and no vertex of degree at least three. We call the vertices with degree one the \emph{endpoints} of the path. The \emph{length} of a path is defined as the number of edges in the path. For two distinct vertices $s,t\in V(G)$, we refer to the path with endpoints $s$ and~$t$ (as subgraph of~$G$) as $s$-$t$~path in~$G$. \iflong{}An $s$-$t$~path in~$G$ is a \emph{shortest $s$-$t$~path} in~$G$ if there is no $s$-$t$~path in~$G$ of smaller length. We denote by~$\dist_G(s,t)$ the length of a shortest $s$-$t$~path in $G$.\fi{}

\iflong{}
A graph~$G$ has \emph{edge capacities} if there is a function $c:E(G)\to \mathbb{R}_{\geq 0}$ that maps each edge in~$G$ to a number in~$\mathbb{R}_{\geq 0}$, where $\mathbb{R}_{\geq 0}$ denotes the non-negative real numbers. For an edge~$e\in E(G)$, we say that $c(e)$ is the \emph{capacity} of edge~$e$ in $G$. We say that graph~$G$ has \emph{unit edge capacities} if $c(e)=1$ for all $e\in E(G)$. In this work, if we consider a graph with edge capacities, then we always consider a graph with unit edge capacities. 

Let $D$ be an directed graph with edge capacities $c:E(D)\to\mathbb{R}_{\geq 0}$ and let $s,t\in V(D)$ be two vertices in~$D$. An \emph{$s$-$t$~flow} in $D$ is a function $f:E(D)\to \{0,1\}$ such that 
\begin{compactenum}[(i)]
\item $f(e)\leq c(e)$ for all $e\in E(D)$,
\item $\sum_{w\in V(D):(v,w)\in E(D)} f((v,w))= \sum_{w\in V(D):(w,v)\in E(D)} f((w,v))$ for all $v\in V(D)\backslash\{s,t\}$, and
\item $\sum_{w\in V(D):(w,t)\in E(D)} f((w,t))-\sum_{w\in V(D):(t,w)\in E(D)} f((t,w)) \geq 0$. 
\end{compactenum}
The \emph{value} of an $s$-$t$~flow~$f$ in~$D$ is defined as $$|f|:=\sum_{w\in V(D):(w,t)\in E(D)} f((w,t))-\sum_{w\in V(D):(t,w)\in E(D)} f((t,w)).$$ An $s$-$t$~flow~$f$ is a \emph{maximum $s$-$t$~flow} in~$D$ if there is no $s$-$t$~flow~$f'$ in~$D$ with $|f'|>|f|$. 

For an undirected graph~$G$ we call the directed graph~$D_G$ the directed version of graph~$G$ if $V(D_G)=V(G)$ and $E(D_G)= \{(u,v),(v,u)\mid \{u,v\}\in E(G)\}$. If $G$ has edge capacities $c:E(G)\to \mathbb{R}_{\geq 0}$, then $D_G$ has edge capacities $c':E(D_G)\to \mathbb{R}_{\geq 0}$ with $c'((u,v)):=c'((v,u)):=c(\{u,v\})$ for all edges $\{u,v\}\in E(G)$. We say that a function $f:E(G)\to\{0,1\}$ is an $s$-$t$~flow with value $|f|:=\sum_{w\in V(G):\{w,t\}\in E(G)} f(\{w,t\})$ in an undirected graph~$G$ with edge capacities $c:E(G)\to \mathbb{R}_{\geq 0}$ and $s,t\in V(G)$, if there is an $s$-$t$~flow~$f'$ in~$D_G$ such that $|f'|=|f|$, and for all edges~$\{u,v\}\in E(G)$ it holds that $f'((u,v))=0$ and $f'((v,u))=f(\{u,v\})$, or $f'((v,u))=0$ and $f'((u,v))=f(\{u,v\})$. An $s$-$t$~flow~$f_1$ is a maximum $s$-$t$~flow in~$G$ if there is no $s$-$t$~flow~$f_2$ in~$G$ with $|f_2|>|f_1|$. We remark that our definition of $s$-$t$~flows is close to the definition given by Goldberg and Rao~\citet{GoldbergR99}. For more information on flows, in particular on integral flows, the max-flow min-cut theorem, and Menger's theorem, we refer to the work of Kleinberg and Tardos~\citet{DBLP:books/daglib/0015106}.
\fi{}

\subparagraph*{Parameterized Complexity.}
A \emph{parameterized problem} is a set of instances $({\cal I}, \ell)$, where ${\cal I} \in \Sigma^*$ for a finite alphabet~$\Sigma$, and $\ell \in \mathbb{N}$ is the \emph{parameter}.
A parameterized problem~$Q$ is \emph{fixed-parameter tractable}, shortly \FPT, if there exists an algorithm that on input
$({\cal I}, \ell)$ decides whether $({\cal I}, \ell)$ is a yes-instance of~$Q$ in $f(\ell)\cdot|{\cal I}|^{O(1)}$ time,
 where $f$ is a computable function independent of~$|{\cal I}|$. 

 \W{t}, $t\geq1$, are classes that (amongst others) contain parameterized problems which presumably do not admit \FPT algorithms. Hardness for \W{t} can be shown by reducing from a \W{t}-hard problem, using a \emph{parameterized} reduction, that is, a many-to-one reduction that runs in \FPT time and maps any instance $(\mathcal{I}, \ell)$ to another instance $(\mathcal{I}', \ell')$ such that $\ell' \leq f(\ell)$ for some computable function~$f$. 

 A parameterized problem $Q$ is \emph{kernelizable}
 if there exists a polynomial-time self-reduction that maps an instance $({\cal I},\ell)$ of
 $Q$ to another instance $({\cal I}',\ell')$ of $Q$ such that: (1) $|{\cal I}'| \leq \lambda(\ell)$ for
 some computable function $\lambda$, (2) $\ell' \leq \lambda(\ell)$, and (3) $({\cal I},\ell)$ is a yes-instance
 of $Q$ if and only if $({\cal I}',\ell')$ is a yes-instance of $Q$. The instance
 $({\cal I}',\ell')$ is called the \emph{problem kernel} of~$({\cal I},\ell)$ and $\lambda$ is called its \emph{size}.

 %


\section{NP-Completeness and W[2]-Hardness With Respect to the Number of Shared Edges}\label{sec:msenphard}

In this section, we show that \msetsc{} is NP-complete and W[2]-hard with respect to the number $k$ of shared edges. 
To this end, we give a parameterized, polynomial reduction from the \textsc{Set Cover} problem. 

\begin{theorem}\label{theorem!kw2hard}
\msetsc{} is NP-complete and W[2]-hard with respect to the number $k$ of shared edges.
\end{theorem}

In the proof of~\Cref{theorem!kw2hard} we provide a reduction from the following problem.

\decprob{Set Cover (SC)}{A set $X$, a set of sets $\mathcal{C}\subseteq 2^X$, and an integer $\ell \in \mathbb{N}_0$.}{Are there sets $C_1,\ldots,C_{\ell'}\in \mathcal{C}$ with $\ell'\leq \ell$ such that $X = \bigcup_{i=1}^{\ell'} C_i$?}
\acrsc{} parameterized by~$\ell$ is well-known to be W[2]-complete~\cite{DowneyF13}. 
Omran et al.~\citet{OmranSZ13} showed that \msetsc{} on directed graphs is NP-hard using a reduction from \sctsc{}. 
Since their reduction is also a parameterized reduction from \acrsc{} parameterized by $\ell$ to \msetsc{} on directed graphs parameterized by the number~$k$ of shared edges, they showed implicitly that \msetsc{} on directed graphs is W[2]-hard with respect to~$k$.
Here, we present a parameterized reduction from \acrsc{} parameterized by~$\ell$ to \msetsc{} (on \emph{undirected} graphs) parameterized by~$k$. In the following, we call a path of length $m\in\mathbb{N}$ an $m$-\emph{chain}.%

\begin{proof}[Proof of \Cref{theorem!kw2hard}]
Let $(X,\mathcal{C},\ell)$ be an instance of \acrsc. Let $\deg(x)$ be the number of sets in~$\mathcal{C}$ containing element $x\in X$, that is, $\deg(x):=|\{C\in \mathcal{C}\mid x\in C\}|$ for every $x\in X$. We construct an instance~$(G,s,t,p,k)$ of \msetsc{} with $p=|\mathcal{C}|+\sum_{x\in X}\deg(x)$ and $k=\ell$ as follows. 

\emph{Construction.} The construction is illustrated in \Cref{fig:reduction}.
 Initially, let $G$ be an empty graph, that is $V(G)=E(G)=\emptyset$. Next, we add the vertices $s$ and $t$ to the vertex set $V(G)$ of graph $G$. Then, we add to $V(G)$ the following vertex sets:
\begin{itemize}
\item $V_X=\{v_i\mid i\in X\}$, the set of vertices corresponding to the elements of $X$,
\item $V_C=\{w_j\mid C_j\in \mathcal{C}\}$, the set of vertices corresponding to the sets in $\mathcal{C}$,
\item $V_D = \{v_{i,j}\mid (i\in X) \wedge (C_j\in \mathcal{C}) \wedge (i\in C_j)\}$, the set of vertices corresponding to the relation of the elements in~$X$ with the sets in~$\mathcal{C}$, i.e.\ a vertex $v_{i,j}$ is in $V_D$ if there is an element $i\in X$ and a set $C_j\in \mathcal{C}$ such that $i\in C_j$, and
\item $V_T=\{t_i \mid i\in X\}$.
\end{itemize}
We connect each $v_{i,j}\in V_D$ via an $(\ell+1)$-chain with~$v_i\in V_X$, with~$t_i\in V_T$ and with~$w_j\in V_C$.
Next, we connect vertex~$s$ with every vertex~$w \in V_C$ via an $(\ell+1)$-chain and with each $v_i\in V_X$ via $\deg(i)$ $(\ell+1)$-chains.
Finally, we connect vertex~$t$ with each $w \in V_C$ via a single edge each and with each $t_i\in V_T$ via $\deg(i)-1$ many $(\ell+1)$-chains.
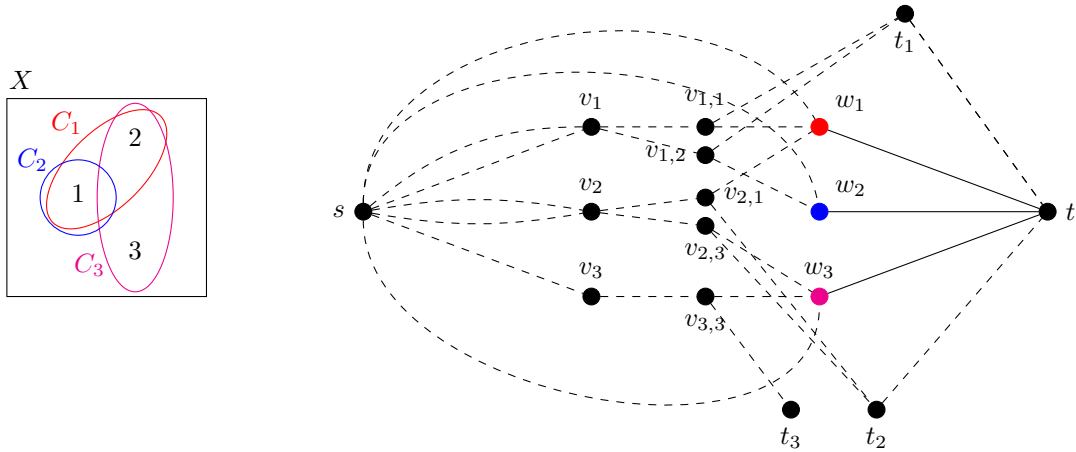
\begin{figure}[!t]
	\centering
\begin{tikzpicture}[x=0.75cm, y=0.75cm]

\node (s) at (0,0) [shape=circle,fill, scale=0.67, label=left:$s$, draw]{};
\node (t) at (12,0) [shape=circle, fill, scale=0.67, label=right:$t$, draw]{};

\node (v_i1) at (4,1.5)  [shape=circle,fill, scale=0.67, label=above:$v_1$,  draw]{};
\node (v_i2) at (4,0)  [shape=circle,fill, scale=0.67, label=above:$v_2$, draw]{};
\node (v_i3) at (4,-1.5)  [shape=circle,fill, scale=0.67, label=above:$v_3$, draw]{};

\node (C_i1) at (8,1.5)  [shape=circle,fill, color=red, scale=0.67, label=45:$w_1$, draw]{};
\node (C_i2) at (8,0)  [shape=circle, fill, color=blue, scale=0.67, label=45:$w_2$, draw]{};
\node (C_i3) at (8,-1.5)  [shape=circle, fill, color=magenta, scale=0.67, label=above:$w_3$, draw]{};

\node (v_{i1,j1}) at (6,1.5)  [shape=circle, fill, scale=0.67, label=above:$v_{1,1}$, draw]{};
\node (v_{i1,j2}) at (6,1)  [shape=circle,  fill, scale=0.67, label=180:$v_{1,2}$, draw]{};

\node (v_{i2,j1}) at (6,0.25)  [shape=circle, fill, scale=0.67, label=right:$v_{2,1}$, draw]{};
\node (v_{i2,j3}) at (6,-0.25)  [shape=circle,  fill, scale=0.67, label=below:$v_{2,3}$, draw]{};

\node (v_{i3,j3}) at (6,-1.5)  [shape=circle,  fill, scale=0.67, label=below:$v_{3,3}$, draw]{};

\node (t_i1) at (9.5,3.5)  [shape=circle, fill, scale=0.67, label=below:$t_1$, draw]{};
\node (t_i2) at (9,-3.5)  [shape=circle,  fill, scale=0.67, label=below:$t_2$, draw]{};
\node (t_i3) at (7.5,-3.5)  [shape=circle,  fill, scale=0.67, label=below:$t_3$, draw]{};

\coordinate [label=1] (1) at (-5,0);
\coordinate [label=2] (2) at (-4,1);
\coordinate [label=3] (3) at (-4,-1);

\draw[color=red] (-4.5,0.75) ellipse [x radius=1cm, y radius=0.5cm, rotate=45]{};
\coordinate[label={[color=red]$C_1$}] (C1) at (-5.2,1.2);
\draw[color=magenta] (-4,0.25) ellipse [x radius=0.5cm, y radius=1.25cm]{};
\coordinate[label={[color=magenta]$C_3$}] (C3) at (-4.8,-1.3);
\draw[color=blue] (-5,0.25) ellipse [x radius=0.5cm, y radius=0.5cm]{};
\coordinate[label={[color=blue]$C_2$}] (C2) at (-5.8,0.55);
\draw (-6.25,-1.5) rectangle (-2.75,2) [label=above:$X$, draw]{};
\coordinate[label={$X$}] (X) at (-6,2);

\draw[-, dashed] (s) to [out=40,in=180] (v_i1);
\draw[-, dashed] (s) to [out=20,in=200] (v_i1);

\draw[-, dashed] (s) to [out=10,in=170] (v_i2);
\draw[-, dashed] (s) to [out=-10,in=190] (v_i2);

\draw[-, dashed] (s) to (v_i3);

\draw[-, dashed] (s) to [out=90,in=120] (C_i1);
\draw[-, dashed] (s) to [out=90,in=100] (C_i2);
\draw[-, dashed] (s) to [out=270,in=270] (C_i3);

\draw[-, dashed] (v_i1) to (v_{i1,j1});
\draw[-, dashed] (v_i1) to (v_{i1,j2});

\draw[-, dashed] (v_i2) to (v_{i2,j1});
\draw[-, dashed] (v_i2) to (v_{i2,j3});

\draw[-, dashed] (v_i3) to (v_{i3,j3});

\draw[-, dashed] (v_{i1,j1}) to (C_i1);
\draw[-, dashed] (v_{i1,j2}) to (C_i2);
\draw[-, dashed] (v_{i1,j1}) to  (t_i1);
\draw[-, dashed] (v_{i1,j2}) to (t_i1);

\draw[-, dashed] (v_{i2,j1}) to (C_i1);
\draw[-, dashed] (v_{i2,j3}) to (C_i3);
\draw[-, dashed] (v_{i2,j1}) to (t_i2);
\draw[-, dashed] (v_{i2,j3}) to (t_i2);

\draw[-, dashed] (v_{i3,j3}) to (C_i3);
\draw[-, dashed] (v_{i3,j3}) to (t_i3);

\draw[-, dashed] (t_i1) to (t);
\draw[-, dashed] (t_i2) to (t);
\draw[-, dashed] (t_i1) to (t);

\draw[-] (C_i1) to (t);
\draw[-] (C_i2) to (t);
\draw[-] (C_i3) to (t);

\end{tikzpicture}
\caption{Illustration of the construction of the graph $G$ (right-hand side) in the reduction from an instance of \sctsc{} on the left-hand side to an instance of \msetsc{}. Dashed lines represent $(\ell+1)$-chains, where $\ell$ is the parameter in the instance of \sctsc{}.}
\label{fig:reduction}
\end{figure}

\emph{Correctness}. We now prove that $(X, \mathcal{C}, \ell)$ is a yes-instance for \acrsc{} if and only of $(G, s, t, p, k)$ is a yes-instance for \msetsc{}.

``$\Leftarrow$'': Suppose that we have~$p$~$s$-$t$~routes in~$G$ that share at most $k$~edges. We show that we can construct a set cover~$\mathcal{C}'\subseteq \mathcal{C}$ for~$X$ with~$|\mathcal{C}'|\leq \ell$. First, we provide some observations.

Since every $(\ell+1)$-chain contains $\ell+1$ edges, every $(\ell+1)$-chain in~$G$ appears in at most one $s$-$t$~route. Since there are $p$~$s$-$t$~routes and there are $p$~$(\ell+1)$-chains incident with vertex~$s$, every $(\ell+1)$-chain incident with vertex~$s$ appears in exactly one $s$-$t$~route. Therefore, each $v_i\in V_X$ appears in at least $\deg(i)$ $s$-$t$~routes and each $w_j\in V_C$ appears in at least one $s$-$t$~route. Moreover, since each~$v_i\in V_X$ is incident with $2\cdot\deg(i)$~$(\ell+1)$-chains, each~$v_i\in V_X$ appears in exactly $\deg(i)$~$s$-$t$~routes.
Each $v_{i,j}\in V_D$ has exactly degree three and is incident with three $(\ell+1)$-chains. Therefore, every $v_{i,j}\in V_D$ appears in at most one $s$-$t$~route. Moreover, since each $v_i\in V_X$ appears in $\deg(i)$~$s$-$t$~routes, and there are $\deg(i)$ vertices in~$V_D$ each connected with~$v_i$ via an~$(\ell+1)$-chain, each~$v_{i,j}\in V_D$ appears in exactly one $s$-$t$~route.
Let $V' := \{w\in V_C|\text{ $\{w,t\}$ is a shared edge}\}$, that is, $V'$~is the set of vertices in~$V_C$ that are incident with the shared edges of the $p$~$s$-$t$~routes. We claim that, if~$w_j\in V_C$ appears in an $s$\nobreakdash-$t$~route~$P$ containing a vertex in~$V_X$, then $w_j\in V'$. Let $v_i\in V_X$ be the vertex that appears in route~$P$. Suppose $V'$ does not contain vertex~$w_j$ and, thus, edge~$\{w_j,t\}$ is not shared. Since the $(\ell+1)$-chain connecting vertex~$s$ with vertex~$w_j$ appears in exactly one $s$\nobreakdash-$t$~route different from~$P$, vertex~$w_j$ appears in at least two $s$-$t$~routes. Since every vertex in~$V_C$ is incident with vertex~$t$ and vertices in~$V_D$ via $(\ell+1)$-chains, there is a vertex~$v_{i',j'}\in V_D$ different from vertex~$v_{i,j}$, such that one of the $s$-$t$~routes containing vertex~$w_j$ contains vertex~$v_{i',j'}$. Let $P'$ be the route containing the vertices~$w_j$ and~$v_{i',j'}$. We know that there is an $s$-$t$~route containing vertex~$v_{i',j'}$ and vertex~$v_{i'}$ different from~$P'$, since there are $p$~$s$-$t$~routes. Thus, vertex~$v_{i',j'}$ appears in at least two $s$-$t$~routes, contradicting the fact that each vertex in $V_D$ appears in exactly one $s$-$t$~route. We conclude that set~$V'$ contains vertex~$w_j$.
We claim that the subset $\mathcal{C}'\subseteq \mathcal{C}$ corresponding to vertices in $V'$, that is $\mathcal{C}':=\{C_j\in \mathcal{C}\mid w_j\in V'\}$, is a set cover of $X$ of size at most $\ell$. Each $t_i\in V_T$ is connected with vertex~$t$ via $\deg(i)-1$ $(\ell+1)$-chains, and connected with $\deg(i)$ vertices in $V_D$. Therefore, for each $i\in X$, there exists at least one $j\in[|\mathcal{C}|]$, such that $v_i$, $v_{i,j}$, and $w_j$ appear in an $s$-$t$~route. As shown before, it follows that~$w_j\in V'$. Thus, for each element~$i\in X$ there exists a set~$C_j\in \mathcal{C}'$ such that $i\in C_j$, and hence, $\mathcal{C'}$ is a set cover of~$X$ of size at most~$\ell$.

``$\Rightarrow$'': Suppose that we have a set~$\mathcal{C}'\subseteq \mathcal{C}$ with $|\mathcal{C}'|\leq \ell$, such that $\mathcal{C}'$~is a set cover of~$X$. We show that we can construct $p$~$s$-$t$~routes in~$G$ that share at most~$k=\ell$ edges.

First, we construct $|\mathcal{C}|$~$s$-$t$ routes in the following way. For each vertex $w\in V_C$, we construct the $s$-$t$~route containing only the $(\ell+1)$-chain connecting $s$ and $w$ and the edge $\{w,t\}$. It follows that each of the $|\mathcal{C}|$~edges connecting a vertex in~$V_C$ with vertex~$t$ appears in exactly one $s$-$t$~route. 

Next, we construct $|X|$ $s$-$t$~routes in the following way. We remark that since $\mathcal{C}'$ is a set cover of $X$, for each $i\in X$ there exists a $C_j\in\mathcal{C}'$ such that $i\in C_j$. For each $v_i\in V_X$, we construct an $s$-$t$~route containing only the $(\ell+1)$-chains connecting $s$ with $v_i$, $v_i$ with $v_{i,j}$, $v_{i,j}$ with $w_j$, and the edge~$\{w_j,t\}$, where vertex~$w_j\in V_C$ corresponds to a $C_j\in \mathcal{C'}$ with $i\in C_j$. Since $|\mathcal{C}'|\leq \ell$, there are at most $\ell$~edges connecting the vertices in~$V_C$ with~$t$ that are shared by the $s$-$t$~routes constructed so far.

Finally, we construct $\sum_{x\in X}\deg(x) - |X|$ $s$-$t$~routes in the following way. Note that for each $v_i\in V_X$, there are $\deg(i)-1$ $(\ell+1)$-chains connecting $s$ and $v_i$ not covered by an $s$-$t$~route and there are $\deg(i)-1$ vertices in $V_D$ connected with $v_i$ via an $(\ell+1)$-chain not covered by an $s$-$t$~route. Moreover, no vertex in $V_T$ is covered by an $s$-$t$~route, and thus, $t_i\in V_T$ is not covered by an $s$-$t$~route. Recall that $t_i\in V_T$ is connected with vertex~$t$ by $\deg(i)-1$ $(\ell+1)$-chains. Thus, for each $v_i\in V_X$, we can lead $\deg(i)-1$ $s$-$t$~routes from $s$ over $v_i$, vertices in~$V_D$ and~$t_i$ to~$t$ without sharing any edge.

In total, we constructed 
\[ |\mathcal{C}|+|X|+\sum_{x\in X}\deg(x) - |X| = |\mathcal{C}|+\sum_{x\in X}\deg(x) = p \]
$s$-$t$~routes sharing at most $k=\ell$ edges. 

Note that the reduction is a polynomial reduction and a parameterized reduction since $k=\ell$. 
Moreover, given $p$ $s$-$t$~routes in a graph~$G$ with~$s,t\in V(G)$, one can check in polynomial time whether the routes share at most~$k$~edges.
Hence, \msetsc{} is NP-complete and W[2]-hard with respect to the number~$k$ of shared edges.
\end{proof}
%


\section{An Algorithm for Small Treewidth and Small Number of Routes}\label{sec:dp}

In this section we show the following theorem.
\begin{theorem}\label{thm:twdp}
Let~$G$ be a graph with $s,t\in V(G)$ given together with a tree decomposition of width $\omega$. %
Let~$p\in \mathbb{N}$ be an integer. Then the minimum number of shared edges in a $(p, s, t)$-routing %
can be computed in~${O(p\cdot (\omega+4)^{3\cdot p\cdot (\omega+3)+4}\cdot n)}$~time.
\end{theorem}

The proof is based on a dynamic program that computes a table for each node of the (arbitrarily rooted) tree decomposition in a bottom-up fashion. For our application, it is convenient to use a nice tree decomposition with introduce edge nodes such that each bag contains the sink and the source node.
For each node $\alpha$ in the tree decomposition $\mathbb{T}$ of $G$, we define $V_\alpha$ as the set of vertices and $E_\alpha$ as the set of edges that are introduced in the subtree rooted at node $\alpha$. In other words, a vertex~$v\in V(G)$ is in~$V_\alpha$ if and only if there exists at least one introduce vertex node in the subtree rooted at node~$\alpha$ that introduced vertex~$v$. As a special case, since the vertices~$s$ and~$t$ are contained in every bag, we consider~$s$ and~$t$ as introduced by each leaf node. An edge~$e\in E(G)$ is in~$E_\alpha$ if and only if there exists an introduce edge node in the subtree rooted at node~$\alpha$ that introduced edge~$e$. Recall that there is a unique introduce edge node for every edge of graph~$G$. We define $G_\alpha:=(V_\alpha, E_\alpha)$ as the graph for node~$\alpha$. For every leaf node~$\alpha$ in~$\mathbb{T}$, we set $V_\alpha=\{s,t\}$ and~$E_\alpha=\emptyset$. 

\iflong{}
\begin{figure}[t]
\centering

\begin{tikzpicture}[x=0.875cm, y=0.875cm]

\def\x{7.5};
\def\xe{5};
\def\y{3.5};

\def\ysc{0.75};
\def\xsc{0.75};

\node (s) at (0-1*\x,-1*\ysc)[shape=circle, label=180:{$s$}, draw]{};
\node (a) at (1*\xsc-1*\x,0*\ysc)[shape=circle, label=90:{$a$}, draw]{};
\node (b) at (3*\xsc-1*\x,0*\ysc)[shape=circle, label=90:{$b$}, draw]{};
\node (c) at (1*\xsc-1*\x,-2*\ysc)[shape=circle, label=270:{$c$}, draw]{};
\node (d) at (2*\xsc-1*\x,-2*\ysc)[shape=circle, label=270:{$d$}, draw]{};
\node (e) at (3*\xsc-1*\x,-2*\ysc)[shape=circle, label=270:{$e$}, draw]{};
\node (t) at (4*\xsc-1*\x,-1*\ysc)[shape=circle, label=0:{$t$}, draw]{};

\draw (s) -- (a);
\draw (s) -- (c);
\draw (a) -- (c);
\draw (a) -- (b);
\draw (c) -- (b);
\draw (c) -- (d);
\draw (d) -- (e);
\draw (e) -- (t);
\draw (b) -- (t);

\node (G) at (0-0.3-1*\x,0)[]{$G=G_\tau$};

\node (s) at (0+0.6*\x,-1*\ysc-0.35*\y)[shape=circle, label=180:{$s$}, draw]{};
\node (a) at (1*\xsc+0.6*\x,0*\ysc-0.35*\y)[shape=circle, label=90:{$a$}, draw]{};
\node (b) at (3*\xsc+0.6*\x,0*\ysc-0.35*\y)[shape=circle, label=90:{$b$}, draw]{};
\node (c) at (1*\xsc+0.6*\x,-2*\ysc-0.35*\y)[shape=circle, label=270:{$c$}, draw]{};
\node (d) at (2*\xsc+0.6*\x,-2*\ysc-0.35*\y)[shape=circle, label=270:{$d$}, draw]{};
\node (e) at (3*\xsc+0.6*\x,-2*\ysc-0.35*\y)[shape=circle, label=270:{$e$}, draw]{};
\node (t) at (4*\xsc+0.6*\x,-1*\ysc-0.35*\y)[shape=circle, label=0:{$t$}, draw]{};

\draw (s) -- (a);
\draw (s) -- (c);
\draw (a) -- (c);
\draw (a) -- (b);
\draw[color=lightgray] (c) -- (b);
\draw[color=lightgray] (c) -- (d);
\draw (d) -- (e);
\draw (e) -- (t);
\draw (b) -- (t);

\node (G) at (0-0.2+0.6*\x,0-0.3*\y)[]{$G_\alpha$};

\node (s) at (0-1*\x,-1*\ysc-1.1*\y)[shape=circle, label=180:{$s$}, draw]{};
\node (a) at (1*\xsc-1*\x,0-1.1*\y)[shape=circle, label=90:{$a$}, draw]{};
\node (b) at (3*\xsc-1*\x,0-1.1*\y)[shape=circle, label=90:{$b$}, draw]{};
\node (c) at (1*\xsc-1*\x,-2*\ysc-1.1*\y)[shape=circle, label=270:{$c$}, draw]{};
\node (d) at (2*\xsc-1*\x,-2*\ysc-1.1*\y)[shape=circle, label=270:{$d$}, color=lightgray, draw]{};
\node (e) at (3*\xsc-1*\x,-2*\ysc-1.1*\y)[shape=circle, label=270:{$e$}, color=lightgray, draw]{};
\node (t) at (4*\xsc-1*\x,-1*\ysc-1.1*\y)[shape=circle, label=0:{$t$}, draw]{};

\draw (s) -- (a);
\draw (s) -- (c);
\draw (a) -- (c);
\draw (a) -- (b);
\draw[color=lightgray] (c) -- (b);
\draw[color=lightgray] (c) -- (d);
\draw[color=lightgray] (d) -- (e);
\draw[color=lightgray] (e) -- (t);
\draw[color=lightgray] (b) -- (t);

\node (G) at (0-0.2-\x,0-1.05*\y)[]{$G_\beta$};

\node (s) at (0+0.6*\x,-1*\ysc-1.5*\y)[shape=circle, label=180:{$s$}, draw]{};
\node (a) at (1*\xsc+0.6*\x,0-1.5*\y)[shape=circle, label=90:{$a$}, color=lightgray, draw]{};
\node (b) at (3*\xsc+0.6*\x,0-1.5*\y)[shape=circle, label=90:{$b$}, draw]{};
\node (c) at (1*\xsc+0.6*\x,-2*\ysc-1.5*\y)[shape=circle, label=270:{$c$}, color=lightgray, draw]{};
\node (d) at (2*\xsc+0.6*\x,-2*\ysc-1.5*\y)[shape=circle, label=270:{$d$}, draw]{};
\node (e) at (3*\xsc+0.6*\x,-2*\ysc-1.5*\y)[shape=circle, label=270:{$e$}, draw]{};
\node (t) at (4*\xsc+0.6*\x,-1*\ysc-1.5*\y)[shape=circle, label=0:{$t$}, draw]{};

\draw[color=lightgray] (s) -- (a);
\draw[color=lightgray] (s) -- (c);
\draw[color=lightgray] (a) -- (c);
\draw[color=lightgray] (a) -- (b);
\draw[color=lightgray] (c) -- (b);
\draw[color=lightgray] (c) -- (d);
\draw[color=lightgray] (d) -- (e);
\draw (e) -- (t);
\draw (b) -- (t);

\node (G) at (0-0.2+0.6*\x,0-1.45*\y)[]{$G_\gamma$};

\node (s) at (0+0.6*\x,-1*\ysc-2.8*\y)[shape=circle, label=180:{$s$}, draw]{};
\node (a) at (1*\xsc+0.6*\x,0-2.8*\y)[shape=circle, label=90:{$a$}, color=lightgray, draw]{};
\node (b) at (3*\xsc+0.6*\x,0-2.8*\y)[shape=circle, label=90:{$b$}, draw]{};
\node (c) at (1*\xsc+0.6*\x,-2*\ysc-2.8*\y)[shape=circle, label=270:{$c$}, color=lightgray, draw]{};
\node (d) at (2*\xsc+0.6*\x,-2*\ysc-2.8*\y)[shape=circle, label=270:{$d$}, color=lightgray, draw]{};
\node (e) at (3*\xsc+0.6*\x,-2*\ysc-2.8*\y)[shape=circle, label=270:{$e$}, draw]{};
\node (t) at (4*\xsc+0.6*\x,-1*\ysc-2.8*\y)[shape=circle, label=0:{$t$}, draw]{};

\draw[color=lightgray] (s) -- (a);
\draw[color=lightgray] (s) -- (c);
\draw[color=lightgray] (a) -- (c);
\draw[color=lightgray] (a) -- (b);
\draw[color=lightgray] (c) -- (b);
\draw[color=lightgray] (c) -- (d);
\draw[color=lightgray] (d) -- (e);
\draw (e) -- (t);
\draw[color=lightgray] (b) -- (t);

\node (G) at (0-0.2+0.6*\x,0-2.75*\y)[]{$G_\delta$};

\node (s) at (0-1*\x,-1*\ysc-3*\y)[shape=circle, label=180:{$s$}, draw]{};
\node (a) at (1*\xsc-1*\x,0-3*\y)[shape=circle, label=90:{$a$}, color=lightgray, draw]{};
\node (b) at (3*\xsc-1*\x,0-3*\y)[shape=circle, label=90:{$b$}, color=lightgray, draw]{};
\node (c) at (1*\xsc-1*\x,-2*\ysc-3*\y)[shape=circle, label=270:{$c$}, color=lightgray, draw]{};
\node (d) at (2*\xsc-1*\x,-2*\ysc-3*\y)[shape=circle, label=270:{$d$}, color=lightgray, draw]{};
\node (e) at (3*\xsc-1*\x,-2*\ysc-3*\y)[shape=circle, label=270:{$e$}, color=lightgray, draw]{};
\node (t) at (4*\xsc-1*\x,-1*\ysc-3*\y)[shape=circle, label=0:{$t$}, draw]{};

\draw[color=lightgray] (s) -- (a);
\draw[color=lightgray] (s) -- (c);
\draw[color=lightgray] (a) -- (c);
\draw[color=lightgray] (a) -- (b);
\draw[color=lightgray] (c) -- (b);
\draw[color=lightgray] (c) -- (d);
\draw[color=lightgray] (d) -- (e);
\draw[color=lightgray] (e) -- (t);
\draw[color=lightgray] (b) -- (t);

\node (G) at (0-0.2 -1*\x,0-2.95*\y)[]{$G_\eta$};

\def\tdcxsc{1.2};

\node (bcd1) at (0,0)[shape=rectangle, label=0:{$\{b,c\}$}, label=180:{$(\tau)$}, draw]{$s,b,c,d,t$};
\node (bcd2) at (0,-1)[shape=rectangle, label=180:{$\{c,d\}$}, draw]{$s,b,c,d,t$};
\node (bcd3) at (0,-2)[shape=rectangle, label=0:{$(\alpha)$}, draw]{$s,b,c,d,t$};
\node (bcd3l) at (-1*\tdcxsc,-3)[shape=rectangle, draw]{$s,b,c,d,t$};
\node (bcd3r) at (1*\tdcxsc,-3)[shape=rectangle, draw]{$s,b,c,d,t$};

\node (bc) at (-1*\tdcxsc,-4)[shape=rectangle, label=180:{$(\beta)$}, draw]{$s,b,c,t$};
\node (abc1) at (-1*\tdcxsc,-5)[shape=rectangle, label=180:{$\{a,b\}$}, draw]{$s,a,b,c,t$};
\node (abc2) at (-1*\tdcxsc,-6)[shape=rectangle, label=180:{$\{a,c\}$}, draw]{$s,a,b,c,t$};
\node (abc3) at (-1*\tdcxsc,-7)[shape=rectangle, draw]{$s,a,b,c,t$};
\node (sac1) at (-1*\tdcxsc,-8)[shape=rectangle, label=180:{$\{s,c\}$}, draw]{$s,a,c,t$};
\node (sac2) at (-1*\tdcxsc,-9)[shape=rectangle, label=180:{$\{s,a\}$}, draw]{$s,a,c,t$};
\node (sac3) at (-1*\tdcxsc,-10)[shape=rectangle, draw]{$s,a,c,t$};
\node (sa) at (-1*\tdcxsc,-11)[shape=rectangle, draw]{$s,a,t$};
\node (s) at (-1*\tdcxsc,-12)[shape=rectangle, label=180:{$(\eta)$}, draw]{$s,t$};

\node (bd) at (1*\tdcxsc,-4)[shape=rectangle, draw]{$s,b,d,t$};
\node (bde1) at (1*\tdcxsc,-5)[shape=rectangle, label=0:{$\{d,e\}$}, draw]{$s,b,d,e,t$};
\node (bde2) at (1*\tdcxsc,-6)[shape=rectangle, label=0:{$(\gamma)$}, draw]{$s,b,d,e,t$};
\node (bet1) at (1*\tdcxsc,-8)[shape=rectangle, label=0:{$\{b,t\}$}, draw]{$s,b,e,t$};
\node (bet2) at (1*\tdcxsc,-9)[shape=rectangle, label=0:{$\{e,t\}, (\delta)$},  draw]{$s,b,e,t$};
\node (bet3) at (1*\tdcxsc,-10)[shape=rectangle, draw]{$s,b,e,t$};
\node (et) at (1*\tdcxsc,-11)[shape=rectangle, draw]{$s,e,t$};
\node (t) at (1*\tdcxsc,-12)[shape=rectangle, draw]{$s,t$};

\draw (bcd1) -- (bcd2);
\draw (bcd2) -- (bcd3);
\draw (bcd3) -- (bcd3l);
\draw (bcd3) -- (bcd3r);

\draw (bcd3l) -- (bc);
\draw (bc) -- (abc1);
\draw (abc1) -- (abc2);
\draw (abc2) -- (abc3);
\draw (abc3) -- (sac1);
\draw (sac1) -- (sac2);
\draw (sac2) -- (sac3);
\draw (sac3) -- (sa);
\draw (sa) -- (s);

\draw (bcd3r) -- (bd);
\draw (bd) -- (bde1);
\draw (bde1) -- (bde2);
\draw (bde2) -- (bet1);
\draw (bet1) -- (bet2);
\draw (bet2) -- (bet3);
\draw (bet3) -- (et);
\draw (et) -- (t);

\end{tikzpicture}
\caption{Example for a nice tree decomposition with introduce edge nodes and vertices~$s$ and $t$ contained in every bag on an example graph $G$ (top-left). Node~$\tau$ is the root node in the tree decomposition. The graphs around the tree decomposition correspond to the graphs for the tagged nodes in the tree decomposition. For example, graph $G_\gamma$ corresponds to the graph for node $\gamma$ in the tree decomposition.}
\label{fig:modifiedTDC}
\end{figure}
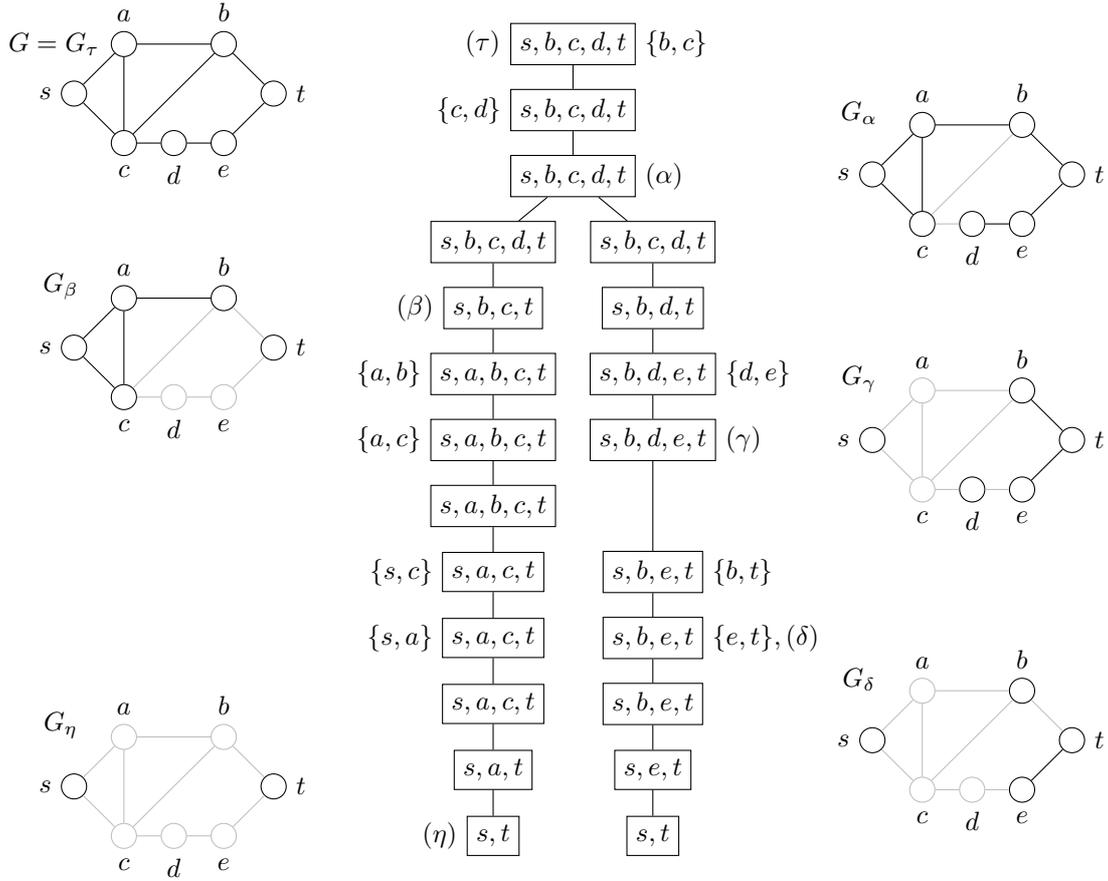
In \Cref{fig:modifiedTDC}, we show for an example graph~$G$ (upper-left) with $s,t\in V(G)$ a nice tree decomposition with introduce edge nodes and vertices~$s$ and~$t$ contained in each bag. Moreover, we illustrate the graphs as defined above for some tagged nodes in the tree decomposition. In the center of the figure, the modified nice tree decomposition with introduce edge nodes is shown. The graphs around the tree decomposition are the graphs for some tagged nodes, for example, the graph~$G_\alpha$ is the graph for node~$\alpha$ in the tree decomposition. For following examples and illustrations, we make use of this example throughout this section, and thus, we denote by $\mathbb{T}^*$ the tree decomposition in \Cref{fig:modifiedTDC}.
\fi

\subparagraph*{Partial Solutions.}
\todo{trees vs. conn. subgraphs}
\todo[inline]{rn: This part is quite try and lacks a (top-down) motivation.}
We define a set of $p$ forests in $G_\alpha$ as a \emph{partial solution} $L_\alpha$ for node $\alpha$. %
Instead of asking for $p$ $s$-$t$ routes that share at most $k$~edges, we can ask for $p$ $s$-$t$~forests that share at most $k$ edges, where an $s$-$t$~forest is a forest that contains at least one tree connecting vertices~$s$ and~$t$. Note that every forest that contains a tree containing both vertices~$s$ and~$t$ can be ``reduced'' to an $s$-$t$~path. A partial solution~$L_\alpha$ has a cost value~$c(L_\alpha)$, which is the number of edges in~$G_\alpha$ that appear in at least two of the~$p$~forests in~$L_\alpha$.

In order to represent the intersection of the trees in a partial solution with the bag that we are currently considering, we use the following notation. For each node $\alpha$ in the tree decomposition $\mathbb{T}$ of $G$, we consider $p$-tuples of pairs $\mathcal{X}^\alpha := (\mathcal{Y}_q^\alpha, Z_q^\alpha)_{q=1,\ldots,p}$, where for each $q\in[p]$, $Z^\alpha_q\subseteq B_\alpha$ together with $\mathcal{Y}^\alpha_q\subseteq 2^{B_\alpha}$ is a partition of~$B_\alpha$, that is,%
\begin{inparaenum}[(i)]
\item $\bigcup_{M\in\mathcal{Y}^\alpha_q} M \cup Z^\alpha_q = B_\alpha$, 
\item for all $X,Y\in \mathcal{Y}^\alpha_q\cup \{Z_q^\alpha\} $ with $X\neq Y$ it holds $X\cap Y=\emptyset$.
\end{inparaenum}
We say that $\mathcal{X}^\alpha$ is a \emph{signature} for node~$\alpha$. For each $q\in[p]$, we call the pair~$(\mathcal{Y}^\alpha_q,Z_q^\alpha)$ a \emph{segmentation} of the vertex set $B_\alpha$. We write segmentation~$q$ instead of segmentation with index~$q$ for short. We call each~$M\in \mathcal{Y}^\alpha_q$~a \emph{segment} of the segmentation~$q$ and we call $Z_q^\alpha$ the \emph{zero}-segment of the segmentation~$q$. 
\looseness=-1 To connect signatures (and segmentations) with the partial solutions that they represent, we use the following notation.  We say that the signature $\mathcal{X}^\alpha$ is a \emph{valid} signature for node~$\alpha$ if there is a partial solution~$L_\alpha$ for node~$\alpha$ such that for each $q\in[p]$, the zero-segment $Z_q^\alpha$ is the set of nodes in $B_\alpha$ that do not appear in the forest with index $q$ and for each set $M\in\mathcal{Y}_q^\alpha$, there is a tree $S$ in the forest with index $q$ such that $M = B_\alpha\cap V(S)$. In other words, the sets in~$\mathcal{Y}^\alpha_q$ correspond to connected components in the forest with index~$q$ of the partial solution. We say that $\mathcal{X}^\alpha$ is a signature \emph{induced} by the partial solution~$L_\alpha$ if $\mathcal{X}^\alpha$ is a valid signature for node~$\alpha$ and the partial solution~$L_\alpha$ validates $\mathcal{X}^\alpha$. In this case, for each $q\in[p]$, the pair $(\mathcal{Y}^\alpha_q,Z^\alpha_q)$ is an \emph{induced} segmentation. We remark that given~$\mathcal{X}^\alpha$, there can be exactly one, more than one or no partial solution with signature~$\mathcal{X}^\alpha$. 
\todo[inline]{This is a strange way to enumerate all possibilities...}
Given a partial solution~$L_\alpha$ for~$G_\alpha$, there is exactly one signature induced by~$L_\alpha$. Let $\mathcal{X}^\alpha$ be a signature for node~$\alpha$ such that there is no partial solution for $G_\alpha$ that induces the signature~$\mathcal{X}^\alpha$, then we say that~$\mathcal{X}^\alpha$ is an \emph{invalid} signature.

\iflong{}Given~$\mathcal{X}^\alpha$ and~$B\subseteq B_\alpha$, we define $\mathcal{X}^\alpha|_{B}$ as the signature $\mathcal{X}^\alpha$ with sets restricted to the set $B$, that is, $Z^\alpha_q\cap B$ and $M^\alpha_q\cap B$ for all $M^\alpha_q\in \mathcal{Y}^\alpha_q$ and for all $q\in[p]$.\fi{}
Let $\mathbb{T}=(T_\mathbb{T}, (B_\alpha)_{\alpha\in V(T_\mathbb{T})})$ be a nice tree decomposition of $G$ with introduce edge nodes and vertices~$s$ and~$t$ contained in every bag. Let $\omega:=\omega(\mathbb{T})\leq \tw(G) + 2$ be the width of~$\mathbb{T}$. We consider the table~$T$ in the following dynamic program that we process bottom-up on the tree decomposition~$\mathbb{T}$, that is, we start to fill the entries of the table~$T$ at the leaf nodes of the tree decomposition~$\mathbb{T}$ and we traverse the tree of the tree decomposition from the leaves to the root. For a node~$\alpha$ in the tree decomposition~$\mathbb{T}$ and a signature~$\mathcal{X}^\alpha$ for node~$\alpha$, the entry~$T[\alpha, \mathcal{X}^\alpha]$ is defined as
\[ T[\alpha,\mathcal{X}^\alpha] := \begin{cases}
\min c(L_\alpha) ,& \text{if $\mathcal{X}^\alpha$ is a valid signature,} \\
\infty ,& \text{otherwise,}
\end{cases}
\]
where the minimum is taken over all partial solutions~$L_\alpha$ in $G_\alpha$ such that~$L_\alpha$ induces the signature~$\mathcal{X}^\alpha$.
For each type of node in $\mathbb{T}$, we define a rule on how to fill each entry in~$T$, prove the correctness of each rule, and discuss the running time for applying the rule and the running time for filling all entries in $T$ for the given type of node.\ifshort{} Due to space constraints, we give some details only for introduce edge nodes, and defer the correctness proof and the remaining nodes to the full version of the paper.\fi\iflong{} We start with the leaf nodes of the tree decomposition~$\mathbb{T}$.
\subparagraph*{Leaf Node.}
Let $\alpha$ be a leaf node of $\mathbb{T}$. Since $s$ and $t$ appear in every bag of $\mathbb{T}$, it holds that $B_\alpha=\{s,t\}$. We set
\[ T[\alpha,\mathcal{X}^\alpha] := \begin{cases}
	0 ,& \text{if $\mathcal{Y}^\alpha_q=\{\{s\},\{t\}\}$ for all $q=1,\ldots, p$,} \\
	\infty ,& \text{otherwise.}
	\end{cases}
\] 

We recall that $V_\alpha=\{s,t\}$ and $E_\alpha=\emptyset$ for every leaf node~$\alpha$ in~$\mathbb{T}$. Since there is no edge in $E_\alpha$, the vertices~$s$ and~$t$ cannot appear together in one tree in a forest in any partial solution in $G_\alpha$, and thus, the vertices~$s$ and $t$ cannot appear together in one segment in any segmentation of a signature for a leaf node. Since in any solution to our problem, $s$ and $t$ appear in each of the $p$ forests, we can set $s$ and $t$ as segments of all $p$~segmentations.

\subparagraph*{Introduce Vertex Node.}
Let $\alpha$ be an introduce vertex node of $\mathbb{T}$ and let $\beta$ be the child node of~$\alpha$ with $B_\alpha\backslash B_\beta = \{v\}$. Two signatures $\mathcal{X}^\alpha$ and $\mathcal{X}^\beta$ are \emph{compatible} if $\mathcal{X}^\alpha|_{B_\beta}=\mathcal{X}^\beta$, and $v\in Z^\alpha_q$ or $\{v\}\in \mathcal{Y}^\alpha_q$ for each $q\in[p]$. We claim that 
\[
T[\alpha,\mathcal{X}^\alpha] = \begin{cases} 
\min_{\text{$\mathcal{X}^\beta$ compatible with $\mathcal{X}^\alpha$}} T[\beta,\mathcal{X}^\beta] ,& \text{if it exists $\mathcal{X}^\beta$ compatible with $\mathcal{X}^\alpha$,} \\
\infty ,& \text{otherwise.}
\end{cases}
\]
Since $\alpha$ is an introduce vertex node for vertex $v\in V(G)$, no edge incident with~$v$ is introduced in any node in the subtree rooted at node $\alpha$, and thus, vertex~$v$ is an isolated vertex in~$G_\alpha$. As a consequence, in every forest in all partial solutions for $G_\alpha$, the introduced vertex $v$ is either a single-vertex tree or does not appear in the forest since $v$ cannot be connected to any vertex in $G_\alpha$. A single-vertex tree is a tree that contains exactly one vertex and does not contain any edge.
\vspace{3pt}\emph{Correctness}. ``$\geq$'': Let $L_\alpha$ be a partial solution for~$G_\alpha$ with signature~$\mathcal{X}^\alpha$ such that $T[\alpha,\mathcal{X}^\alpha]=c(L_\alpha)$. We construct a partial solution~$L_\beta$ for $G_\beta$ and a signature $\mathcal{X}^\beta$ such that $L_\beta$ induces $\mathcal{X}^\beta$ and $\mathcal{X}^\beta$ is compatible with $\mathcal{X}^\alpha$. For each forest in the partial solution~$L_\alpha$, vertex~$v$ is either a single-vertex tree or does not appear in the forest, since there is no edge incident with vertex~$v$ in~$G_\alpha$. If vertex~$v$ appears as single-vertex tree in any forest in~$L_\alpha$, deleting $v$ yields a forest in $G_\beta$. We define the partial solution~$L_\beta$ as~$L_\alpha$ restricted to~$V_\beta$, which are the forests without the isolated vertex $v$. The partial solution~$L_\beta$ is a partial solution for~$G_\beta$ with valid signature~$\mathcal{X}^\beta:=\mathcal{X}^\alpha|_{B_\beta}$. Signature~$\mathcal{X}^\beta$ is compatible with signature~$\mathcal{X}^\alpha$. It follows that
\[ T[\alpha,\mathcal{X}^\alpha] = c(L_\alpha) = c(L_\beta) \geq T[\beta,\mathcal{X}^\beta] \geq  \min_{\text{$\mathcal{X'}^{\beta}$ compatible with $\mathcal{X}^\alpha$}} T[\beta,\mathcal{X'}^{\beta}] . \]

``$\leq$'': Let $L_\beta$ be a partial solution for~$G_\beta$ with signature~$\mathcal{X}^\beta$ compatible with signature~$\mathcal{X}^\alpha$ such that $T[\beta,\mathcal{X}^\beta]=c(L_\beta)$ and $T[\beta,\mathcal{X}^\beta]= 
\min_{\text{$\mathcal{X'}^{\beta}$ compatible with $\mathcal{X}^\alpha$}} T[\beta,\mathcal{X'}^{\beta}]$. We construct a partial solution $L_\alpha$ for $G_\alpha$ with signature $\mathcal{X}^\alpha$. For each $q\in[p]$, if $v\in Z^\alpha_q$, then we do not add $v$ to the forest with index $q$ in $L_\beta$. If $v$ is a single segment in the segmentation $q$, i.e.\ $\{v\}\in \mathcal{Y}^\alpha_q$, then we add $v$ as a single-vertex tree to the forest with index $q$ in $L_\beta$. Since $L_\beta$ is a partial solution for $G_\beta$, the constructed $L_\alpha$ is a partial solution for~$G_\alpha$ with signature $\mathcal{X}^\alpha$. It follows that
\[ \min_{\text{$\mathcal{X'}^{\beta}$ compatible with $\mathcal{X}^\alpha$}} T[\beta,\mathcal{X'}^{\beta}] = T[\beta,\mathcal{X}^\beta] = c(L_\beta) = c(L_\alpha) \geq T[\alpha,\mathcal{X}^\alpha] . \]

\vspace{3pt}\emph{Running time}. For each signature $\mathcal{X}^\alpha$, we check for all $q\in[p]$ whether $v\in Z^\alpha_q$ or $\{v\}\in \mathcal{Y}^\alpha_q$ in~$O(p\cdot |B_\alpha|)$~time. If for all $q\in[p]$ holds that $v\in Z^\alpha_q$ or $\{v\}\in \mathcal{Y}^\alpha_q$, then we check all signatures~$\mathcal{X}^\beta$ for node~$\beta$ for compatibility with signature $\mathcal{X}^\alpha$, that means, we check if $\mathcal{X}^\alpha|_{B_\beta}=\mathcal{X}^\beta$. This can be done in~$O(p\cdot |B_\alpha|^2)$~time. Since there are $O((|B_\beta|+1)^{p\cdot |B_\beta|})$~signatures for node~$\beta$ and $|B_\beta|\leq |B_\alpha|$, the running time for this step is in $O(p\cdot (|B_\alpha|+1)^{p\cdot |B_\alpha|+2})$. Since there are $O((|B_\alpha|+1)^{p\cdot |B_\alpha|})$~signatures for node~$\alpha$ and $|B_\beta|\leq |B_\alpha|\leq \omega+1$, the overall running time for filling the entries in $T$ for an introduce vertex node is in~$O(p\cdot (\omega+2)^{2\cdot p\cdot (\omega+1)+2})$.

\subparagraph*{Forget Node.}
Let $\alpha$ be a forget node of $\mathbb{T}$ and let $\beta$ be the child node of~$\alpha$ with~$B_\beta\backslash B_\alpha = \{v\}$. Two signatures $\mathcal{X}^\alpha$ and $\mathcal{X}^\beta$ are \emph{compatible} if $\mathcal{X}^\alpha = \mathcal{X}^\beta|_{B_\alpha}$. We claim that 
\[
T[\alpha,\mathcal{X}^\alpha] = \min_{\text{$\mathcal{X}^\beta$ compatible with $\mathcal{X}^\alpha$}} T[\beta,\mathcal{X}^\beta] .
\]
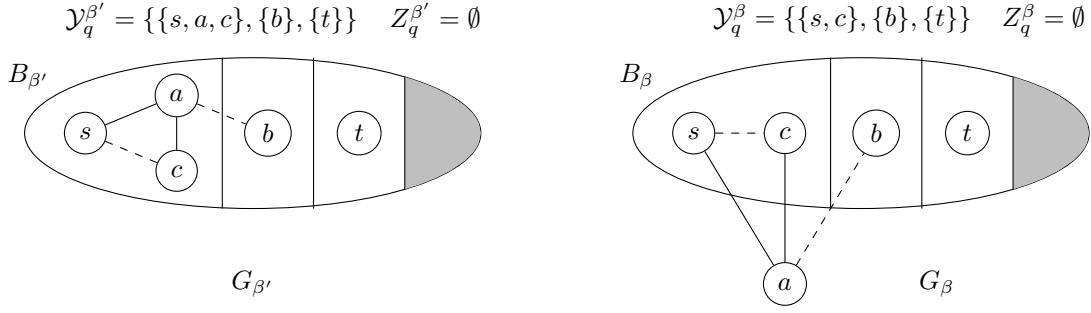
\begin{figure}[!t]
\centering
\begin{tikzpicture}

\draw[color=black] (0,0) ellipse [x radius=3cm, y radius=1cm]{};

\fill[lightgray] (2,0.7453) -- (2.2,0.6798) -- (2.4,0.6) -- (2.6,0.4988)
-- (2.8,0.3590) -- (2.9,0.2560) -- (2.95,0.1818) -- (2.975,0.1288)
-- (3,0) 
-- (2.975,- 0.1288) -- (2.95,- 0.1818) -- (2.9,-0.2560) -- (2.8,-0.3590)
 -- (2.6,-0.4988) -- (2.4,-0.6) -- (2.2,-0.6798)
 -- (2,-0.7453)  -- cycle;

\draw (-0.4,1) -- (-0.4,-1);
\draw (0.8,1) -- (0.8,-1);
\draw (2,0.7453) -- (2,-0.7453);

\node (s) at (-2.2,0)[shape=circle, draw]{$s$};
\node (a) at (-1.0,0.5)[shape=circle, draw]{$a$};
\node (c) at (-1,-0.5)[shape=circle, draw]{$c$};
\node (b) at (0.2,0)[shape=circle, draw]{$b$};
\node (t) at (1.4,0)[shape=circle, draw]{$t$};

\draw (s) -- (a);
\draw (a) -- (c);
\draw[dashed] (s) -- (c);
\draw[dashed] (a) -- (b);

\node[align=center] at (-0.55,1.5){$\mathcal{Y}^{\beta'}_q = \left\{ \{s,a,c\}, \{b\}, \{t\} \right\}$};
\node[align=center] at (2.4,1.5){$Z^{\beta'}_q = \emptyset $};

\node[align=center] at (-2.95,0.75){$B_{\beta'}$};
\node[align=center] at (0,-2){$G_{\beta'}$};

\draw[color=black] (0+8,0) ellipse [x radius=3cm, y radius=1cm]{};

\fill[lightgray] (2+8,0.7453) -- (2.2+8,0.6798) -- (2.4+8,0.6) -- (2.6+8,0.4988)
-- (2.8+8,0.3590) -- (2.9+8,0.2560) -- (2.95+8,0.1818) -- (2.975+8,0.1288)
-- (3+8,0) 
-- (2.975+8,- 0.1288) -- (2.95+8,- 0.1818) -- (2.9+8,-0.2560) -- (2.8+8,-0.3590)
 -- (2.6+8,-0.4988) -- (2.4+8,-0.6) -- (2.2+8,-0.6798)
 -- (2+8,-0.7453)  -- cycle;

\draw (-0.4+8,1) -- (-0.4+8,-1);
\draw (0.8+8,1) -- (0.8+8,-1);
\draw (2+8,0.7453) -- (2+8,-0.7453);

\node (s) at (-2.2+8,0)[shape=circle, draw]{$s$};
\node (a) at (-1.0+8,-2)[shape=circle, draw]{$a$};
\node (c) at (-1+8, 0)[shape=circle, draw]{$c$};
\node (b) at (0.2+8,0)[shape=circle, draw]{$b$};
\node (t) at (1.4+8,0)[shape=circle, draw]{$t$};

\draw (s) -- (a);
\draw (a) -- (c);
\draw[dashed] (s) -- (c);
\draw[dashed] (a) -- (b);

\node[align=center] at (-0.25+8,1.5){$\mathcal{Y}^\beta_q = \left\{ \{s,c\}, \{b\}, \{t\} \right\}$};
\node[align=center] at (2.4+8,1.5){$Z^\beta_q = \emptyset $};
\node[align=center] at (-2.95+8,0.75){$B_\beta$};
\node[align=center] at (0+8+1,-2){$G_\beta$};

\end{tikzpicture}
\caption{Example for a segmentation $q$ of two compatible signatures $\mathcal{X}^\beta$ and $\mathcal{X}^{\beta'}$ for a forget node~$\beta$ with child node~$\beta'$ in~$\mathbb{T}^*$.}
\label{fig:TDCAlgo-fn}
\end{figure}
Since node $\alpha$ is a forget node for the vertex $v\in V(G)$, all edges incident with~$v$ have been introduced in the subtree rooted at $\alpha$. Therefore, every possible way of $v$ appearing in a forest has been considered. We remark that $G_\alpha$ and $G_\beta$ are equal.
In \Cref{fig:TDCAlgo-fn}, we provide an example for a segmentation $q$ of two compatible signatures $\mathcal{X}^\beta$ and $\mathcal{X}^{\beta'}$ for a forget node $\beta$ with child node $\beta'$ in $\mathbb{T}^*$ of~\Cref{fig:modifiedTDC}. All lines connecting two vertices correspond to the edges in the graphs $G_{\beta'}$ and $G_\beta$, where only the solid lines are the edges in the partial solutions that induce the heading segmentations. Node~$\beta$ forgets vertex~$a$. The vertices~$s, a, c\in V_{\beta'}$ form a segment in~$\mathcal{Y}^{\beta'}_q$. As node~$\beta$ forgets vertex~$a$, the vertices~$s,c \in V_\beta$ form a segment in $\mathcal{Y}^{\beta}_q$, since they are connected via the vertex~$a$.
\vspace{3pt}\emph{Correctness}. ``$\geq$'': Let $L_\alpha$ be a partial solution for~$G_\alpha$ with signature~$\mathcal{X}^\alpha$ such that $T[\alpha,\mathcal{X}^\alpha] = c(L_\alpha)$. We construct a partial solution~$L_\beta$ for~$G_\beta$ and a signature~$\mathcal{X}^\beta$ such that $L_\beta$ induces $\mathcal{X}^\beta$ and $\mathcal{X}^\beta$ is compatible with $\mathcal{X}^\alpha$. Since $G_\alpha = G_\beta$, the set of $p$~forests~$L_\beta:=L_\alpha$ is a partial solution for~$G_\beta$. We set~$\mathcal{X}^\beta|_{B_\alpha}:=\mathcal{X}^\alpha$. For each $q\in[p]$, if vertex~$v$ does not appear in the forest with index $q$, then we set $Z^\beta_q:=Z^\alpha_q\cup \{v\}$ and $\mathcal{Y}^\beta_q:=\mathcal{Y}^\alpha_q$. If vertex~$v$ appears in the forest with index~$q$, then we set $Z^\beta_q:=Z^\alpha_q$, and we add~$v$ to the segmentation~$\mathcal{Y}^\beta_q$ as follows. If vertex~$v$ appears as a single-vertex tree in the forest with index~$q$, then we add~$\{v\}$ to~$\mathcal{Y}^\beta_q$. If vertex~$v$ appears in a tree with vertices in~$M\in\mathcal{Y}^\alpha_q$, then we set $\mathcal{Y}^\beta_q := (\mathcal{Y}^\alpha_q\backslash \{M\})\cup \{M\cup \{v\}\}$. Signature~$\mathcal{X}^\beta$ is compatible with signature~$\mathcal{X}^\alpha$, and the partial solution~$L_\beta$ induces $\mathcal{X}^\beta$. It follows that
\[
T[\alpha,\mathcal{X}^\alpha] = c(L_\alpha) = c(L_\beta) \geq T[\beta,\mathcal{X}^\beta] \geq  \min_{\text{$\mathcal{X'}^{\beta}$ compatible with $\mathcal{X}^\alpha$}} T[\beta,\mathcal{X'}^{\beta}] .
\]
``$\leq$'': Let $L_\beta$ be a partial solution for $G_\beta$ with signature $\mathcal{X}^\beta$ compatible with $\mathcal{X}^\alpha$ such that $T[\beta,\mathcal{X}^\beta]=c(L_\beta)$ and $T[\beta,\mathcal{X}^\beta]= 
\min_{\text{$\mathcal{X'}^{\beta}$ compatible with $\mathcal{X}^\alpha$}} T[\beta,\mathcal{X'}^{\beta}]$. We construct a partial solution for $G_\alpha$ that induces $\mathcal{X}^\alpha$. Since $G_\alpha = G_\beta$, we set $L_\alpha := L_\beta$ as the partial solution~$L_\alpha$ for~$G_\alpha$. Since $\mathcal{X}^\alpha=\mathcal{X}^\beta|_{B_\alpha}$, the partial solution~$L_\alpha$ induces~$\mathcal{X}^\alpha$. It follows that
\[
\min_{\text{$\mathcal{X'}^{\beta}$ compatible with $\mathcal{X}^\alpha$}} T[\beta,\mathcal{X'}^{\beta}] = T[\beta,\mathcal{X}^\beta] = c(L_\beta) = c(L_\alpha) \geq T[\alpha,\mathcal{X}^\alpha] .
\]

\vspace{3pt}\emph{Running time}. For a signature~$\mathcal{X}^\alpha$, we check whether $\mathcal{X}^\alpha=\mathcal{X}^\beta|_{B_\alpha}$ for all signatures~$\mathcal{X}^\beta$ for node~$\beta$. This can be done in~$O(p\cdot (|B_\beta|+1)^{p\cdot |B_\beta|+2})$ time. Since $|B_\alpha|\leq |B_\beta|\leq \omega+1$, the overall running time for filling all entries in $T$ for a forget node is in~$O(p\cdot (\omega+2)^{2\cdot p\cdot(\omega+1)+2})$.
\fi
\subparagraph*{Introduce Edge Node.}
Let $\alpha$ be an introduce edge node of $\mathbb{T}$, let $\beta$ be the child node of~$\alpha$, and let $e=\{v,w\}$ be the edge introduced by node~$\alpha$. Two signatures~$\mathcal{X}^\alpha$ and~$\mathcal{X}^\beta$ are \emph{compatible} if for each~$q\in[p]$, one of the following conditions holds:
\begin{compactenum}[(i)]
\item $\mathcal{Y}^\alpha_q=\mathcal{Y}^\beta_q$, or
\item $\mathcal{Y}^\alpha_q = (\mathcal{Y}^\beta_q\backslash \{M_1,M_2\}) \cup \{M_1\cup M_2\}$ with $M_1,M_2\in \mathcal{Y}^\beta_q$, $M_1\neq M_2$, and $v\in M_1$ and $w\in M_2$.
\end{compactenum}
If $\mathcal{X}^\alpha$ and $\mathcal{X}^\beta$ are compatible, then let $Q\subseteq [p]$ be the set of indices such that for all $q\in Q$ (ii) holds and for all $q\in [p]\backslash Q$ (i) holds. We say that $\mathcal{X}^\alpha$ and $\mathcal{X}^\beta$ are \emph{share-compatible} if~$|Q|\geq 2$. We claim that
\[
 T[\alpha,\mathcal{X}^\alpha] = \min_{\text{$\mathcal{X}^\beta$ compatible with $\mathcal{X}^\alpha$}} \left(T[\beta,\mathcal{X}^\beta] + 
 \begin{cases}
 1 ,& \text{if $\mathcal{X}^\beta$ and $\mathcal{X}^\alpha$ are share-compatible,} \\
 0 ,& \text{otherwise.} 
 \end{cases}\right)
\]
In other words, two signatures $\mathcal{X}^\alpha$ for node $\alpha$ and $\mathcal{X}^\beta$ for node $\beta$ are compatible if and only if for all $q\in [p]$, either by (i) it holds that the segmentation $q$ in $\mathcal{X}^\alpha$ is equal to the segmentation~$q$ of~$\mathcal{X}^\beta$, or by (ii) it holds that the segmentation~$q$ of~$\mathcal{X}^\alpha$ is the result of merging two segments in the segmentation~$q$ of $\mathcal{X}^\beta$, where none of the two segments is the zero-segment, and vertex~$v$ is in the one segment, and vertex~$w$ is in the other segment. This corresponds to connecting two trees by edge $e$ in the forest with index~$q$, where $v$~is in the one tree and $w$ in the other tree. Note that connecting two vertex-disjoint trees by exactly one edge yields a tree. The deletion of edge $e$ in every forest of a partial solution for~$G_\alpha$ that includes the edge~$e$ yields a partial solution for~$G_\beta$.
We remark that $G_\alpha = G_\beta+\{e\}$, that is, $G_\alpha$ differs from~$G_\beta$ only by the additional edge $e$. 
\iflong{}%
In \Cref{fig:TDCAlgo-ien}, we provide an example for a segmentation $q$ of two compatible signatures $\mathcal{X}^\delta$ and $\mathcal{X}^{\delta'}$ of an introduce edge node $\delta$ with child node $\delta'$ in $\mathbb{T}^*$. Graph~$G_{\delta'}$ does not contain any edge. Edge~$\{e,t\}$ is introduced by node~$\delta$. This allows to connect the segments containing vertex~$e$ on the one hand, and vertex~$t$ on the other hand, using edge~$\{e,t\}$.
\begin{figure}[tb]
\centering
\begin{tikzpicture}

\draw[color=black] (0,0) ellipse [x radius=3cm, y radius=1cm]{};

\fill[lightgray] (2,0.7453) -- (2.2,0.6798) -- (2.4,0.6) -- (2.6,0.4988)
-- (2.8,0.3590) -- (2.9,0.2560) -- (2.95,0.1818) -- (2.975,0.1288)
-- (3,0) 
-- (2.975,- 0.1288) -- (2.95,- 0.1818) -- (2.9,-0.2560) -- (2.8,-0.3590)
 -- (2.6,-0.4988) -- (2.4,-0.6) -- (2.2,-0.6798)
 -- (2,-0.7453)  -- cycle;

\node (s) at (-2.0,0)[shape=circle, draw]{$s$};
\node (b) at (2.5,0)[shape=circle, draw]{$b$};
\node (e) at (-0.5,0)[shape=circle, draw]{$e$};
\node (t) at (1,0)[shape=circle, draw]{$t$};

\draw (-1.25,0.9) -- (-1.25,-0.9);
\draw (0.25,1) -- (0.25,-1);
\draw (2,0.7453) -- (2,-0.7453);

\node[align=center] at (-0.25,1.5){$\mathcal{Y}^{\delta'}_q = \{\{s\},\{e\},\{t\}\}$};
\node[align=center] at (2.5,1.5){$Z^{\delta'}_q = \{b\}$};

\node[align=center] at (-2.95,0.75){$B_{\delta'}$};
\node[align=center] at (1.5,-2){$G_{\delta'}$};

\draw[color=black] (0+8,0) ellipse [x radius=3cm, y radius=1cm]{};

\fill[lightgray] (2+8,0.7453) -- (2.2+8,0.6798) -- (2.4+8,0.6) -- (2.6+8,0.4988)
-- (2.8+8,0.3590) -- (2.9+8,0.2560) -- (2.95+8,0.1818) -- (2.975+8,0.1288)
-- (3+8,0) 
-- (2.975+8,- 0.1288) -- (2.95+8,- 0.1818) -- (2.9+8,-0.2560) -- (2.8+8,-0.3590)
 -- (2.6+8,-0.4988) -- (2.4+8,-0.6) -- (2.2+8,-0.6798)
 -- (2+8,-0.7453)  -- cycle;

\node (s) at (-2.0+8,0)[shape=circle, draw]{$s$};
\node (b) at (2.5+8,0)[shape=circle, draw]{$b$};
\node (e) at (-0.5+8,0)[shape=circle, draw]{$e$};
\node (t) at (1+8,0)[shape=circle, draw]{$t$};

\draw (-1.25+8,0.9) -- (-1.25+8,-0.9);
\draw (2+8,0.7453) -- (2+8,-0.7453);

\draw[color=blue] (e) -- (t);

\node[align=center] at (-0.5+8,1.5){$\mathcal{Y}^\delta_q = \{\{s\},\{e,t\}\}$};
\node[align=center] at (2.25+8,1.5){$Z^\delta_q = \{b\}$};

\node[align=center] at (-2.95+8,0.75){$B_\delta$};
\node[align=center] at (1.5+8,-2){$G_\delta$};

\end{tikzpicture}
\caption{Example for a segmentation $q$ of two compatible signatures $\mathcal{X}^\delta$ and $\mathcal{X}^{\delta'}$ of an introduce edge node $\delta$ with child node $\delta'$ in $\mathbb{T}^*$.}
\label{fig:TDCAlgo-ien}
\end{figure}
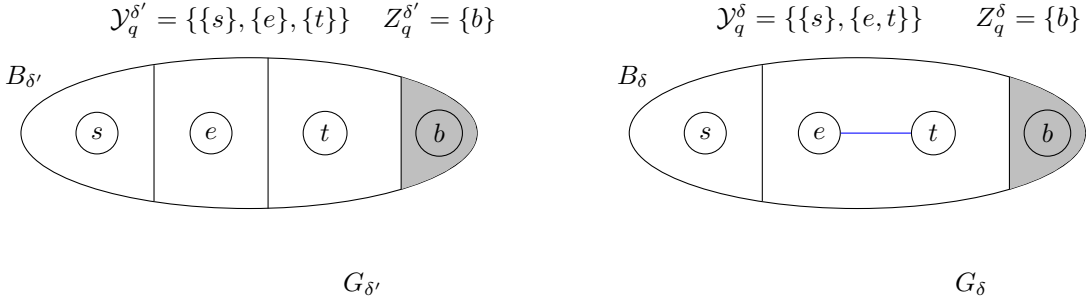
\fi{}

\iflong{}
\emph{Correctness}. ``$\geq$'': Let $L_\alpha$ be a partial solution for $G_\alpha$ with signature $\mathcal{X}^\alpha$ such that $T[\alpha,\mathcal{X}^\alpha] = c(L_\alpha)$. We construct a partial solution $L_\beta$ for $G_\beta$ and a signature $\mathcal{X}^\beta$ such that $L_\beta$ induces $\mathcal{X}^\beta$ and $\mathcal{X}^\beta$ is compatible with $\mathcal{X}^\alpha$. For each $q\in[p]$, if edge~$e$ is not part of the forest with index $q$ in $L_\alpha$, then the forest is a forest in $G_\beta$ as well. Then, we set $Z^\beta_q := Z^\alpha_q$ and $\mathcal{Y}^\beta_q := \mathcal{Y}^\alpha_q$. If edge~$e$ is part of the forest with index~$q$ in~$L_\alpha$, then deleting edge~$e$ from the forest with index~$q$ disconnects a tree of the forest such that two trees result, with $v$ in the one tree and $w$ in the other tree. Let $M\in \mathcal{Y}^\alpha_q$ be the segment in the segmentation~$q$ with $v,w\in M$. Let $M_1,M_2$ be the induced sets by splitting tree~$T_\alpha$ in the forest with index~$q$ in~$L_\alpha$ at edge~$e$, that is, if $T_1$ and $T_2$ are the connected subgraphs of $T_\alpha\backslash\{e\}$, then $M_1:=V(T_1)\cap B_\alpha$ and $M_2:=V(T_2)\cap B_\alpha$. We set $\mathcal{Y}^\beta_q := (\mathcal{Y}^\alpha_q\backslash \{M\})\cup \{M_1,M_2\}$ and $Z^\beta_q:=Z^\alpha_q$. Signature~$\mathcal{X}^\beta$ for node~$\beta$ is compatible with signature~$\mathcal{X}^\alpha$ for node~$\alpha$.

Let $L_\beta$ be the set of $p$ forests in $L_\alpha$ restricted to edge set~$E_\beta$. Then, $L_\beta$ is a partial solution for~$G_\beta$ and induces signature~$\mathcal{X}^\beta$. If edge~$e$ appears in more than one of the $p$ forests in $L_\alpha$, then $c(L_\beta)=c(L_\alpha)-1$ and the signatures $\mathcal{X}^\alpha$ and $\mathcal{X}^\beta$ are share-compatible. It follows that
\[
T[\alpha,\mathcal{X}^\alpha] = c(L_\alpha) = c(L_\beta) + 1 \geq T[\beta,\mathcal{X}^\beta] + 1 \geq  \min_{\text{$\mathcal{X'}^{\beta}$ compatible with $\mathcal{X}^\alpha$}} T[\beta,\mathcal{X'}^{\beta}] + 1 .
\]
If edge~$e$ appears in at most one of the $p$ forests in the partial solution~$L_\alpha$, then
\[
T[\alpha,\mathcal{X}^\alpha] = c(L_\alpha) = c(L_\beta) \geq T[\beta,\mathcal{X}^\beta] \geq  \min_{\text{$\mathcal{X'}^{\beta}$ compatible with $\mathcal{X}^\alpha$}} T[\beta,\mathcal{X'}^{\beta}].
\]

``$\leq$'': Let $L_\beta$ be a partial solution for~$G_\beta$ with signature~$\mathcal{X}^\beta$ compatible with~$\mathcal{X}^\alpha$ such that $T[\beta,\mathcal{X}^\beta]=c(L_\beta)$ and $T[\beta,\mathcal{X}^\beta]= 
\min_{\text{$\mathcal{X'}^{\beta}$ compatible with $\mathcal{X}^\alpha$}} T[\beta,\mathcal{X'}^{\beta}]$. We construct a partial solution~$L_\alpha$ for~$G_\alpha$ that induces signature~$\mathcal{X}^\alpha$. For each $q\in[p]$, if condition~(i) holds, that is, if $\mathcal{Y}^\alpha_q=\mathcal{Y}^\beta_q$, then we set the forest with index~$q$ in~$L_\alpha$ to the forest with index~$q$ in the partial solution~$L_\beta$. In case of condition (ii), that is, if $v$ and $w$ belong to the same segment in the segmentation $q$ of $\mathcal{X}^\alpha$ but are not in the same segment in the segmentation $q$ of $\mathcal{X}^\beta$, then we add edge~$e$ to the forest with index~$q$ in the partial solution~$L_\beta$ and we set the resulting forest as the forest with index~$q$ in the partial solution~$L_\alpha$. Since the vertices~$v$ and~$w$ are in two vertex-disjoint trees in the forest with index~$q$ in~$L_\beta$, adding edge~$e$ connects the two trees at the vertices~$v$ and~$w$, which results again in a tree. The set of~$p$~forests~$L_\alpha$, constructed as mentioned above, is a partial solution for~$G_\alpha$ and induces signature~$\mathcal{X}^\alpha$. If the signatures~$\mathcal{X}^\beta$ and~$\mathcal{X}^\alpha$ are share-compatible, then the partial solution~$L_\alpha$ is the result of adding edge~$e$ to at least two forests in~$L_\beta$, and thus, the number of common edges of the forests increases by exactly one, i.e.\ $c(L_\beta) = c(L_\alpha)-1$. It follows that
\[
\min_{\text{$\mathcal{X'}^{\beta}$ compatible with $\mathcal{X}^\alpha$}} T[\beta,\mathcal{X'}^{\beta}] = T[\beta,\mathcal{X}^\beta] = c(L_\beta) = c(L_\alpha) - 1 \geq T[\alpha,\mathcal{X}^\alpha] - 1 .
\]
If the signatures~$\mathcal{X}^\beta$ and~$\mathcal{X}^\alpha$ are compatible but not share-compatible, then the partial solution~$L_\alpha$ is the result of adding edge~$e$ to at most one forest in~$L_\beta$. It follows that
\[
\min_{\text{$\mathcal{X'}^{\beta}$ compatible with $\mathcal{X}^\alpha$}} T[\beta,\mathcal{X'}^{\beta}] = T[\beta,\mathcal{X}^\beta] = c(L_\beta) = c(L_\alpha) \geq T[\alpha,\mathcal{X}^\alpha].
\]
\fi{}

\vspace{3pt}\emph{Running time}. For each signature~$\mathcal{X}^\alpha$, we check all signatures~$\mathcal{X}^\beta$  for node~$\beta$ for compatibility, that means, we need to check for each~$q\in[p]$ whether the segmentations are equal (i) or whether the segmentation~$q$ of~$\mathcal{X}^\alpha$ is derived by merging two segments in the segmentation~$q$ of $\mathcal{X}^\beta$ (ii). To check condition (i) as well as to check condition (ii) can be done in~$O(p\cdot |B_\alpha|^2)$~time. Therefore, the overall running time for filling all entries in~$T$ for an introduce edge node is in~$O(p\cdot (\omega+2)^{2\cdot p\cdot (\omega+1)+2})$.
\iflong{}
\subparagraph*{Join Node.}
Let $\alpha$ be a join node of $\mathbb{T}$ and let $\beta,\gamma$ be the two child nodes of $\alpha$. A signature~$\mathcal{X}^\alpha$ for node $\alpha$ and a pair of two signatures~$\mathcal{X}^\beta$ for node~$\beta$ and~$\mathcal{X}^\gamma$ for node $\gamma$ are \emph{compatible} if for all $q\in[p]$ it holds that
\begin{enumerate}[(i)]
\item $Z^\alpha_q=Z^\beta_q=Z^\gamma_q$,
\item $v,w\in M^\alpha\in \mathcal{Y}^\alpha_q$ with $v\neq w$ if and only if there exists $\ell\geq 1$ and $M_1,\ldots,M_\ell\in \mathcal{Y}^\beta_q\cup\mathcal{Y}^\gamma_q$ with $|M_i\cap M_{i+1}|=1$ for all $i=1,\ldots,\ell-1$ and $v\in M_1$ and $w\in M_\ell$,
\item for all $M^\beta\in \mathcal{Y}^\beta_q$ and $M^\gamma\in \mathcal{Y}^\gamma_q$ holds $|M^\beta\cap M^\gamma| \leq 1$, and
\item there do not exist $\ell\geq 3$ and $M_1,\ldots,M_\ell\in \mathcal{Y}^\beta_q\cup\mathcal{Y}^\gamma_q$ with $|M_i\cap M_{i+1}|=1$ for all $i=1,\ldots,\ell-1$ and $M_i\neq M_j$ for all $i,j\in[\ell]$, $i\neq j$, such that $v\in M_1$ and $v\in M_\ell$. 
\end{enumerate} 
We claim that
\[
T[\alpha,\mathcal{X}^\alpha] = \min\limits_{\text{$(\mathcal{X}^\beta,\mathcal{X}^\gamma)$ compatible with $\mathcal{X}^\alpha$}} (T[\beta,\mathcal{X}^\beta]+T[\gamma,\mathcal{X}^\gamma] ) .
\]
In other words, a signature~$\mathcal{X}^\alpha$ is compatible with a pair of two signatures~$\mathcal{X}^\beta$ for node~$\beta$ and $\mathcal{X}^\gamma$ for node~$\gamma$, if and only if for every $q\in[p]$ it holds that 
\begin{enumerate}[(i)]
\item the vertices that appear in the segmentations with index~$q$ in all three signatures are the same,
\item every segment in the segmentation~$q$ of~$\mathcal{X}^\alpha$ is a union of segments in the segmentation~$q$ in~$\mathcal{X}^\beta$ and segments in the segmentation~$q$ in~$\mathcal{X}^\gamma$,
\item every pair of segments with one segment in the segmentation~$q$ in~$\mathcal{X}^\beta$ and one segment in the segmentation~$q$ in~$\mathcal{X}^\gamma$ has at most one vertex in~$B_\beta=B_\gamma$ in common, and
\item there is no chain of at least three segments in the union of the segmentations with index~$q$ in $\mathcal{X}^\beta$ and $\mathcal{X}^\gamma$ with one vertex in the first and last segment.
\end{enumerate}
We say that segments $M_1,\ldots,M_\ell$, $\ell\geq 2$, form a \emph{chain} of segments, if $|M_i\cap M_{i+1}|=1$ for all $i=1,\ldots,\ell-1$ and $M_i\neq M_j$ for all $i,j\in[\ell]$, $i\neq j$.

Intuitively, (i)-(iv) define how to combine segmentations in signatures of child nodes to segmentations in a signature of a join node. The condition (ii) ensures that every forest with index~$q$ in a partial solution $L_\alpha$ for $G_\alpha$ is a union of the two forests with index~$q$ in some partial solutions~$L_\beta$ for~$G_\beta$ and~$L_\gamma$ for~$G_\gamma$, respectively. 

\begin{figure}[!t]
\centering
\begin{tikzpicture}[x=0.8cm, y=0.8cm]

\tikzstyle{every node}=[scale=0.4798255,font={\huge}];

\node[align=center, red] (one) at (4,-0.5){not possible};

\draw[color=black] (0+4,0-2) ellipse [x radius=1.25cm, y radius=0.35cm]{};

\node (v) at (-0.5+4,-2)[shape=circle,draw]{};
\node (w) at (0.5+4,-2)[shape=circle,draw]{};
\node (u) at (-1+4,-4)[shape=circle, draw]{};
\node (x) at (1+4,-4)[shape=circle, draw]{};

\def\colone{gray!90}%
\def\coltwo{gray!40}%

\draw[decorate,decoration=snake,\colone] (v) to [out=225, in=135](u);
\draw[decorate,decoration=snake,\coltwo] (w) to [out=-45, in=45](x);

\draw[decorate,decoration=snake,\colone] (u) to [out=45, in=135](x);
\draw[decorate,decoration=snake,\coltwo] (u) to [out=-45, in=-135](x);

\node[align=center, red] (one) at (2+4+5,-0.5){not allowed};

\draw[color=black] (0+4+5,0-2-1) ellipse [x radius=1.25cm, y radius=0.35cm]{};

\node (v) at (-0.5+4+5,-2-1)[shape=circle,draw]{};
\node (w) at (0.5+4+5,-2-1)[shape=circle,draw]{};

\draw[decorate,decoration=snake,\colone] (v) to [out=90, in=180](0+4+5,-2-1+1.5);
\draw[decorate,decoration=snake,\colone] (w) to [out=90, in=0](0+4+5,-2-1+1.5);
\draw[decorate,decoration=snake,\coltwo] (v) to [out=270, in=180](0+4+5,-2-1-1.5);
\draw[decorate,decoration=snake,\coltwo] (w) to [out=270, in=0](0+4+5,-2-1-1.5);

\draw[color=black] (0+4+5+4.5,0-2-1) ellipse [x radius=2cm, y radius=0.35cm]{};

\node (v) at (-0.5+4+5+4.5,-2-1)[shape=circle,draw]{};
\node (w) at (0.5+4+5+4.5,-2-1)[shape=circle,draw]{};

\node (u) at (-1.5+4+5+4.5,-2-1)[shape=circle,draw]{};
\node (x) at (1.5+4+5+4.5,-2-1)[shape=circle,draw]{};

\draw[decorate,decoration=snake,\colone] (u) to [out=270, in=180](-1+4+5+4.5,-2-1-1.5);
\draw[decorate,decoration=snake,\colone] (v) to [out=270, in=0](-1+4+5+4.5,-2-1-1.5);

\draw[decorate,decoration=snake,\colone] (w) to [out=270, in=180](1+4+5+4.5,-2-1-1.5);
\draw[decorate,decoration=snake,\colone] (x) to [out=270, in=0](1+4+5+4.5,-2-1-1.5);

\draw[decorate,decoration=snake,\coltwo] (v) to [out=90, in=180](0+4+5+4.5,-2-1+1.25);
\draw[decorate,decoration=snake,\coltwo] (w) to [out=90, in=0](0+4+5+4.5,-2-1+1.25);

\draw[decorate,decoration=snake,\coltwo] (u) to [out=90, in=180](0+4+5+4.5,-2-1+1.75);
\draw[decorate,decoration=snake,\coltwo] (x) to [out=90, in=0](0+4+5+4.5,-2-1+1.75);

\end{tikzpicture}
\caption{Sketch of scenarios when combining two forests with the same index in two partial solutions for the two child nodes of a join node.}
\label{fig:TDCAlgo-jn-scen}
\end{figure}
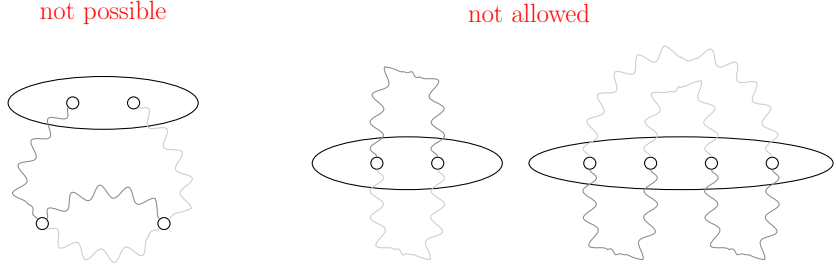
\Cref{fig:TDCAlgo-jn-scen} exemplifies three scenarios of cycle creation caused by a union of two forests, where the colors dark-gray and light-gray %
indicate each of the two forests. We used curved lines to highlight that there could be trees connecting two vertices. The scenario on the left-hand side illustrates a situation that is not possible. The scenario implies that there are vertices that have been forgotten in the subtrees rooted at each child node of the tree decomposition, which is not possible by the definition of a tree decomposition (cf.\ \cref{sec:prelim}). The two scenarios on the right-hand side illustrate two scenarios that are not allowed to occur by our definition of compatibility of join nodes. More precisely, conditions (iii) and~(iv) ensure that none of these two scenarios occurs.  

The conditions (iii) and (iv) ensure that a union of two forests in~$L_\beta$ and~$L_\gamma$ does not close a cycle. 
Condition~(iii) prevents the following creation of cycles. If there is a tree~$T_\beta$ in the forest with index~$q$ in~$L_\beta$ and a tree~$T_\gamma$ in the forest with index~$q$ in~$L_\gamma$ that have at least two vertices in common, then the union of these two trees creates a cycle in $G_\alpha$. 
Condition~(iv) prevents the following creation of cycles. Let $v$ be a vertex in $V_\alpha$ such that there exist some trees $T_1,\ldots, T_\ell$ in the forests with index~$q$ in~$L_\beta$ and~$L_\gamma$ such that $|V(T_i)\cap V(T_{i+1})|=1$ for $i=1,\ldots,\ell-1$ and $v\in V(T_1)$ and $v\in V(T_\ell)$. Then the graph $T^\alpha=T_1\cup\ldots\cup T_\ell$ as union of the trees in~$G_\alpha$ contains a cycle and vertex~$v$ is part of a cycle in~$T^\alpha$. 

\iflong{}
\begin{figure}[!t]
\centering
\begin{tikzpicture}

\draw[color=black] (0,0+3) ellipse [x radius=3cm, y radius=1cm]{};

\fill[lightgray] (2,0.7453+3) -- (2.2,0.6798+3) -- (2.4,0.6+3) -- (2.6,0.4988+3)
-- (2.8,0.3590+3) -- (2.9,0.2560+3) -- (2.95,0.1818+3) -- (2.975,0.1288+3)
-- (3,0+3) 
-- (2.975,- 0.1288+3) -- (2.95,- 0.1818+3) -- (2.9,-0.2560+3) -- (2.8,-0.3590+3)
 -- (2.6,-0.4988+3) -- (2.4,-0.6+3) -- (2.2,-0.6798+3)
 -- (2,-0.7453+3)  -- cycle;

\node (s) at (-1,0+3)[shape=circle, draw]{$s$};
\node (a) at (-1.0,-2.5+3)[shape=circle, draw]{$a$};
\node (b) at (0.2,0+3)[shape=circle, draw]{$b$};
\node (c) at (-2.2,0+3)[shape=circle, draw]{$c$};
\node (d) at (2.5,0+3)[shape=circle, draw]{$d$};
\node (t) at (1.4,0+3)[shape=circle, draw]{$t$};

\draw (-1.6,0.85+3) -- (-1.6,-0.85+3);
\draw (0.9,0.95+3) -- (0.9,-0.95+3);
\draw (2,0.7453+3) -- (2,-0.7453+3);

\draw[blue] (s) -- (a);
\draw[blue] (a) -- (b);
\draw[dashed] (a) -- (c);
\draw[dashed] (s) -- (c);

\node[align=center] at (-0.55,2.5+2){$\mathcal{Y}^{\alpha_\ell}_q=\left\{\{s,b\},\{c\},\{t\}\right\}$};
\node[align=center] at (2.5,2.5+2){$Z^{\alpha_\ell}_q=\{d\}$};

\node[align=center] at (-2.95,1.75+2){$B_{\alpha_\ell}$};
\node[align=center] at (1,-0.75+2){$G_{\alpha_\ell}$};

\def\x{2.5}

\draw[color=black] (0,0-\x) ellipse [x radius=3cm, y radius=1cm]{};

\fill[lightgray] (2,0.7453-\x) -- (2.2,0.6798-\x) -- (2.4,0.6-\x) -- (2.6,0.4988-\x)
-- (2.8,0.3590-\x) -- (2.9,0.2560-\x) -- (2.95,0.1818-\x) -- (2.975,0.1288-\x)
-- (3,0-\x) 
-- (2.975,- 0.1288-\x) -- (2.95,- 0.1818-\x) -- (2.9,-0.2560-\x) -- (2.8,-0.3590-\x)
 -- (2.6,-0.4988-\x) -- (2.4,-0.6-\x) -- (2.2,-0.6798-\x)
 -- (2,-0.7453-\x)  -- cycle;

\node (s) at (-1,0-\x)[shape=circle, draw]{$s$};
\node (b) at (0.2,0-\x)[shape=circle, draw]{$b$};
\node (c) at (-2.2,0-\x)[shape=circle, draw]{$c$};
\node (d) at (2.5,0-\x)[shape=circle, draw]{$d$};
\node (e) at (2,-2.5-\x)[shape=circle, draw]{$e$};
\node (t) at (1.4,0-\x)[shape=circle, draw]{$t$};

\draw (-1.6,0.85-\x) -- (-1.6,-0.85-\x);
\draw (-0.4,0.995-\x) -- (-0.4,-0.995-\x);
\draw (2,0.7453-\x) -- (2,-0.7453-\x);

\draw[dashed] (d) -- (e);
\draw[dashed] (e) -- (t);
\draw[darkgreen] (b) -- (t);

\node[align=center] at (-0.55,2.5-\x-1){$\mathcal{Y}^{\alpha_r}_q=\left\{\{s\},\{b,t\},\{c\}\right\}$};
\node[align=center] at (2.5,2.5-\x-1){$Z^{\alpha_r}_q=\{d\}$};

\node[align=center] at (-2.95,1.75+2-\x-3){$B_{\alpha_r}$};
\node[align=center] at (1,-0.75+2-\x-3){$G_{\alpha_r}$};

\draw[color=black] (0+8,0) ellipse [x radius=3cm, y radius=1cm]{};

\fill[lightgray] (2+8,0.7453) -- (2.2+8,0.6798) -- (2.4+8,0.6) -- (2.6+8,0.4988)
-- (2.8+8,0.3590) -- (2.9+8,0.2560) -- (2.95+8,0.1818) -- (2.975+8,0.1288)
-- (3+8,0) 
-- (2.975+8,- 0.1288) -- (2.95+8,- 0.1818) -- (2.9+8,-0.2560) -- (2.8+8,-0.3590)
 -- (2.6+8,-0.4988) -- (2.4+8,-0.6) -- (2.2+8,-0.6798)
 -- (2+8,-0.7453)  -- cycle;

\draw (2+8,0.7453) -- (2+8,-0.7453);

\node (s) at (-1+8,0)[shape=circle, draw]{$s$};
\node (a) at (-1.0+8,-2.5)[shape=circle, draw]{$a$};
\node (b) at (0.2+8,0)[shape=circle, draw]{$b$};
\node (c) at (-2.2+8,0)[shape=circle, draw]{$c$};
\node (d) at (2.5+8,0)[shape=circle, draw]{$d$};
\node (e) at (2+8,-2.5)[shape=circle, draw]{$e$};
\node (t) at (1.4+8,0)[shape=circle, draw]{$t$};

\draw (-1.6+8,0.85) -- (-1.6+8,-0.85);
\draw (2+8,0.7453) -- (2+8,-0.7453);

\draw[blue] (s) -- (a);
\draw[blue] (a) -- (b);
\draw[dashed] (a) -- (c);
\draw[dashed] (s) -- (c);
\draw[dashed] (d) -- (e);
\draw[dashed] (e) -- (t);
\draw[darkgreen] (b) -- (t);

\node[align=center] at (0+8-0.5,1.5){$\mathcal{Y}^\alpha_q = \left\{\{s,b,t\},\{c\}\right\}$};
\node[align=center] at (0+8+2.25,1.5){$Z^\alpha_q=\{d\}$};

\node[align=center] at (-2.95+8,1.75+2-3){$B_\alpha$};
\node[align=center] at (1+8,-0.75+2-3){$G_\alpha$};

\end{tikzpicture}
\caption{Example for a segmentation~$q$ of three compatible signatures $\mathcal{X}^\alpha$, $\mathcal{X}^{\alpha_\ell}$, and $\mathcal{X}^{\alpha_r}$ for a join node~$\alpha$ with child nodes~$\alpha_\ell$ and~$\alpha_r$ in~$\mathbb{T}^*$.}
\label{fig:TDCAlgo-jn}
\end{figure}
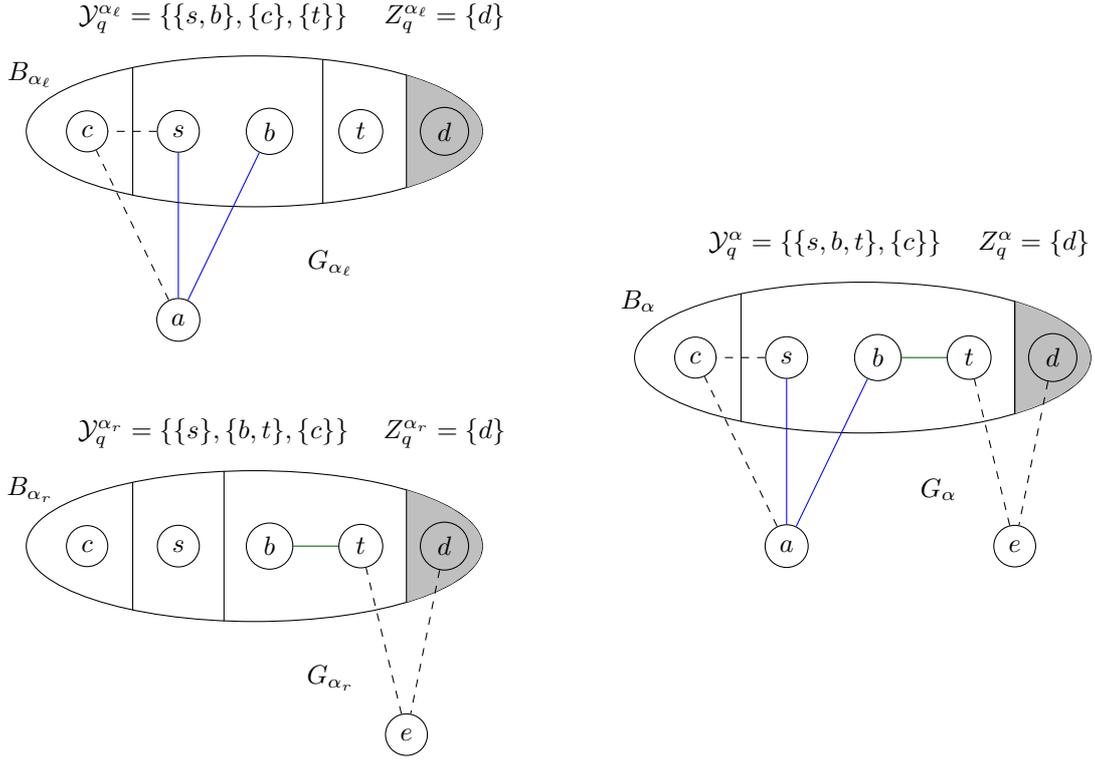
In \Cref{fig:TDCAlgo-jn}, we provide an example for a segmentation~$q$ of three compatible signatures $\mathcal{X}^\alpha$, $\mathcal{X}^{\alpha_\ell}$, and $\mathcal{X}^{\alpha_r}$ for a join node $\alpha$ with child nodes~$\alpha_\ell$ and $\alpha_r$ in $\mathbb{T}^*$. The segments $\{s,b\}$ and $\{t\}$ in $\mathcal{Y}^{\alpha_\ell}_q$ together with the segments $\{s\}$ and $\{b,t\}$ in~$\mathcal{Y}^{\alpha_r}$ form segment~$\{s,b,t\}$ in~$\mathcal{Y}^{\alpha}_q$. Note that the four conditions for compatibility hold. Condition~(i) holds since $Z^{\alpha}_q=Z^{\alpha_\ell}_q=Z^{\alpha_r}_q=\{d\}$. Moreover, note that condition~(iii) holds. According to condition~(ii), note that for any pair in the segment~$\{s,b,t\}\in\mathcal{Y}^\alpha_q$, the segments~$\{s,b\}\in\mathcal{Y}^{\alpha_\ell}_q$ and~$\{b,t\}\in\mathcal{Y}^{\alpha_r}_q$ provide a required chain of segments. Conversely, for any possible chain of segments in~$\mathcal{Y}^{\alpha_\ell}_q\cup \mathcal{Y}^{\alpha_r}_q$, segment~$\{s,b,t\}\in\mathcal{Y}^\alpha_q$ is the required segment in condition~(ii). According to condition~(iv), note that there is no chain of at least three segments in~$\mathcal{Y}^{\alpha_\ell}_q\cup \mathcal{Y}^{\alpha_r}_q$ such that a vertex~$v\in\{s,b,c,t\}$ appears in the first and last segment of the chain.
\fi

\vspace{3pt}\emph{Correctness}. ``$\geq$'': Let $L_\alpha$ be a partial solution for~$G_\alpha$ with signature~$\mathcal{X}^\alpha$ such that $T[\alpha,\mathcal{X}^\alpha] = c(L_\alpha)$. We construct a partial solution~$L_\beta$ for~$G_\beta$, a partial solution~$L_\gamma$ for~$G_\gamma$ and two signatures~$\mathcal{X}^\beta$ and~$\mathcal{X}^\gamma$, such that the pair~$(\mathcal{X}^\beta,\mathcal{X}^\gamma)$ is compatible with~$\mathcal{X}^\alpha$, the partial solution~$L_\beta$ induces signature~$\mathcal{X}^\beta$ and the partial solution~$L_\gamma$ induces signature~$\mathcal{X}^\gamma$. If we restrict each forest in~$L_\alpha$ to the edge sets~$E_\beta$ and~$E_\gamma$, then each forest restricted to~$E_\beta$ is a forest in~$G_\beta$ and each forest restricted to~$E_\gamma$ is a forest in~$G_\gamma$. Therefore, restricting each forest in~$L_\alpha$ to~$E_\beta$ yields a partial solution~$L_\beta$ for~$G_\beta$, and restricting each forest in~$L_\alpha$ to~$E_\gamma$ yields a partial solution~$L_\gamma$ for~$G_\gamma$. We set~$\mathcal{X}^\beta$ and~$\mathcal{X}^\gamma$ as the signatures induced by the partial solutions~$L_\beta$ and~$L_\gamma$ respectively. 

We show that the pair of signatures~$\mathcal{X}^\beta$ and~$\mathcal{X}^\gamma$ is compatible with~$\mathcal{X}^\alpha$. 
Condition (i) holds for every~$q\in[p]$ since every vertex that does not appear in the forest with index~$q$ in~$L_\alpha$ neither appears in the forests with index~$q$ nor in~$L_\beta$ nor in~$L_\gamma$. Since the segmentations with index~$q$ are induced by~$L_\beta$ and~$L_\gamma$, it follows that~$Z^\alpha_q = Z^\beta_q = Z^\gamma_q$ for all~$q\in[p]$.
Suppose that there exists a~$q\in[p]$ such that condition~(iii) does not hold for~$q\in[p]$. This means that there exist~$M^\beta\in \mathcal{Y}^\beta_q$ and~$M^\gamma\in \mathcal{Y}^\gamma_q$ with~$|M^\beta\cap M^\gamma|\geq 2$. Let~$v,w\in M^\beta\cap M^\gamma$. Let $T^\beta$ be the tree in the forest with index~$q$ in~$L_\beta$ corresponding to $M^\beta$ and let $T^\gamma$ be the tree in the forest with index~$q$ in~$L_\gamma$ corresponding to~$M^\gamma$. Note that $v,w\in V(T^\beta)\cap V(T^\gamma)$. By our construction of $L_\beta$ and $L_\gamma$, there is a tree $T^\alpha$ in the forest with index $q$ in $L_\alpha$, such that $T^\beta$ is a subtree of $T^\alpha$ restricted to $E_\beta$, and $T^\gamma$ is a subtree of $T^\alpha$ restricted to $E_\gamma$. Since $v,w\in V(T^\alpha)$, there is a $v$-$w$~path in~$T^\alpha$ using only edges in~$E_\beta$ and a $v$-$w$~path in~$T^\alpha$ using only edges in~$E_\gamma$. Since $E_\beta\cap E_\gamma=\emptyset$, the two paths form a cycle in $T^\alpha$. This is a contradiction to the fact that $T^\alpha$ is a tree.
For condition (ii), direction ``$\Rightarrow$'', we consider~$q\in[p]$, $M^\alpha\in\mathcal{Y}^\alpha_q$ and~$v,w\in M^\alpha$, $v\neq w$, if such a~$M^\alpha\in\mathcal{Y}^\alpha_q$ exists. Segment~$M^\alpha$ corresponds to a tree~$T^\alpha$ in the forest with index~$q$ in~$L_\alpha$. Since~$v,w\in M^\alpha$, the vertices~$v$ and~$w$ appear in tree~$T^\alpha$. Since $E_\alpha=E_\beta\cup E_\gamma$ and $E_\beta\cap E_\gamma=\emptyset$, the restriction of~$T^\alpha$ to~$E_\beta$ and~$E_\gamma$ splits the tree in maximal subtrees~$T_1,\ldots,T_\ell$ alternating by~$G_\beta$ and~$G_\gamma$. Note that $|V(T_i)\cap V(T_j)|\leq 1$ for all $i,j\in[\ell]$, $i\neq j$, and $T^\alpha=T_1\cup\ldots\cup T_\ell$. Let~$M_1,\ldots,M_\ell\in \mathcal{Y}^\beta_q\cup \mathcal{Y}^\gamma_q$ be segments such that segment~$M_i$ corresponds to subtree~$T_i$ for all $i\in[\ell]$. We claim that if $|V(T_i)\cap V(T_j)|=1$ for some $i\neq j$ and $u\in V(T_i)\cap V(T_j)$, then $u\in B_\alpha$. 

Suppose that $u\not\in B_\alpha=B_\beta=B_\gamma$. Since the trees~$T_1,\ldots, T_\ell$ are maximal subtrees of tree~$T^\alpha$ restricted to $E_\beta$ and $E_\gamma$, one of the trees~$T_i$ or $T_j$ is a tree in~$G_\beta$, and the other is a tree in~$G_\gamma$. Therefore, vertex~$u$ is incident with an edge in~$E_\beta$ and an edge in~$E_\gamma$. Thus, vertex~$u$ appears in the subtree rooted at node $\beta$ and in the subtree rooted at node $\gamma$. This is a contradiction to the fact that $\mathbb{T}$ is a tree decomposition, and hence, $u\in B_\alpha=B_\beta=B_\gamma$.

Moreover, if $|V(T_i)\cap V(T_j)|=1$ for some $i\neq j$ and $u\in V(T_i)\cap V(T_j)$, then $u\in M_i$ and $u\in M_j$. If there is a $j\in[\ell]$ such that $v,w\in M_j$, then we are done. Thus, let $v\in M_{j_1}$ and $w\in M_{j_2}$ with $j_1,j_2\in[\ell]$, $j_1\neq j_2$. Then there exists a subset $S_1,\ldots, S_{\ell'}$ of the trees $T_1,\ldots,T_\ell$ with $\ell'\leq \ell$, $S_1=T_{j_1}$, $S_{\ell'} = T_{j_2}$ and $|V(S_i)\cap V(S_{i+1})|=1$ for all $i=1,\ldots,\ell'-1$.   Let $M_{S_1},\ldots,M_{S_{\ell'}}$ be the corresponding segments to $S_1,\ldots, S_{\ell'}$. Then, $|M_{S_i}\cap M_{S_{i+1}}|=1$ for all $i=1,\ldots,\ell'-1$, $v\in M_{S_1}$ and $w\in M_{S_{\ell'}}$, and hence, direction ``$\Rightarrow$'' of condition~(ii) is proven.
For condition~(ii), direction ``$\Leftarrow$'', we consider~$q\in[p]$, $\ell\geq 1$ and $M_1,\ldots,M_\ell\in \mathcal{Y}^\beta_q\cup\mathcal{Y}^\gamma_q$ with $|M_i\cap M_{i+1}|=1$ for all $i=1,\ldots,\ell-1$, $v\in M_1$ and $w\in M_\ell$. We show that there exists a segment~$M^\alpha\in \mathcal{Y}^\alpha_q$ with $v,w\in M^\alpha$. Let $T_1,\ldots, T_\ell$ be trees in the forests with index $q$ in $L_\beta$ and $L_\gamma$ such that tree~$T_i$ corresponds to segment~$M_i$ for all $i\in[\ell]$. Since $|M_i\cap M_{i+1}|=1$ for all $i=1,\ldots,\ell-1$, it follows that $|V(T_i)\cap V(T_{i+1})|=1$ for all~$i=1,\ldots,\ell-1$. Therefore, $T_1,\ldots,T_\ell$ are subtrees of a tree~$T^\alpha$ in the forest with index~$q$ in $L_\alpha$ with~$v,w\in V(T^\alpha)$. Let~$M^\alpha$ be the segment corresponding to~$T^\alpha$. Then, segment~$M^\alpha$ contains the vertices~$v$ and~$w$, i.e.\ $v,w\in M^\alpha$, and hence, direction ``$\Leftarrow$'' of condition~(ii) is proven.

Suppose that there exists a~$q\in[p]$ such that condition~(iv) does not hold for~$q\in[p]$. 
Then there exist a vertex~$v\in B_\alpha$, an integer~$\ell\geq 3$ and segments~$M_1,\ldots,M_\ell\in \mathcal{Y}^\beta_q\cup\mathcal{Y}^\gamma_q$ with~$|M_i\cap M_{i+1}|=1$ for all~$i=1,\ldots,\ell-1$ and~$M_i\neq M_j$ for all~$i\neq j$, such that $v\in M_1$ and $v\in M_\ell$. Let $T_1,\ldots,T_\ell$ be the trees in the forests with index~$q$ in~$L_\beta$ and~$L_\gamma$ such that tree~$T_i$ corresponds to segment~$M_i$ for all $i\in[\ell]$. Note that~$|V(T_i)\cap V(T_{i+1})|=1$ for all~$i=1,\ldots,\ell-1$, and vertex~$v$ appears in the trees~$T_1$ and~$T_\ell$.
For all~$i=1,\ldots,\ell-1$, let~$w_i$ be the vertex in the intersection~$V(T_i)\cap V(T_{i+1})$ of the vertex sets of the trees~$T_i$ and~$T_{i+1}$. 
By construction, the union of the trees~$T':=T_1\cup\ldots\cup T_\ell$ is a subtree of a tree~$T^\alpha$ in the forest with index~$q$ in~$L_\alpha$. 
Thus, the tuple~$(v,w_1,w_2,\ldots,w_{\ell-1},v)$ represents a cycle in~$T'$, and thus, in~$T^\alpha$. 
This is a contradiction to the fact that~$L_\alpha$ is a partial solution for~$G_\alpha$, and hence, condition~(iv) holds.

We conclude that the pair of signatures~$\mathcal{X}^\beta$ and~$\mathcal{X}^\gamma$ is compatible with~$\mathcal{X}^\alpha$. Since $E_\beta\cap E_\gamma=\emptyset$, the number of edges that appear in at least two forests in~$L_\alpha$ is the sum of the number of edges that appear in at least two forests in~$L_\beta$ and the number of edges that appear in at least two forests in~$L_\gamma$. It follows that
\begin{align*}
T[\alpha,\mathcal{X}^\alpha] &= c(L_\alpha) = c(L_\beta) + c(L_\gamma) \geq T[\beta,\mathcal{X}^\beta] + T[\gamma,\mathcal{X}^\gamma] \\
&\geq \min_{\text{$(\mathcal{X'}^\beta,\mathcal{X'}^\gamma)$ compatible with $\mathcal{X}^\alpha$}} \left(T[\beta,\mathcal{X'}^\beta] + T[\gamma,\mathcal{X'}^\gamma]\right) .
\end{align*}

``$\leq$'':  Let~$L_\beta$ and~$L_\gamma$ be partial solutions for~$G_\beta$ and~$G_\gamma$ with signatures~$\mathcal{X}^\beta$ and~$\mathcal{X}^\gamma$, as pair compatible with signature~$\mathcal{X}^\alpha$ for node~$\alpha$, such that~$T[\beta, \mathcal{X}^\beta]=c(L_\beta)$, $T[\gamma, \mathcal{X}^\gamma]=c(L_\gamma)$ and $T[\beta, \mathcal{X}^\beta]+T[\gamma, \mathcal{X}^\gamma] = \min_{\text{$(\mathcal{X'}^\beta,\mathcal{X'}^\gamma)$ compatible with $\mathcal{X}^\alpha$}} (T[\beta,\mathcal{X'}^\beta]+T[\gamma,\mathcal{X'}^\gamma])$. We construct a partial solution~$L_\alpha$ for~$G_\alpha$ with signature~$\mathcal{X}^\alpha$. We claim that for each~$q\in[p]$, the union of the forests with index~$q$ in~$L_\beta$ and~$L_\gamma$ yields a forest in~$G_\alpha$, that induces the segmentation~$(\mathcal{Y}^\alpha_q,Z^\alpha_q)$ in signature~$\mathcal{X}^\alpha$. 
Let~$B:=B_\alpha$. We remark that~$B_\alpha=B_\beta=B_\gamma$ since~$\alpha$ is a join node in~$\mathbb{T}$. We claim that the intersection of the vertex sets of~$G_\beta$ and~$G_\gamma$ are only the vertices in~$B$, that is~$V_\beta\cap V_\gamma=B$. Suppose that there is a vertex~$v\in(V_\beta\cap V_\gamma)\backslash B$. Then the graph induced by the node set~$\{\rho\in V(T_\mathbb{T})\mid v\in B_\rho\}$ is not connected. This contradicts the fact that~$\mathbb{T}$ is a tree decomposition, and thus, $V_\beta\cap V_\gamma = B$. 

Recall that for each~$q\in[p]$, the zero-segments are equal in all three segmentations, that is,~$Z^\alpha_q=Z^\beta_q=Z^\gamma_q$. Hence, the vertex sets in both forests with index~$q$ in~$L_\beta$ and~$L_\gamma$ are the same. In addition, we know that~$E_\beta\cap E_\gamma=\emptyset$ and therefore, the two forests with index~$q$ in~$L_\beta$ and~$L_\gamma$ do not have any edge in common. 
We need to show that for all~$q\in[p]$ the union of the forests with index~$q$ in~$L_\beta$ and~$L_\gamma$ does not contain a cycle in~$G_\alpha$. Suppose there is a~$q\in[p]$ such that the union of the forests with index~$q$ in~$L_\beta$ and~$L_\gamma$ contains a cycle in~$G_\alpha$. 

\emph{Case 1}: There is a tree~$T_1$ in the forest with index~$q$ in~$L_\beta$ and a tree~$T_2$ in the forest with index~$q$ in~$L_\gamma$, such that the union~$T_0:=T_1\cup T_2$ contains a cycle. Let~$M_1\in \mathcal{Y}^\beta_q$ and~$M_2\in \mathcal{Y}^\gamma_q$, such that segment~$M_1$ corresponds to tree~$T_1$ and segment~$M_2$ corresponds to tree~$T_2$. Since graph~$T_0$ contains a cycle in~$G_\alpha$, the trees~$T_1$ and~$T_2$ have at least two vertices in common. Because of~$V(T_1)\subseteq V_\beta$,~$V(T_2)\subseteq V_\gamma$ and~$V_\beta\cap V_\gamma = B$, the common vertices are in the vertex set~$B$. This means that there are two vertices~$v,w\in B$ such that~$v,w\in M_1$ and~$v,w\in M_2$. This contradicts condition (iii), and hence, there are no two trees in the forests with index~$q$ in~$L_\beta$ and~$L_\gamma$ such that their union contains a cycle in~$G_\alpha$. 

\emph{Case 2}: There are trees~$T_1,\ldots, T_{\ell}$,~$\ell\geq 3$, in the forests with index~$q$ in~$L_\alpha$ and~$L_\beta$, such that their union~$T_0:=T_1\cup\ldots\cup T_{\ell}$ contains a cycle in~$G_\alpha$ and $T_0\backslash T_i$ does not contain a cycle in~$G_\alpha$ for all~$i\in[\ell]$. It follows that $|V(T_i)\cap V(T_j)|\leq 1$ for all~$i,j\in[\ell]$ with~$i\neq j$. Let~$M_1,\ldots, M_{\ell}\in \mathcal{Y}^\beta_q\cup \mathcal{Y}^\gamma_q$, such that segment~$M_i$ corresponds to tree~$T_i$ for all~$i\in[\ell]$. Since~$T_0$ contains a cycle in~$G_\alpha$, there exists an ordering~$\pi$ on the set~$[\ell]$, such that~$|V(T_{\pi(i)})\cap V(T_{\pi(i+1)})|=1$ for all~$i=1,\ldots,\ell'-1$ and~$|V(T_{\pi(\ell)})\cap V(T_{\pi(1)})|=1$. Since~$V(T_i)\cap V(T_j)\subseteq B$ for all~$i,j\in[\ell]$ with~$i\neq j$, it follows that $|M_{\pi(i)}\cap M_{\pi(i+1)}|=1$ for all~$i=1,\ldots,\ell-1$. Let $v$ be the vertex such that~$\{v\} = V(T_{\pi(1)})\cap V(T_{\pi(\ell)})$. Since~$V(T_{\pi(1)})\cap V(T_{\pi(\ell)})\subseteq B$, the segments~$M_{\pi(1)}$ and~$M_{\pi(\ell)}$ contain vertex~$v$. Altogether, this contradicts condition~(iv), and hence, there are no trees~$T_1,\ldots, T_{\ell}$,~$\ell\geq 3$, in the forests with index~$q$ in~$G_\alpha$ and~$G_\beta$ such that their union~$T_0=T_1\cup\ldots\cup T_{\ell}$ contains a cycle in~$G_\alpha$.
We conclude that there are no two forests with index~$q$ in~$L_\beta$ and~$L_\gamma$, such that their union contains a cycle, and thus, $L_\alpha$ is a partial solution for~$G_\alpha$. Moreover, by condition~(ii), $L_\alpha$ induces signature~$\mathcal{X}^\alpha$. It follows that
\begin{align*}
\min\limits_{\text{$(\mathcal{X'}^\beta,\mathcal{X'}^\gamma)$ compatible with $\mathcal{X}^\alpha$}} (T[\beta,\mathcal{X'}^\beta]+T[\gamma,\mathcal{X'}^\gamma]) &= T[\beta,\mathcal{X}^\beta] + T[\gamma,\mathcal{X}^\gamma] = c(L_\beta) + c(L_\gamma)  \\
& = c(L_\alpha) \geq T[\alpha,\mathcal{X}^\alpha].
\end{align*}

\vspace{3pt}\emph{Running time}. For each signature~$\mathcal{X}^\alpha$, we check all pairs of signatures~$\mathcal{X}^\beta$,~$\mathcal{X}^\gamma$ for node~$\beta$ and~$\gamma$ for compatibility, that means we check conditions~(i)-(iv)  for~$O((|B_\beta|+1)^{p\cdot|B_\beta|}\cdot(|B_\gamma|+1)^{p\cdot|B_\gamma|})$ pairs of signatures with respect to the signature~$\mathcal{X}^\alpha$. Let~$B:=B_\alpha$. Recall that~$B_\alpha=B_\beta=B_\gamma$.

For each pair, we can check condition~(i) in~$O(p\cdot |B|^3)$~time. We can check conditions~(ii)-(iv) in~$O(p\cdot |B|^3)$~time as follows. 

For each~$q\in[p]$, we construct a graph~$\hat{G}_q$ in the following way. We set~$V(\hat{G}_q):=\{v_i\mid M_i\in \mathcal{Y}^\beta_q\cup \mathcal{Y}^\gamma_q\}$ and~$E(\hat{G}_q):=\{\{v_i,v_j\}\in V(\hat{G}_q)^2\mid |M_i\cap M_j|=1,\, M_i,M_j\in \mathcal{Y}^\beta_q\cup \mathcal{Y}^\gamma_q\}$. We can construct the graph~$\hat{G}_q$ in~$O(|B|^3)$~time. We can check condition~(iii) while constructing graph~$\hat{G}_q$. If condition~(iv) does not hold, then there exists a cycle in~$\hat{G}_q$. We can detect a cycle in~$\hat{G}_q$ in~$O(|B|^2)$~time, for example by applying a depth-first search on~$\hat{G}_q$, and thus, we can check condition~(iv) in~$O(|B|^2)$~time. 

For condition~(ii), we compare the corresponding segments of the vertex sets of the connected components in~$\hat{G}_q$ with the segments in~$\mathcal{Y}^\alpha_q$. Finding the connected components in~$\hat{G}_q$ can be done in~$O(|B|^2)$~time, for example by applying a depth-first search in~$\hat{G}_q$. The comparison of the segments can be done in~$O(|B|^2)$~time. Thus, condition~(ii) can be verified in~$O(|B|^2)$~time. We conclude that for each~$q\in[p]$, we can check conditions~(ii)-(iv) in~$O(|B|^3)$~time.

We can check conditions~(i)-(iv) for each pair of signatures for node~$\beta$ and node~$\gamma$ in~$O(p\cdot|B|^3)$~time. Therefore, the overall running time for filling all entries in~$T$ for a join node is in~$O(p\cdot (\omega+2)^{3\cdot p\cdot (\omega+1)+3})$.
\fi{}
\medskip

\ifshort{}
\looseness=-1 The bottleneck in computing the tables is in the join nodes; they induce a running time portion of $O(p\cdot (\omega+2)^{3\cdot p\cdot (\omega+1)+3})$. Hence, filling the tables for each node in the tree decomposition can be done in the running time claimed by \cref{thm:twdp}. By the above arguments about partial solutions, the minimum number of shared edges in a $(p, s, t)$-routing can then be read off from the table in the root node of the tree decomposition, where we take the minimum value over all signatures where for each of the $p$~segmentations there exists a segment that contains the vertices~$s$ and~$t$. Hence, \cref{thm:twdp} follows.
\else{}
Now we describe how to fill the entries in the table~$T$ of the dynamic program according to each type of nodes in the tree decomposition~$\mathbb{T}$. %

\begin{proof}[Proof of \Cref{thm:twdp}]
Let $G$ be graph with $s,t\in V(G)$ given together with a tree decomposition $\mathbb{T'}=(T',(B_\alpha')_{\alpha\in V(T')})$ of width $\omega':=\omega(\mathbb{T'})$ of~$G$. We modify the tree decomposition $\mathbb{T}'$ in polynomial time to a nice tree decomposition with introduce edge nodes of equal width, and add the vertices~$s$ and $t$ to every bag. Let $\mathbb{T}$ be the nice tree decomposition with introduce edge nodes and vertices~$s$ and $t$ contained in every bag obtained from $\mathbb{T'}$. Note that $\omega:=\omega(\mathbb{T}) \leq \omega'+2$. We apply the dynamic program described above bottom-up on the tree decomposition~$\mathbb{T}$. The dynamic program runs in~$O(p\cdot (\omega+2)^{3\cdot p\cdot (\omega+1)+4}\cdot n)$~time. Since $\omega\leq \omega'+2$, it follows that the dynamic program runs in~$O(p\cdot (\omega'+4)^{3\cdot p\cdot (\omega'+3)+4}\cdot |V(G)|)$~time. Finally, we read out the minimum number of shared edges for~$p$~$s$-$t$ routes in the entries of the root node in~$\mathbb{T}$ as follows. 
Let~$\tau$ be the root node of~$\mathbb{T}$. Note that~$\{s,t\}\subseteq B_\tau$. Let~$\mathcal{F}$ be the set of all signatures for node~$\tau$ such that for all signatures~$\mathcal{X}^\tau=(\mathcal{Y}^\tau_q,Z^\tau_q)_{q=1,\ldots,p}$ in~$\mathcal{F}$ it holds that for all~$q\in[p]$ there exists a segment~$M\in \mathcal{Y}^\tau_q$ with~$\{s,t\}\subseteq M$. Due to our construction, a segment of a segmentation corresponds to a tree in a partial solution for the given graph. Hence, a set of~$p$ segmentations, where for each of the $p$~segmentations there exists a segment that contains the vertices~$s$ and~$t$, corresponds to a solution for \msetsc{} with~$p$~routes. Thus, the minimum number of shared edges for~$p$~$s$-$t$~routes equals~$\min_{\mathcal{X}^\tau\in \mathcal{F}} T[\tau,\mathcal{X}^\tau]$.
\end{proof}

\fi{}

We remark that we can modify the dynamic program in such a way that we can solve the weighted variant of \msetsc{}, that is, with weights~$w:E(G)\to \mathbb{N}$ on the edge set of the input graph. \iflong{}The cost of the partial solutions is the sum of the weights of shared edges, and thus the entry in the table of the dynamic program. For an introduce edge node, in the case of share-compability, we increase the value of the entry by the weight of the introduced edge. More precisely, for an introduce edge node $\alpha$ that introduces edge~$e$ and a signature~$\mathcal{X}^\alpha$ for node~$\alpha$, the filling rule is adjusted by
\[
 T[\alpha,\mathcal{X}^\alpha] = \min \left(T[\beta,\mathcal{X}^\beta] + 
 \begin{cases}
 w(e) ,& \text{if $\mathcal{X}^\beta$ and $\mathcal{X}^\alpha$ are share-compatible,} \\
 0 ,& \text{otherwise} 
 \end{cases}\right),
\] 
where the minimum is taken over all signatures $\mathcal{X}^\beta$ for node~$\beta$ compatible with $\mathcal{X}^\alpha$.\fi{}

%


\section{Fixed-Parameter Tractability with Respect to the Number of Routes}\label{sec:twred}
In this section we prove the following.
\begin{theorem}\label{theorem!fptwrtp}
\msetsc{} is fixed-parameter tractable with respect to the number~$p$ of routes. 
\end{theorem}

\looseness=-1 The basic idea for the proof is to use treewidth reduction~\cite{MarxOR13}, a way to process a graph~$G$ containing terminals $s, t$ in such a way that each minimal $s$-$t$ separator of size at most~$p - 1$ is preserved and the treewidth of the resulting graph is bounded by a function of~$p$. The reason that this approach works is (we prove below) that each $(p,s,t)$-routing is characterized by its shared edges, and these are contained in minimal cuts of size at most $p - 1$. However, treewidth reduction preserves only minimal separators, that is, vertex sets, and not necessarily minimal cuts, that is, edge sets. Hence, we need to further process input graph and the graph coming out of the treewidth reduction process. 

We now describe the approach in more detail; refer to \Cref{fig:pftpoverview} for an overview of the following modifications and the graphs obtained in each step.%
\iflong{} Let $(G,s,t,p,k)$ be an instance of \MSE{}, where~$G$ is the input graph with $s,t\in V(G)$. First, we obtain a graph~$H$ by subdividing each edge in $G$. We denote by~$V_E$ the set of vertices obtained from the subdivisions. As a consequence, every minimal $s$-$t$~cut in $G$ of size at most $p-1$ corresponds to a minimal $s$-$t$~separator in $H$ of size at most $p-1$. Next, we apply the treewidth reduction technique to~$H$, obtaining the graph~$H^*$. By the treewidth reduction technique, graph~$H^*$ contains all minimal $s$-$t$~separators in~$H$ of size at most $p-1$ and the treewidth of graph~$H^*$ is upper-bounded by a function only depending on $p$. We denote by $V_E^*:=\{v\in V(H^*)\mid v\in V_E\}$ the set of vertices in~$V_E$ which are preserved by the treewidth reduction technique in $H^*$. Finally, we contract an incident edge for each vertex in $V_E^*\subseteq V(H^*)$ to obtain the graph~$G^*$.%
\todo[inline]{rn:This is too vague. What did we discuss where?}\fi{}
\begin{figure}[!t]%
\centering
\begin{tikzpicture}%

\node (G) at (0-0.5,0) []{$G$};
\node (H) at (4-0.25,0) []{$H$};
\node (H*) at (7,0) []{$H^*$};
\node (G*) at (11+0.5,0) []{$G^*$};

\draw[color=black] (-1,-1) rectangle (12,1);

\draw[->,>=stealth] (G) to (H); 
\node (fA) at (2-0.25,0) [label=90:{Subdivide each}, label=270:{edge in~$G$}]{};
\draw[->,>=stealth] (H)  -- (H*);
\node (sA) at (5.5-0.25,0) [label=90:{Treewidth}, label=270:{Reduction}]{};
\draw[->,>=stealth] (H*)  -- (G*) ;
\node (tA) at (9+0.25,0) [label=90:{Contract an incident}, label=270:{edge for each $v\in V_E^*$}]{};

\node[align=center] (whyfA) at (1,2.5) {Each minimal~$s$-$t$ cut of\\  size at most~$p-1$\\ corresponds to a\\ minimal~$s$-$t$ separator \\ of size at most~$p-1$.};
\draw[->,>=stealth, dashed, very thin, shorten >=6mm] (whyfA) to (fA);

\node[align=center] (whysA) at (5.75,2.5) {Constructs a graph of\\ treewidth bounded by\\ a function in~$p$ that preserves\\ all minimal~$s$-$t$ separators\\ of size at most~$p-1$.};
\draw[->,>=stealth, dashed, very thin, shorten >=6mm] (whysA) to (sA);

\node[align=center] (whytA) at (10.25,2.5) {Yields 1-to-1\\correspondence between \\ minimal~$s$-$t$ cuts\\ of size at most~$p-1$\\ in~$G$ and~$G^*$.};
\draw[->,>=stealth, dashed, very thin, shorten >=6mm] (whytA) to (tA);

\node[align=center] (propH) at (0.75,-3.25) {Includes the vertex set~$V_E$,\\ the vertices corresponding \\ to the subdivisions.\\ Each minimal $s$-$t$~cut in~$G$ \\ of size at most $p-1$ \\ corresponds to a \\ minimal $s$-$t$~separator in~$H$ \\ of size at most~$p-1$.};
\draw[->,>=stealth, dashed, very thin] (propH) to (H);

\node[align=center] (propH*) at (5.75,-3.25) {Has treewidth bounded\\ by a function in~$p$, contains\\ every minimal~$s$-$t$ separator\\ of size at most~$p-1$ in~$H$ and\\ contains the neighborhood of \\ every vertex in $V_E$ which is in \\ a minimal~$s$-$t$ separator\\ of size at most~$p-1$ in~$H$.};
\draw[->,>=stealth, dashed, very thin] (propH*) to (H*);

\node[align=center] (propG*) at (10.25,-3.25) {Has treewidth bounded\\ by the treewidth of~$H^*$.\\ An edge set \\ $C\subseteq E(G)\cap E(G^*)$ \\ with $|C|<p$ is a\\ minimal~$s$-$t$ cut in $G^*$\\ if and only if it is a\\ minimal~$s$-$t$ cut in $G$.};
\draw[->,>=stealth, dashed, very thin] (propG*) to (G*);

\end{tikzpicture}
\caption{Overview of the strategy behind the proof of~\Cref{theorem!fptwrtp}.}
\label{fig:pftpoverview}
\end{figure}
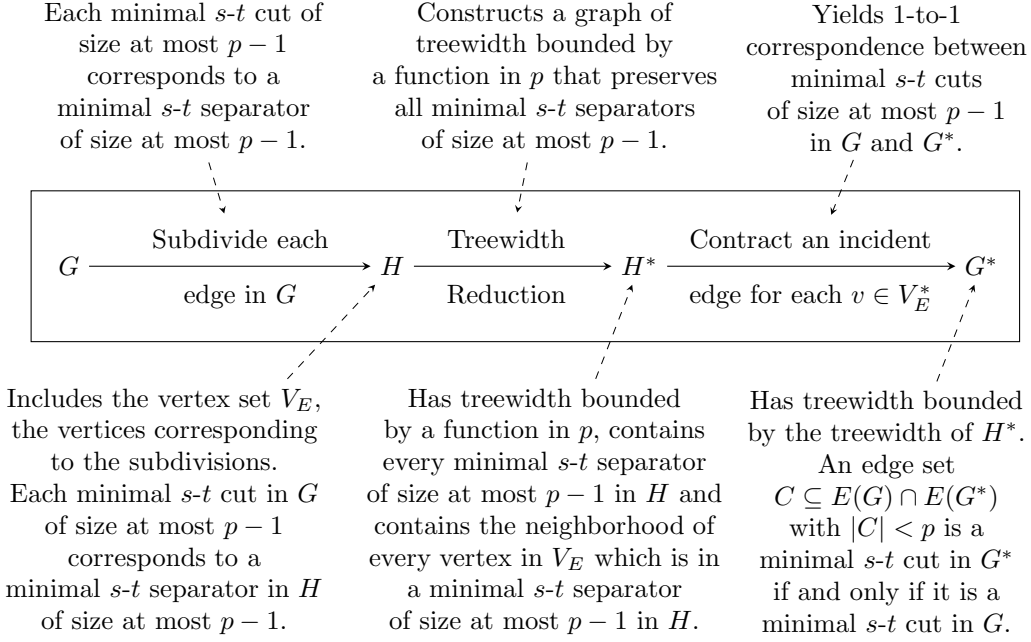%
In the following, we modify step by step graph~$G$ to graph~$G^*$.\iflong{} We discuss each step and we prove the properties of the obtained graphs described above. Finally, we give a proof of~\Cref{theorem!fptwrtp}.\fi{}
We start with the following lemma which states that if our instance is a yes-instance, then we can find a solution where each of the shared edges is part of a minimal~$s$-$t$~cut of size smaller than the number~$p$ of routes.

\begin{lem}\label{lemma!yesinstancethenedgesetminimalcut}
If~$(G,s,t,p,k)$ is a yes-instance of \MSE{} and~$G$ has a minimal~$s$-$t$ cut of size smaller than~$p$, then there exists a solution~$F\subseteq E$ such that each~$e\in F$ is in a minimal~$s$-$t$ cut of size smaller than~$p$ in~$G$.
\end{lem}

Recall that if~$G$ does not have a minimal~$s$-$t$ cut of size smaller than~$p$, then we can find~$p$~$s$-$t$ routes without sharing an edge. \iflong{}In the following proof, we make use of the following equivalent formulation of \MSE{} based on edge contractions. Given an undirected graph $G=(V,E)$, $s,t\in V(G)$, $p\in\mathbb{N}$, and $k\in \mathbb{N}_0$, the question is whether there is a subset~$F\subseteq E$ of edges of cardinality at most~$k$ in~$G$ such that the graph~$G/F$ with unit edge capacities allows an $s$-$t$~flow of value at least~$p$. In the following, we call such a set $F$ a \emph{solution}. Using Menger's theorem, one can obtain with small effort the equivalence of \MSE{} and the problem above. 

\begin{proof}[Proof of \Cref{lemma!yesinstancethenedgesetminimalcut}]

We make use of the contraction equivalent of \MSE{}. We show that for every minimal solution for \MSE{} it holds that each edge of the solution is part of a minimal~$s$-$t$ cut of size smaller than~$p$, where a solution is minimal if it is not a superset of another solution.

Let $G=(V,E)$ be the graph. Let~$(G,s,t,p,k)$ be a yes-instance of \MSE{}. Then there exists a solution~$L\subseteq E$,~$|L|\leq k$, such that graph~$G_L:=G/L$ with unit edge capacities allows a maximum~$s$-$t$~flow of value at least~$p$. We call a solution~$L$ minimal if there is no edge~$e\in L$ such that graph~$G/(L\backslash \{e\})$ with unit edge capacities allows a maximum~$s$-$t$~flow of value at least~$p$.

Let~$L$ be a minimal solution and let~$e\in L$. Suppose that~$e$ is not part of a minimal~$s$-$t$ cut of size smaller than~$p$ in~$G$. Let~$L':=L\backslash\{e\}$ and $G_{L'}:=G/L'$. We consider the following two cases.

\emph{Case 1:} The maximum~$s$-$t$~flow of~$G_{L'}$ has value smaller than~$p$. Then, using the max-flow min-cut theorem,~$G_{L'}$ has an~$s$-$t$ cut~$C$ of size smaller than~$p$. Since~$e\not\in C$, contracting edge~$e$ in~$G_{L'}$ does not affect cut~$C$. Therefore,~$C$ is also an~$s$-$t$ cut of size smaller than~$p$ in~$G_L$ and, again by the max-flow min-cut theorem, this implies a maximum $s$-$t$~flow of value smaller than~$p$ in~$G_L$. This is a contradiction to the fact that~$L$ is a solution.

\emph{Case 2:} The maximum~$s$-$t$~flow of~$G_{L'}$ has value at least~$p$. Then~$L'$ is a solution, which contradicts the minimality of~$L$.

Since~$|L|\leq k$ and each edge in~$L$ is in a minimal~$s$-$t$ cut of size smaller than~$p$ in~$G$, this completes the proof.
\end{proof}
\fi{}

\looseness=-1 As mentioned before, as part of our approach we use the treewidth reduction technique~\cite{MarxOR13}. Given a graph~$G=(V,E)$ with~$T=\{s,t\}\subseteq V(G)$ and an integer~$\ell\in \mathbb{N}$, first the treewidth reduction technique computes the set~$C$ of vertices containing all vertices in~$G$ which are part of a minimal~$s$-$t$~separator of size at most~$\ell$ in~$G$. Then, it constructs the so-called \emph{torso} of graph~$G$ given~$C$ and~$T$, that is, the induced subgraph~$G[C\cup T]$ with additional edges between each pair of vertices~$v,w\in C\cup T$ with~$\{v,w\}\not\in E(G)$ if there is a~$v$-$w$~path in~$G$ whose internal vertices are not contained  in~$C\cup T$. Finally, each of these additional edges is subdivided and~$\ell$~additional copies of each of that subdivisions are introduced, that is, if~$\{v,w\}$ is one of these additional edges, then the vertices~$x^{vw}_1,\ldots,x^{vw}_{\ell+1}$ are added and edge~$\{v,w\}$ is replaced by the edges~$\{v,x^{vw}_1\},\ldots,\{v,x^{vw}_{\ell+1}\},\{x^{vw}_1,w\},\ldots,\{x^{vw}_{\ell+1},w\}$. In the following, we denote these paths by \emph{copy paths}. The resulting graph contains all minimal~$s$-$t$~separators of size at most~$\ell$ in~$G$ and has treewidth upper-bounded by~$h(\ell)$ for some function~$h$ only depending on~$\ell$.

\begin{theorem}[Treewidth reduction {\cite[Theorem~2.15]{MarxOR13}}]\label{theorem!twrt}
Let~$G$ be a graph, $T\subseteq V(G)$, and let~$\ell$ be an integer. Let~$C$ be the set of all vertices of~$G$ participating in a minimal~$s$-$t$ separator of size at most~$\ell$ for some $s,t\in T$. For every fixed~$\ell$ and~$|T|$, there is a linear-time algorithm that computes a graph $G^*$ having the following properties:
\begin{compactenum}[(1)]
\item $C\cup T\subseteq V(G^*)$.
\item For every~$s,t\in T$, a set~$L\subseteq V(G^*)$ with~$|L|\leq \ell$ is a minimal~$s$-$t$ separator of $G^*$ if and only if~$L\subseteq C\cup T$ and~$L$ is a minimal~$s$-$t$ separator of~$G$.
\item The treewidth of~$G^*$ is at most~$h(\ell,|T|)$ for some function~$h$.
\item $G^*[C\cup T]$ is isomorphic to $G[C\cup T]$.
\end{compactenum}
\end{theorem}

\iflong{}
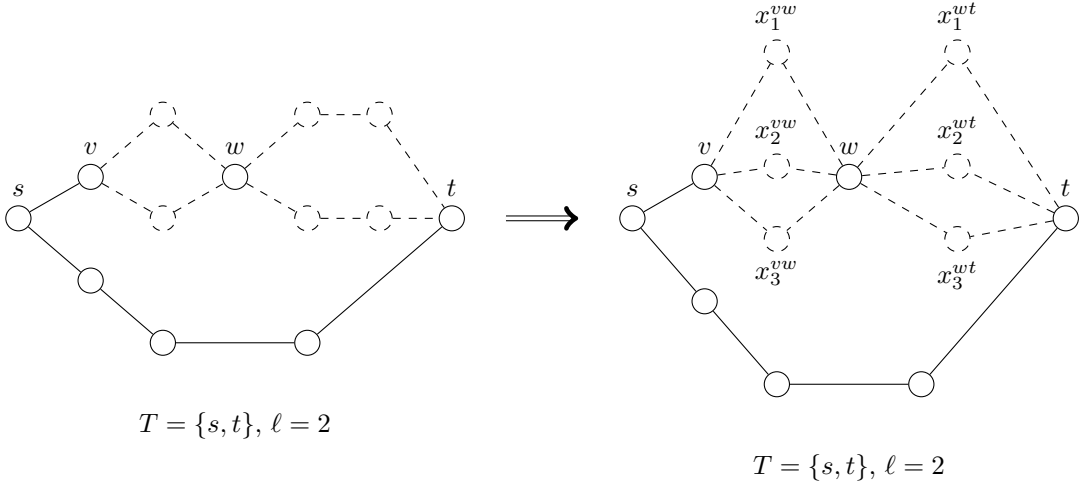
\begin{figure}[!t]
\begin{tikzpicture}[x=0.475cm, y=0.55cm]%

\node (s) at (-4,0) [shape=circle, label=90:{$s$}, draw]{};
\node (t) at (8,0) [shape=circle, label=90:{$t$}, draw]{};

\node (v1) at (-2,1) [shape=circle, label=90:{$v$}, draw]{};
\node (v2) at (-2,-1.5) [shape=circle, draw]{};
\node (v3) at (0,2.5) [shape=circle, dashed, draw]{};
\node (v4) at (0,0) [shape=circle, dashed, draw]{};
\node (v5) at (0,-3) [shape=circle, draw]{};
\node (v6) at (2,1) [shape=circle, label=90:{$w$}, draw]{};
\node (v7) at (4,2.5) [shape=circle, dashed, draw]{};
\node (v8) at (4,0) [shape=circle, dashed, draw]{};
\node (v9) at (4,-3) [shape=circle, draw]{};
\node (v10) at (6,2.5) [shape=circle, dashed, draw]{};
\node (v11) at (6,0) [shape=circle, dashed, draw]{};

\node[align=center] (Gstar) at (2,-5) {$T=\{s,t\}$,~$\ell=2$};

\draw (s) -- (v1);
\draw[dashed] (v1) -- (v3);
\draw[dashed] (v1) -- (v4);
\draw[dashed] (v3) -- (v6);
\draw[dashed] (v4) -- (v6);
\draw[dashed] (v6) -- (v7);
\draw[dashed] (v6) -- (v8);
\draw[dashed] (v7) -- (v10);
\draw[dashed] (v8) -- (v11);
\draw[dashed] (v10) -- (t);
\draw[dashed] (v11) -- (t);

\draw (s) -- (v2);
\draw (v2) -- (v5);
\draw (v5) -- (v9);
\draw (v9) -- (t);

\draw[double distance=1.4pt, ->] (9.5,0) to (11.5,0);

\node (s) at (-4+17,0) [shape=circle, label=90:{$s$}, draw]{};
\node (t) at (8+17,0) [shape=circle, label=90:{$t$}, draw]{};

\node (v1) at (-2+17,1) [shape=circle, label=90:{$v$}, draw]{};
\node (v2) at (-2+17,-2) [shape=circle, draw]{};
\node (x26_1) at (0+17,4) [shape=circle, dashed, label=90:{$x^{vw}_1$}, draw]{};
\node (x26_2) at (0+17,1.25) [shape=circle, dashed, label=90:{$x^{vw}_2$}, draw]{};
\node (x26_3) at (0+17,-0.5) [shape=circle, dashed, label=270:{$x^{vw}_3$}, draw]{};
\node (v5) at (0+17,-4) [shape=circle, draw]{};
\node (v6) at (2+17,1) [shape=circle, label=90:{$w$}, draw]{};
\node (x6t_1) at (5+17,4) [shape=circle, dashed, label=90:{$x^{wt}_1$}, draw]{};
\node (x6t_2) at (5+17,1.25) [shape=circle, dashed, label=90:{$x^{wt}_2$}, draw]{};
\node (x6t_3) at (5+17,-0.5) [shape=circle, dashed, label=270:{$x^{wt}_3$}, draw]{};
\node (v9) at (4+17,-4) [shape=circle, draw]{};

\node[align=center] (Gstar) at (2+17,-6) {$T=\{s,t\}$,~$\ell=2$};

\draw (s) -- (v1);
\draw[dashed] (v1) -- (x26_1);
\draw[dashed] (v1) -- (x26_2);
\draw[dashed] (v1) -- (x26_3);
\draw[dashed] (v6) -- (x26_1);
\draw[dashed] (v6) -- (x26_2);
\draw[dashed] (v6) -- (x26_3);
\draw[dashed] (v6) -- (x6t_1);
\draw[dashed] (v6) -- (x6t_2);
\draw[dashed] (v6) -- (x6t_3);
\draw[dashed] (t) -- (x6t_1);
\draw[dashed] (t) -- (x6t_2);
\draw[dashed] (t) -- (x6t_3);

\draw (s) -- (v2);
\draw (v2) -- (v5);
\draw (v5) -- (v9);
\draw (v9) -- (t);

\end{tikzpicture}
\caption{Example for the treewidth reduction technique.}
\label{fig:TWR}
\end{figure}
\fi{}

\iflong{}Figure~\ref{fig:TWR} shows an example for the application of the treewidth reduction technique. We use dashed edges and vertices to highlight the changes when applying the treewidth reduction technique with~$T=\{s,t\}$ and parameter~$\ell=2$. On the left-hand side, the original graph is shown. On the right-hand side, the resulting graph after applying the treewidth reduction technique with~$T=\{s,t\}$ and~$\ell=2$ on the left-hand side graph is shown.
\fi

For finding a $p$-routing we are interested in minimal~$s$-$t$ cuts of size smaller than~$p$ in~$G$. The treewidth reduction technique guarantees to preserve minimal~$s$-$t$~separators of a specific size, but does not guarantee to preserve minimal~$s$-$t$~cuts of a specific size. Thus, we need to modify our graph~$G$ in such a way that each minimal $s$-$t$~cut in $G$ corresponds to a minimal $s$-$t$~separator in the modified graph. We modify graph~$G$ in the following way. 

\begin{step}\label{step!one}
Subdivide each edge in~$E(G)$, that is, for each edge~$e=\{v,w\}$ in~$E(G)$ add a vertex~$x_e$ and replace edge~$e$ by edge~$\{v,x_e\}$ and edge~$\{x_e,w\}$. We say that vertex~$x_e$ as well as edges~$\{v,x_e\}$ and~$\{x_e,w\}$ \emph{correspond} to edge~$e$. Let $V_E:=\{x_e\mid e\in E\}$ and $E'$~be the edge set replacing the edges in~$E$. Then~$H:=(V\cup V_E,E')$ is the resulting graph.
\end{step}

Note that each edge in~$H$ is incident with exactly one vertex in~$V_E$ and one vertex in~$V$. Thus, no two vertices in~$V_E$ and no two vertices in~$V$ are neighbors. Moreover, note that each vertex in~$V_E$ has degree exactly two. It holds that~$|V\cup V_E|=|V|+|E|$ and~$|E'|=2\cdot|E|$. %

\iflong{}
\begin{lem}\label{lemma!modgisyes}
$(G,s,t,p,k)$ is a yes-instance of \MSE{} if and only if~$(H,s,t,p,2k)$ is a yes-instance of \MSE{}.
\end{lem}

\begin{proof}
Intuitively, every edge in~$G$ corresponds to two edges in~$H$ and every two edges in~$H$ both incident with an vertex in~$V_E$ correspond to an edge in~$G$. 

``$\Rightarrow$'': Consider a solution for the yes-instance~$(G,s,t,p,k)$ of~\MSE{}. For each edge~$e=\{v,w\}\in E(G)$ that is shared in the solution, consider the corresponding two edges~$\{v,x_e\}$ and~$\{x_e,w\}$ in graph~$H$. Sharing these at most~$2k$ edges yields a solution for instance $(H,s,t,p,2k)$ of~\MSE{}.

``$\Leftarrow$'': Consider a minimal solution for the yes-instance~$(H,s,t,p,2k)$. Observe that in such a solution, a vertex in~$V_E$ is incident with either no or two shared edges. Each vertex in~$V_E$ that appears in at least two~$s$-$t$ routes is incident with two shared edges. Each vertex in~$V_E$ corresponds to one edge in~$G$. Let~$F\subseteq E(G)$ be the set of edges such that~$e=\{v,w\}\in F$ if the edges~$\{v,x_e\}$ and~$\{x_e,w\}$ in~$E(H)$ are shared in the solution for~$(H,s,t,p,2k)$. Note that $|F|\leq k$ since there are at most $2 k$ shared edges. Thus, $F$~is a solution for instance~$(G,s,t,p,k)$ of~\MSE{}.
\end{proof}
\fi{}

Recall that we are interested in~$s$-$t$ cuts in~$G$. By our modification from~\Cref{step!one} of~$G$ to~$H$, for each edge in~$G$ there is a corresponding vertex in~$V_E$ in~$H$.\iflong{} The following lemma gives a one-to-one correspondence between~$s$-$t$ cuts in~$G$ and those~$s$-$t$ separators in~$H$ that contain only vertices in~$V_E$. 

\begin{lem}\label{lemma!cutisseparator}
If~$C$ is an~$s$-$t$ cut in~$G$, then~$V_C:=\{w\in V_E\mid w$~corresponds~to~$e\in C \}$ is an~$s$-$t$ separator in~$H$. If~$W\subseteq V_E$ is an~$s$-$t$ separator in~$H$, then~$C_W:=\{e\in E\mid e$~corresponds~to~$w\in W\}$ is an~$s$-$t$ cut in~$G$.
\end{lem}

\begin{proof}
Let~$C$ be an~$s$-$t$ cut in~$G$. Suppose that the set~$V_C:=\{w\in V_E\mid w$~corresponds~to~$e\in C \}$ is not an~$s$-$t$~separator in~$H$. Then there exists a path~$P'$ avoiding~$V_C$ in~$H$ connecting~$s$~and~$t$. Since no two vertices in~$V_E$ are neighbors and no two vertices in~$V$ are neighbors, the vertices in path~$P'$ alternate in~$V$ and~$V_E$. Since we know that the vertices in~$V_E$ correspond to edges in~$G$,~$P:=P'\cap V$ describes a path in~$G$ connecting~$s$ and~$t$ avoiding all edges in~$C$. This is a contradiction to the fact that $C$~is an $s$-$t$~cut in~$G$, and hence set~$V_C$ is an $s$-$t$~separator in~$H$.

Let~$W\subseteq V_E$ be an~$s$-$t$ separator in~$H$. Suppose that the set~$C_W:=\{e\in E\mid e$~corresponds~to~$w\in W\}$ is not an~$s$-$t$ cut in~$G$. Then there exists a path~$P$ avoiding~$C_W$ in~$G$ connecting~$s$ and~$t$. Let~$V_P\subseteq V(H)$ be the set of vertices in~$H$ such that each vertex in~$V_P$ either corresponds to an edge in~$P$ or is an endpoint of an edge in~$P$. We remark that~$W\cap V_P=\emptyset$. Moreover, set~$V_P$ is the set of vertices of an~$s$-$t$~route in~$H$. This is a contradiction to the fact that $W$~is an~$s$-$t$ separator in~$H$, and hence set~$C_W$ is an~$s$-$t$~cut in~$G$.
\end{proof}

In the following lemma, we show that \Cref{lemma!cutisseparator} holds also for minimal~$s$-$t$ cuts and minimal~$s$-$t$~separators. This is important, since we will use a combination of the treewidth reduction technique and~\Cref{lemma!yesinstancethenedgesetminimalcut} later on. %
\fi{}\ifshort{} One can show that there is a one-to-one correspondence between~$s$-$t$ cuts in~$G$ and those~$s$-$t$ separators in~$H$ that contain only vertices in~$V_E$. Moreover, the following lemma holds.\fi{}

\begin{lem}\label{lemma!minimalcutcorrespondstominimalsep}
Every minimal~$s$-$t$~cut in~$G$ corresponds to a minimal~$s$-$t$~separator in~$H$.
\end{lem}

\iflong{}
\begin{proof}
Let~$C$ be a minimal~$s$-$t$ cut in~$G$. By \Cref{lemma!cutisseparator}, we know that~$V_C:=\{w\in V_E\mid w$~corresponds~to~$e\in C \}$ is an~$s$-$t$ separator in~$H$. If~$V_C$ is a minimal~$s$-$t$~separator in $H$, then we are done. Thus, suppose that $V_C$~is an $s$-$t$~separator in~$H$, but $V_C$ is not a minimal~$s$-$t$~separator in~$H$. Then there exists a vertex~$w\in V_C$ such that~$V_C\backslash \{w\}$ is an~$s$-$t$ separator in~$H$. Let~$e\in C$ be the edge in~$G$ corresponding to vertex~$w$. Since~$V_C\backslash \{w\}\subseteq V_E$, again by \Cref{lemma!cutisseparator} we know that~$C\backslash \{e\}$ is an~$s$-$t$~cut in~$G$. This is a contradiction to the fact that~$C$ is a minimal~$s$-$t$ cut in~$G$, and hence, $V_C$ is a minimal~$s$-$t$~separator in $H$.
\end{proof}

We know that each minimal $s$-$t$~cut in~$G$ corresponds to a minimal $s$-$t$~separator in~$H$.\fi{} Next, we show that every vertex in the neighborhood of each minimal~$s$-$t$~separator containing only vertices in~$V_E$ belongs to a minimal~$s$-$t$~separator. \todo[inline]{rn: I found this sentence unclear and ambiguous.}\iflong{} Recall that for~$W\subseteq V$ we denote by~$N_G(W)$ the open neighborhood of the vertex set~$W$ in~$G$ and by~$N_G[W]:=W \cup N_G(W)$ the closed neighborhood of the vertex set~$W$ in~$G$.\fi{}

\begin{lem}\label{lemma!eachofneighborsarepartofspearators}
Let~$W\subseteq V_E\subseteq V(H)$ be the set of vertices corresponding to a minimal~$s$\nobreakdash-$t$~cut of size at most~$\ell\in\mathbb{N}$ in~$G$. Then, each vertex in~$N_{H}[W]$ is part of a minimal~{$s$-$t$}~separator of size at most~$\ell$ in~$H$.
\end{lem}

\iflong{}
\begin{proof}
Let~$W\subseteq V_E\subseteq V(H)$ be given such that~$W$ corresponds to a minimal~$s$-$t$ cut in~$G$ of size at most~$\ell$. Note that by~\Cref{lemma!minimalcutcorrespondstominimalsep}, $W$~is a minimal~$s$-$t$~separator in $H$. Let~$x$ be an arbitrary vertex in~$N_H(W)$. First, we show that~$W':=(W\backslash N_H(x))\cup \{x\}$ is an~$s$-$t$~separator in~$H$. 
Suppose that $W'$~is not an $s$-$t$~separator in~$H$. Then there exists an $s$-$t$~path~$P$ in $H-W'$. Note that each vertex in~$W\cap N_H(x)$ is incident with vertex~$x$ and exactly one other vertex in $V(H)$. Thus, no vertex in~$W\cap N_H(x)$ appears in path~$P$. Hence, $P$~is an $s$-$t$~path in~$H-W$. This is a contradiction to the fact that $W$~is an~$s$-$t$~separator in~$H$, and hence, $W'$~is an $s$-$t$~separator in~$H$. 
Next, we show that if $W'$~is not a minimal~$s$-$t$~separator in~$H$, then there exists a set~$U\subseteq W'\backslash \{x\}$ such that $W'\backslash U$ is a minimal~$s$-$t$~separator in~$H$. 
Let $W'$ be an~$s$\nobreakdash-$t$~separator in~$H$, but not a minimal~$s$-$t$~separator in~$H$. Suppose that for all~$U\subseteq W'\backslash \{x\}$ it holds that $W'\backslash U$ is not a minimal~$s$-$t$~separator. Then there exists a set~$X\subseteq W'$ with $x\in X$ such that $W'\backslash X$ is a minimal~$s$-$t$~separator in~$H$. Since $W'\backslash X = W\backslash(N_H(x)\cap W)\backslash X\subseteq W$, this contradicts the fact that $W$~is a minimal~$s$\nobreakdash-$t$~separator in~$H$. Hence, there exists a set~$U\subseteq W'\backslash \{x\}$ such that $W'\backslash U$ is a minimal~$s$-$t$~separator in~$H$. 
Let $U\subseteq W'\backslash \{x\}$ be a set such that $W'':=W'\backslash U$ is a minimal~$s$-$t$~separator. Since $x\in W''$ and $|W''|\leq |W'|\leq |W|$, vertex~$x$ appears in a minimal~$s$-$t$~separator in~$H$ of size at most~$\ell$. Since vertex~$x$ was chosen arbitrarily in~$N_H(W)$, each vertex in~$N_{H}[W]$ is part of a minimal~$s$-$t$~separator of size at most~$\ell$ in~$H$.  
\end{proof}
\fi{}

We obtained graph~$H$ from graph~$G$ by applying \Cref{step!one}. By \Cref{lemma!minimalcutcorrespondstominimalsep}, we know that each minimal~$s$-$t$ cut in~$G$ corresponds to a minimal~$s$-$t$~separator in~$H$. Moreover, by \Cref{lemma!eachofneighborsarepartofspearators}, if we consider a minimal~$s$-$t$~cut of size smaller than~$p$ in~$G$, then, for each neighbor of the vertex set in~$H$ corresponding to the minimal~$s$-$t$ cut in~$G$, there exists a minimal~$s$-$t$ separator of size smaller than~$p$ in~$H$ that contains that neighbor. As the next step (cf. \Cref{fig:pftpoverview}) we apply the treewidth reduction technique~\cite{MarxOR13} on graph~$H$.

\begin{step}\label{step!two} %
Apply the treewidth reduction (\cref{theorem!twrt}) to graph~$H$ with~$T=\{s,t\}$ and~$p-1$ as upper bound for the size of the minimal~$s$-$t$ separators. Denote the resulting graph by~$H^*$. 
\end{step}

Let~$V_E^*:=\{v\in V(H^*)\mid v\in V_E\}$. Graph~$H^*$ contains all minimal~$s$-$t$~separators of size at most~$p-1$ in~$H$. By \Cref{lemma!minimalcutcorrespondstominimalsep}, every minimal~$s$-$t$ cut of size at most~$p-1$ in~$G$ corresponds to a minimal~$s$-$t$ separator of size at most~$p-1$ in~$H$ and thus, by \Cref{theorem!twrt}, to a minimal~$s$-$t$~separator of size at most~$p-1$ in~$H^*$. By \Cref{lemma!eachofneighborsarepartofspearators}, the neighborhood of each vertex in~$H$ corresponding to a vertex in~$V_E^*$ is contained in the vertex set~$V(H^*)$. As a consequence, we can reconstruct each edge in graph~$G$ that appears in a minimal~$s$-$t$~cut of size at most~$p-1$ in~$G$ as an edge in the graph~$H^*$. As our next step (cf. \Cref{fig:pftpoverview}), we contract for each vertex in~$V_E^*$ an incident edge in graph~$H^*$. We remark that if~$x^{vw}$ is a vertex in~$V_E^*$, then the only edges incident with vertex~$x^{vw}$ are~$\{v,x^{vw}\}$ and~$\{x^{vw},w\}$. In addition, the vertices~$v$ and~$w$ are the only neighbors of~$x^{vw}$ in graph~$H$ and in graph~$H^*$.

\begin{step}\label{step!three}
Contract for each vertex in~$V_E^*$ exactly one incident edge in~$H^*$ to obtain the graph~$G^*$. In other words, undo the subdivision applied on~$G$ to obtain~$H$.
\end{step}

We remark that~$\tw(G^*)\leq \tw(H^*)$, since edge contraction does not increase the treewidth of a graph \cite{RobertsonS86}. 
\iflong{}
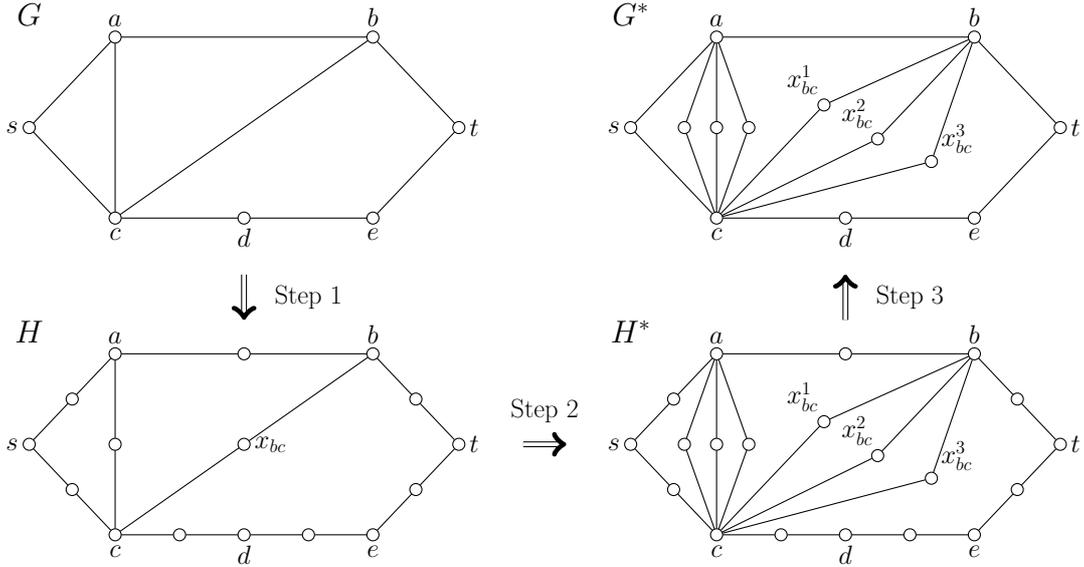
\begin{figure}[tb]
\centering

\begin{tikzpicture}[x=0.565cm, y=0.6cm]%

\tikzstyle{every node}=[scale=0.4798255,font={\huge}];

\node[align=center,scale=1.2] (whyfA) at (0,2.5) {$G$};

\node (s) at (0,0)[shape=circle, label=180:{$s$}, draw]{};
\node (a) at (2,2)[shape=circle, label=90:{$a$}, draw]{};
\node (b) at (8,2)[shape=circle, label=90:{$b$}, draw]{};
\node (c) at (2,-2)[shape=circle, label=270:{$c$}, draw]{};
\node (d) at (5,-2)[shape=circle, label=270:{$d$}, draw]{};
\node (e) at (8,-2)[shape=circle, label=270:{$e$}, draw]{};
\node (t) at (10,0)[shape=circle, label=0:{$t$}, draw]{};

\draw (s) -- (a);
\draw (s) -- (c);
\draw (a) -- (c);
\draw (a) -- (b);
\draw (c) -- (b);
\draw (c) -- (d);
\draw (d) -- (e);
\draw (e) -- (t);
\draw (b) -- (t);

\draw[double distance=1.4pt, ->] (5,-3.25) to (5,-4.25);
\node[align=center] (whyfA) at (6.5,-3.75) {\Cref{step!one}};

\node[align=center,scale=1.2] (whyfA) at (0,2.5-7) {$H$};

\node (s) at (0,0-7)[shape=circle, label=180:{$s$}, draw]{};
\node (sa) at (1,1-7)[shape=circle, draw]{};
\node (a) at (2,2-7)[shape=circle, label=90:{$a$}, draw]{};
\node (ac) at (2,0-7)[shape=circle, draw]{};
\node (sc) at (1,-1-7)[shape=circle, draw]{};
\node (b) at (8,2-7)[shape=circle, label=90:{$b$}, draw]{};
\node (ab) at (5,2-7)[shape=circle, draw]{};
\node (bc) at (5,0-7)[shape=circle, label=0:{$x_{bc}$}, draw]{};
\node (c) at (2,-2-7)[shape=circle, label=270:{$c$}, draw]{};
\node (cd) at (3.5,-2-7)[shape=circle, draw]{};
\node (d) at (5,-2-7)[shape=circle, label=270:{$d$}, draw]{};
\node (de) at (6.5,-2-7)[shape=circle, draw]{};
\node (e) at (8,-2-7)[shape=circle, label=270:{$e$}, draw]{};
\node (bt) at (9,1-7)[shape=circle, draw]{};
\node (et) at (9,-1-7)[shape=circle, draw]{};
\node (t) at (10,0-7)[shape=circle, label=0:{$t$}, draw]{};

\draw (s) -- (sa);
\draw (sa) -- (a);
\draw (s) -- (sc);
\draw (sc) -- (c);
\draw (a) -- (ac);
\draw (ac) -- (c);
\draw (a) -- (ab);
\draw (ab) -- (b);
\draw (c) -- (bc);
\draw (bc) -- (b);
\draw (c) -- (cd);
\draw (cd) -- (d);
\draw (d) -- (de);
\draw (de) -- (e);
\draw (e) -- (et);
\draw (et) -- (t);
\draw (b) -- (bt);
\draw (bt) -- (t);

\draw[double distance=1.4pt, ->] (5+6.5,-7) to (5+7.5,-7);
\node[align=right] (whyfA) at (5+7,-6.25) {\Cref{step!two}};

\node[align=center,scale=1.2] (whyfA) at (0+14,2.5-7) {$H^*$};

\node (s) at (0+14,0-7)[shape=circle, label=180:{$s$}, draw]{};
\node (sa) at (1+14,1-7)[shape=circle, draw]{};
\node (a) at (2+14,2-7)[shape=circle, label=90:{$a$}, draw]{};
\node (ac1) at (1.25+14,0-7)[shape=circle, draw]{};
\node (ac2) at (2+14,0-7)[shape=circle, draw]{};
\node (ac3) at (2.75+14,0-7)[shape=circle, draw]{};
\node (sc) at (1+14,-1-7)[shape=circle, draw]{};
\node (b) at (8+14,2-7)[shape=circle, label=90:{$b$}, draw]{};
\node (ab) at (5+14,2-7)[shape=circle, draw]{};
\node (bc1) at (4.5+14,0.5-7)[shape=circle, label=105:{$x_{bc}^1$}, draw]{};
\node (bc2) at (5.75+14,-0.25-7)[shape=circle, label=95:{$x_{bc}^2$}, draw]{};
\node (bc3) at (7+14,-0.75-7)[shape=circle, label=35:{$x_{bc}^3$}, draw]{};
\node (c) at (2+14,-2-7)[shape=circle, label=270:{$c$}, draw]{};
\node (cd) at (3.5+14,-2-7)[shape=circle, draw]{};
\node (d) at (5+14,-2-7)[shape=circle, label=270:{$d$}, draw]{};
\node (de) at (6.5+14,-2-7)[shape=circle, draw]{};
\node (e) at (8+14,-2-7)[shape=circle, label=270:{$e$}, draw]{};
\node (bt) at (9+14,1-7)[shape=circle, draw]{};
\node (et) at (9+14,-1-7)[shape=circle, draw]{};
\node (t) at (10+14,0-7)[shape=circle, label=0:{$t$}, draw]{};

\draw (s) -- (sa);
\draw (sa) -- (a);
\draw (s) -- (sc);
\draw (sc) -- (c);
\draw (a) -- (ac1);
\draw (ac1) -- (c);
\draw (a) -- (ac2);
\draw (ac2) -- (c);
\draw (a) -- (ac3);
\draw (ac3) -- (c);
\draw (a) -- (ab);
\draw (ab) -- (b);
\draw (c) -- (bc1);
\draw (bc1) -- (b);
\draw (c) -- (bc2);
\draw (bc2) -- (b);
\draw (c) -- (bc3);
\draw (bc3) -- (b);
\draw (c) -- (cd);
\draw (cd) -- (d);
\draw (d) -- (de);
\draw (de) -- (e);
\draw (e) -- (et);
\draw (et) -- (t);
\draw (b) -- (bt);
\draw (bt) -- (t);

\draw[double distance=1.4pt, ->] (5+14,-4.25) to (5+14,-3.25);
\node[align=center] (whyfA) at (6.5+14,-3.75) {\Cref{step!three}};

\node[align=center,scale=1.2] (whyfA) at (0+14,2.5) {$G^*$};

\node (s) at (0+14,0)[shape=circle, label=180:{$s$}, draw]{};
\node (a) at (2+14,2)[shape=circle, label=90:{$a$}, draw]{};
\node (ac1) at (1.25+14,0)[shape=circle, draw]{};
\node (ac2) at (2+14,0)[shape=circle, draw]{};
\node (ac3) at (2.75+14,0)[shape=circle, draw]{};
\node (b) at (8+14,2)[shape=circle, label=90:{$b$}, draw]{};
\node (bc1) at (4.5+14,0.5)[shape=circle, label=105:{$x_{bc}^1$}, draw]{};
\node (bc2) at (5.75+14,-0.25)[shape=circle, label=95:{$x_{bc}^2$}, draw]{};
\node (bc3) at (7+14,-0.75)[shape=circle, label=35:{$x_{bc}^3$}, draw]{};
\node (c) at (2+14,-2)[shape=circle, label=270:{$c$}, draw]{};
\node (d) at (5+14,-2)[shape=circle, label=270:{$d$}, draw]{};
\node (e) at (8+14,-2)[shape=circle, label=270:{$e$}, draw]{};
\node (t) at (10+14,0)[shape=circle, label=0:{$t$}, draw]{};

\draw (s) -- (a);
\draw (s) -- (c);
\draw (a) -- (ac1);
\draw (ac1) -- (c);
\draw (a) -- (ac2);
\draw (ac2) -- (c);
\draw (a) -- (ac3);
\draw (ac3) -- (c);
\draw (a) -- (b);
\draw (c) -- (bc1);
\draw (bc1) -- (b);
\draw (c) -- (bc2);
\draw (bc2) -- (b);
\draw (c) -- (bc3);
\draw (bc3) -- (b);
\draw (c) -- (d);
\draw (d) -- (e);
\draw (e) -- (t);
\draw (b) -- (t);

\end{tikzpicture}

\caption{Example of~\Cref{step!one,step!two,step!three} on the example graph~$G$ (top-left) with~$T=\{s,t\}$ and~$p=3$.}
\label{fig:modificationong}
\end{figure}

In \Cref{fig:modificationong}, we illustrate \Cref{step!one,step!two,step!three} on an example graph~$G$ with~$T=\{s,t\}$ and~$p=3$. The top-left graph is the original graph~$G$. The bottom-left graph is graph~$H$, obtained from~$G$ by applying \Cref{step!one}. The bottom-right graph is graph~$H^*$, obtained from~$H$ by applying \Cref{step!two}. The top-right graph is the final graph~$G^*$, obtained from~$H^*$ by applying \Cref{step!three}.
\fi{}

Let~$e=\{v,w\}\in E(G)$ be an edge in~$G$ and~$x_e\in V_E\subseteq V(H)$ be the corresponding vertex in~$H$. Then~$\{v,x_e\}$ and~$\{x_e,w\}$ are the incident edges of~$x_e$ in~$H$. If~$x_e \in V(H^*)$, then one of the incident edges~$\{v,x_e\}$ and~$\{x_e,w\}$ with vertex~$x_e$ is contracted and yields edge~$\{v,w\}\in E(G^*)$. We say that the edges~$\{v,w\}\in E(G)$ and~$\{v,w\}\in E(G^*)$ correspond one-to-one, and, for example, we write~$\{v,w\}\in E(G)\cap E(G^*)$. 

\ifshort{}\looseness=-1 Considering the graphs~$G$ and~$G^*$, we remark that one can show that, given an~$s$-$t$~path in the one graph, one can find an~$s$-$t$~path in the other graph using a common set of edges in~$E(G)\cap E(G^*)$.\fi{}\iflong{}Considering the graphs~$G$ and~$G^*$, we show that, given an~$s$-$t$~path in the one graph, we can find an~$s$-$t$~path in the other graph using a common set of edges in~$E(G)\cap E(G^*)$. 

\begin{lem}\label{lemma!pathswithsameedges}
\begin{enumerate}[(i)]
\item If~$P$ is an~$s$-$t$ path in~$G$, then there exists an~$s$-$t$ path~$P^*$ in~$G^*$ that contains all edges in~$E(P)\cap E(G^*)$.
\item If~$P^*$ is an~$s$-$t$ path in~$G^*$, then there exists an~$s$-$t$ path~$P$ in~$G$ that contains all edges in~$E(P^*)\cap E(G)$.
\end{enumerate}
\end{lem}

\begin{proof}
(i): Let~$P$ be an~$s$-$t$ path in~$G$. If~$P$ just contains edges in~$E(G)\cap E(G^*)$, then we set~$P^*=P$. If~$P$ contains edges in~$E(G)\backslash E(G^*)$, then~$P$ has a representation of consecutive subpaths~$P_i$,~$1\leq i \leq j$, and~$Q_i$,~$1\leq i\leq \ell$, where~$\{P_i\}_{1\leq i\leq j}$ is the set of subpaths of~$P$ that just contain edges in~$E(G)\cap E(G^*)$ and~$\{Q_i\}_{1\leq i\leq \ell}$ is the set of subpaths of~$P$ with endpoints in~$V(G)\cap V(G^*)$, inner vertices in~$V(G)\backslash V(G^*)$ and edges in~$E(G)\backslash E(G^*)$. Since for each~$1\leq i\leq \ell$, path~$Q_i$ is connecting two vertices~$v,w\in V(G)\cap V(G^*)$ in~$G$, there are~$p$ edge-disjoint paths of length 2 in~$G^*$ connecting~$v$ and~$w$ using the edges in~$E(G^*)\backslash E(G)$, that are the copy paths. For each~$i\in [\ell]$, let~$Q_i'$ be one of the copy paths connecting the endpoints of~$Q_i$. 
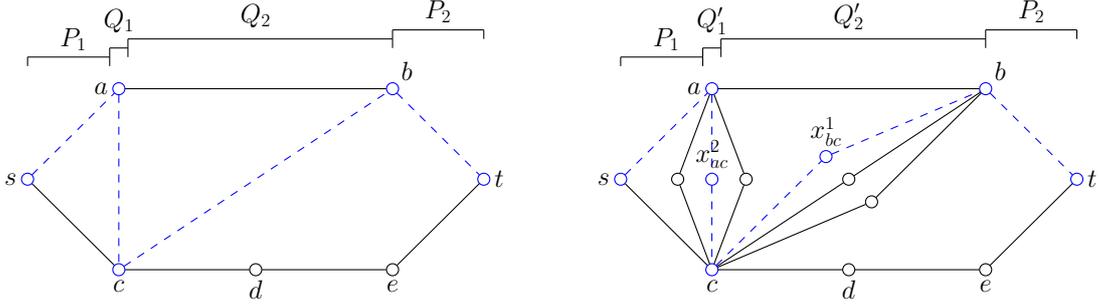
\begin{figure}[tb]
\begin{tikzpicture}[x=0.60cm,y=0.60cm]

\tikzstyle{every node}=[scale=0.4798255,font={\huge}];
\node (s) at (0,0)[shape=circle, color=blue, label=180:{$s$}, draw]{};
\node (a) at (2,2)[shape=circle, color=blue, label=180:{$a$}, draw]{};
\node (b) at (8,2)[shape=circle, color=blue, label=45:{$b$}, draw]{};
\node (c) at (2,-2)[shape=circle, color=blue, label=270:{$c$}, draw]{};
\node (d) at (5,-2)[shape=circle, label=270:{$d$}, draw]{};
\node (e) at (8,-2)[shape=circle, label=270:{$e$}, draw]{};
\node (t) at (10,0)[shape=circle, color=blue, label=0:{$t$}, draw]{};

\draw[color=blue, dashed] (s) -- (a);
\draw (s) -- (c);
\draw[color=blue, dashed] (a) -- (c);
\draw (a) -- (b);
\draw[color=blue, dashed] (c) -- (b);
\draw (c) -- (d);
\draw (d) -- (e);
\draw (e) -- (t);
\draw[color=blue, dashed] (b) -- (t);

\draw (0,2.5) -- (0,2.7);
\draw (1.8,2.5) -- (1.8,2.7);
\draw (0,2.7) -- (1.8,2.7);
\node[align=center] (P1) at (1,3.1) {$P_1$};

\draw (1.8,2.7) -- (1.8,2.9);
\draw (2.2,2.7) -- (2.2,2.9);
\draw (1.8,2.9) -- (2.2,2.9);
\node[align=center] (Q1) at (2,3.5) {$Q_1$};

\draw (2.2,2.9) -- (2.2,3.1);
\draw (8,2.9) -- (8,3.1);
\draw (2.2,3.1) -- (8,3.1);
\node[align=center] (Q2) at (5,3.6) {$Q_2$};

\draw (8,3.1) -- (8,3.3);
\draw (10,3.1) -- (10,3.3);
\draw (8,3.3) -- (10,3.3);
\node[align=center] (Q2) at (9,3.7) {$P_2$};

\node (s) at (0+13,0)[shape=circle, color=blue, label=180:{$s$}, draw]{};
\node (a) at (2+13,2)[shape=circle, color=blue, label=180:{$a$}, draw]{};
\node (ac1) at (1.25+13,0)[shape=circle, label=135:{}, draw]{};
\node (ac2) at (2+13,0)[shape=circle, color=blue, label=90:{$x_{ac}^2$}, draw]{};
\node (ac3) at (2.75+13,0)[shape=circle, label=45:{}, draw]{};
\node (b) at (8+13,2)[shape=circle, color=blue, label=45:{$b$}, draw]{};
\node (bc1) at (4.5+13,0.5)[shape=circle, color=blue, label=90:{$x_{bc}^1$}, draw]{};
\node (bc2) at (5+13,0)[shape=circle, label=0:{}, draw]{};
\node (bc3) at (5.5+13,-0.5)[shape=circle, label=0:{}, draw]{};
\node (c) at (2+13,-2)[shape=circle, color=blue, label=270:{$c$}, draw]{};
\node (d) at (5+13,-2)[shape=circle, label=270:{$d$}, draw]{};
\node (e) at (8+13,-2)[shape=circle, label=270:{$e$}, draw]{};
\node (t) at (10+13,0)[shape=circle, color=blue, label=0:{$t$}, draw]{};

\draw[color=blue, dashed] (s) -- (a);
\draw (s) -- (c);
\draw (a) -- (ac1);
\draw (ac1) -- (c);
\draw[color=blue, dashed] (a) -- (ac2);
\draw[color=blue, dashed] (ac2) -- (c);
\draw (a) -- (ac3);
\draw (ac3) -- (c);
\draw (a) -- (b);
\draw[color=blue, dashed] (c) -- (bc1);
\draw[color=blue, dashed] (bc1) -- (b);
\draw (c) -- (bc2);
\draw (bc2) -- (b);
\draw (c) -- (bc3);
\draw (bc3) -- (b);
\draw (c) -- (d);
\draw (d) -- (e);
\draw (e) -- (t);
\draw[color=blue, dashed] (b) -- (t);

\draw (0+13,2.5) -- (0+13,2.7);
\draw (1.8+13,2.5) -- (1.8+13,2.7);
\draw (0+13,2.7) -- (1.8+13,2.7);
\node[align=center] (P1) at (1+13,3.1) {$P_1$};

\draw (1.8+13,2.7) -- (1.8+13,2.9);
\draw (2.2+13,2.7) -- (2.2+13,2.9);
\draw (1.8+13,2.9) -- (2.2+13,2.9);
\node[align=center] (Q1) at (2+13,3.5) {$Q_1'$};

\draw (2.2+13,2.9) -- (2.2+13,3.1);
\draw (8+13,2.9) -- (8+13,3.1);
\draw (2.2+13,3.1) -- (8+13,3.1);
\node[align=center] (Q2) at (5+13,3.6) {$Q_2'$};

\draw (8+13,3.1) -- (8+13,3.3);
\draw (10+13,3.1) -- (10+13,3.3);
\draw (8+13,3.3) -- (10+13,3.3);
\node[align=center] (Q2) at (9+13,3.7) {$P_2$};
\end{tikzpicture}
\caption{The graphs~$G$ (left-hand side) and~$G^*$ (right-hand side) from Figure~\ref{fig:modificationong}. The dashed edges belong to an~$s$-$t$~path in~$G$ and~$G^*$ respectively. The upper braces show the range of the consecutive subpaths~$P_1$,~$Q_1$,$Q_2$,~$P_2$ for~$G$ and~$P_1$,~$Q_1'$,~$Q_2'$,~$P_2$ for~$G^*$.}
\label{fig:PisQis}
\end{figure}
Figure~\ref{fig:PisQis} illustrates this correspondence on an example graph. Replacing each~$Q_i$ by such a path~$Q_i'$ in~$G^*$ yields a path~$P'$ with consecutive subpaths~$P_i$,~$1\leq i \leq j$, and~$Q_i'$,~$1\leq i\leq \ell$, in~$G^*$ connecting~$s$ and~$t$ that contains all edges in~$E(P)\cap E(G^*)$.

(ii): Let~$P^*$ be an~$s$-$t$ path in~$G^*$. If~$P^*$ just contains edges in~$E(G)\cap E(G^*)$, then we set~$P=P^*$. If~$P^*$ contains edges in~$E(G^*)\backslash E(G)$, then~$P^*$ has a representation of consecutive subpaths~$P_i'$,~$1\leq i \leq j$, and~$Q_i'$,~$1\leq i\leq \ell$, where~$\{P_i'\}_{1\leq i\leq j}$ is the set of subpaths of~$P^*$ that just contain edges in~$E(G^*)\cap E(G)$ and~$\{Q_i'\}_{1\leq i\leq \ell}$ is the set of subpaths of~$P^*$ with endpoints in~$V(G^*)\cap V(G)$, inner vertices in~$V(G^*)\backslash V(G)$, and edges in~$E(G^*)\backslash E(G)$. We remark that each~$Q_i'$ is one of the copy paths in~$G^*$. By construction of~$G^*$, each~$Q_i'$ connects two vertices in~$V(G^*)\cap V(G)$ that are connected by a path in~$G$ with no inner vertices in~$V(G^*)\cap V(G)$. Therefore, for each $i\in[\ell]$, we can replace path~$Q_i'$ by such a path~$Q_i$ in~$G$. This yields an~$s$-$t$~path~$P$ in~$G$ with consecutive subpaths~$P_i'$,~$1\leq i \leq j$, and~$Q_i$,~$1\leq i\leq \ell$, that contains all edges in~$E(P^*)\cap E(G)$.
\end{proof}

We modified graph~$G$ to graph~$G^*$ by applying \Cref{step!one,step!two,step!three}. By \Cref{lemma!pathswithsameedges}, we can construct $s$-$t$~routes in~$G$ and~$G^*$ that use edges in the common set of edges~$E(G)\cap E(G^*)$.\fi{} The next lemma states that each minimal~$s$-$t$~cut of size smaller than~$p$ in one of the graphs~$G$ and~$G^*$ is also a minimal~$s$-$t$~cut of size smaller than~$p$ in the other graph.

\begin{lem}\label{lemma!cutingiffcutingstar}
Let~$C\subseteq E(G)\cap E(G^*)$. Edge set~$C$ is a minimal~$s$-$t$ cut in~$G$ of size smaller than~$p$ if and only if~$C$ is a minimal~$s$-$t$ cut in~$G^*$ of size smaller than~$p$.
\end{lem}

\iflong{}
\begin{proof}
We make use of \Cref{lemma!pathswithsameedges} in the following proof. We remark that no edge in~$E(G^*)\backslash E(G)$ is in any minimal~$s$-$t$ cut of size smaller than~$p$ in~$G^*$ since, by the treewidth reduction technique, for each of these edges there are~$p-1$ copies in~$G^*$.

``$\Rightarrow$'': Let~$C$ be a minimal~$s$-$t$~cut of size smaller than~$p$ in~$G$. By \Cref{lemma!minimalcutcorrespondstominimalsep},~$C$ has a corresponding minimal~$s$-$t$ separator~$S_C$ of size smaller than~$p$ in~$H$. By the treewidth reduction technique,~$S_C$ is a minimal~$s$-$t$ separator in~$H^*$. By \Cref{lemma!eachofneighborsarepartofspearators}, every neighbor of~$S_C$ is contained in~$H^*$. By our contraction of edges of~$H^*$ to~$G^*$, for each vertex of~$S_C$ an incident edge is contracted and yields the edge set~$C$ again. Since $S_C$~is a minimal~$s$-$t$~separator in~$H^*$ of size smaller than~$p$ and each vertex in~$S_C$ has degree exactly two, set~$C$ is a minimal~$s$-$t$~cut in~$G^*$ of size smaller than~$p$. 

``$\Leftarrow$'': Let~$C$ be a minimal~$s$-$t$~cut in~$G^*$ of size smaller than~$p$. Suppose~$C$ is not a minimal~$s$-$t$ cut in~$G$ of size smaller than~$p$. We distinguish two cases.

\emph{Case 1:}~$C$ is not an~$s$-$t$~cut in~$G$. Then there exists a path~$P$ in~$G$ connecting~$s$ and~$t$ avoiding the edges in~$C$. By \Cref{lemma!pathswithsameedges}, there exists an~$s$-$t$~path~$P^*$ in~$G^*$ that contains all edges in~$E(P)\cap E(G^*)$. Since no edge in~$E(G^*)\backslash E(G)$ is in any minimal~$s$-$t$ cut of size at most~$p-1$ of~$G^*$,~$P^*$ avoids the edges in~$C$. This is a contradiction to the fact that~$C$ is a minimal~$s$-$t$~cut in~$G^*$.

\emph{Case 2:}~$C$ is an~$s$-$t$~cut in~$G$, but~$C$ is not a minimal~$s$-$t$~cut in $G$. Then there exists~$e\in C$ such that~$C':=C\backslash\{e\}$ is an~$s$-$t$ cut in~$G$. Since~$C$ is a minimal~$s$-$t$~cut in~$G^*$, the set~$C'$ is not an $s$-$t$~cut in $G^*$. Thus, there exists an~$s$-$t$~path~$P^*$ in $G^*$ that avoids the edges in~$C'$. By \Cref{lemma!pathswithsameedges}, there exists an~$s$-$t$~path~$P$ in~$G$ that contains all the edges in~$E(P^*)\cap E(G)$. Since no edge in~$E(G)\backslash E(G^*)$ is in any minimal~$s$-$t$~cut of size at most~$p-1$ in~$G$, path~$P$ avoids the edges in~$C'$. Therefore, set~$C'$ is not an~$s$-$t$~cut in~$G$, and thus,~$C$ is a minimal~$s$-$t$~cut in~$G$.
\end{proof}
\fi{}

\looseness=-1 Recalling \Cref{lemma!yesinstancethenedgesetminimalcut}, we know that if an instance of \MSE{} is a yes-instance, then we can find $k$~edges such that the $k$~edges form a solution for the instance and each of the $k$~edges is part of a minimal~$s$-$t$~cut of size smaller than~$p$ in~$G$. By \Cref{lemma!cutingiffcutingstar}, the graphs~$G$ and~$G^*$ have the same set of minimal~$s$-$t$~cuts of size smaller than~$p$ in common. Combining \Cref{lemma!yesinstancethenedgesetminimalcut} and \Cref{lemma!cutingiffcutingstar} leads to the following lemma.

\begin{lem}\label{lemma!gstaryesiffgyes}
$(G^*,s,t,p,k)$ is a yes-instance of \MSE{} if and only if~$(G,s,t,p,k)$ is a yes-instance of \MSE{}.
\end{lem}

\iflong
\begin{proof}
We make use of the contraction equivalent of \MSE{}.

``$\Rightarrow$'': Let~$(G^*,s,t,p,k)$ be a yes-instance of \MSE{}. By \Cref{lemma!yesinstancethenedgesetminimalcut}, we find a solution~$F\subseteq E(G^*)$ such that each edge in~$F$ is part of a minimal~$s$-$t$ cut in~$G^*$ of size smaller than~$p$. It follows that~$F\subseteq E(G)\cap E(G^*)$, since by our construction no edge in~$(E(G^*)\backslash E(G))$ is part of a minimal~$s$-$t$ cut of size smaller than~$p$ in~$G^*$. Let~$G_F:=G/F$ be the graph~$G$ with all edges in $F$ contracted. Suppose that~$G_F$ with unit edge capacities allows a maximum~$s$-$t$~flow of value smaller than~$p$. Then there exists a minimal~$s$-$t$ cut~$C$ of size smaller than~$p$ in~$G_F$. By \Cref{lemma!cutingiffcutingstar},~$C$ is also a minimal~$s$-$t$~cut of size smaller than~$p$ in~$G^*_F:=G^*/F$. This is a contradiction to the fact that the value of any maximum~$s$-$t$~flow in~$G^*_F$ with unit edge capacities is at least~$p$, and hence, set~$F$ is a solution for instance~$(G,s,t,p,k)$.

``$\Leftarrow$'': Let~$(G,s,t,p,k)$ be a yes-instance of \MSE{}. By \Cref{lemma!yesinstancethenedgesetminimalcut}, we find a solution~$F\subseteq E(G)$ such that each edge in~$F$ is part of a minimal~$s$-$t$~cut in~$G$ of size smaller than~$p$. It follows that~$F\subseteq E(G)\cap E(G^*)$. Suppose that~$G^*_F:=G^*/F$ with unit edge capacities allows a maximum~$s$-$t$~flow of value smaller than~$p$. Then there exists a minimal~$s$-$t$ cut~$C$ of size smaller than~$p$ in~$G^*_F$. By \Cref{lemma!cutingiffcutingstar}, $C$~is a minimal~$s$-$t$~cut of size smaller than~$p$ in~$G_F:=G/F$. This is a contradiction to the fact that the value of any maximum~$s$-$t$~flow in~$G_F$ with unit edge capacities is at least~$p$, and hence, set~$F$ is a solution for instance~$(G^*,s,t,p,k)$.
\end{proof}
\fi

By \Cref{lemma!gstaryesiffgyes}, we know that the instances~$(G^*,s,t,p,k)$ and~$(G,s,t,p,k)$ are equivalent for \MSE{}. By our construction, we know that the treewidth of $G^*$ is upper-bounded by a function only depending on the number $p$ of routes. In addition, we know that \msetsc{} is fixed-parameter tractable with respect to the number~$p$ of routes and an upper bound on the treewidth of the input graph. Thus, we are ready to prove \cref{theorem!fptwrtp}.

\begin{proof}[Proof of \Cref{theorem!fptwrtp}]
First we modify our graph~$G=(V,E)$ by applying \Cref{step!one,step!two,step!three}. Let~$H$,~$H^*$, and~$G^*$ be the according graphs. By \Cref{theorem!twrt}, the treewidth of~$H^*$ is upper-bounded by~$h(p)$ for some function~$h$. Since edge contractions do not increase the treewidth of a graph \cite{RobertsonS86}, it follows that~$\tw(G^*)\leq \tw(H^*)$. By \Cref{lemma!gstaryesiffgyes}, the instances~$(G^*,s,t,p,k)$ and~$(G,s,t,p,k)$ are equivalent for \MSE{}.

We know from~\Cref{thm:twdp} that \acrmse{$p,\omega$} is fixed-parameter tractable when parameterized by the number $p$ of routes and by an upper bound~$\omega$ on the treewidth of the input graph. Since function~$h$ only depends on~$p$ and $h(p)$ is upper-bounding the treewidth of graph~$G^*$, we can solve instance~$(G^*,s,t,p,k)$ in~$f(p)\cdot O(|V(G^*)|)$~time, where~$f$ is a computable function only depending on parameter~$p$. Since~$|V(G^*)|\leq |V(G)|+p\cdot |E(G)|\leq p\cdot |G|$ and the instances~$(G^*,s,t,p,k)$ and~$(G,s,t,p,k)$ are equivalent for \MSE{}, we can decide instance~$(G,s,t,p,k)$ in~$f(p)\cdot p\cdot O(|G|)$~time, that is, in FPT-time.
\end{proof}

Using the dynamic program from \cref{sec:dp} the running time of the above algorithm amounts to $O(p^2 \cdot (h(p)+4)^{3\cdot p\cdot (h(p)+3) +3}\cdot |G|)$. Using the bound $h(p) \leq 2^{O(p^2)}$~\cite{MarxOR13}, we obtain a running time of $2^{p^3\cdot2^{O(p^2)}}\cdot (n+m)$.

%


\section{No Polynomial Problem Kernel for the Parameter Number of Routes}\label{sec:nopoly}

In the previous section, we showed that \MSEl is fixed-parameter tractable with respect to the number $p$~of routes. It is well known that a problem is fixed-parameter tractable if and only if it admits a problem kernel. Of particular interest is the minimal possible size of a problem kernel. Accordingly, in this section we prove the following lower bound.

\begin{theorem}\label{thm:nopolyker}
\MSEl does not admit a polynomial-size problem kernel with respect to the number $p$ of routes, unless $\text{NP}\subseteq\text{coNP/poly}$.
\end{theorem}

\iflong{}
\begin{figure}
\centering
\begin{tikzpicture}

\def\usc{0.75}

\draw (0.5,0) ellipse (0.5 and 0.25);
\node (s1) at (0,0)[label=90:{$s_1$}, fill, scale=\usc, circle, draw]{};
\node (t1) at (1,0)[label=90:{$t_1$}, fill, scale=\usc, circle, draw]{};

\draw (0.5+2.5,0) ellipse (0.5 and 0.25);
\node (s1) at (0+2.5,0)[label=90:{$s_2$}, fill, scale=\usc, circle, draw]{};
\node (t1) at (1+2.5,0)[label=90:{$t_2$}, fill, scale=\usc, circle, draw]{};

\node (ld) at (6.25-1,0)[]{$\ldots$};
\node (ld) at (6.25+1,0)[]{$\ldots$};

\draw (0.5+9,0) ellipse (0.5 and 0.25);
\node (s1) at (0+9,0)[label=90:{$s_\ell$}, fill, scale=\usc, circle, draw]{};
\node (t1) at (1+9,0)[label=90:{$t_\ell$}, fill, scale=\usc, circle, draw]{};

\def\py{-0.5}

\draw[->] (0.5,-1-\py) to (0.75+0.5,-2-\py);
\draw[->] (0.5+2.5,-1-\py) to (0.75+2.25,-2-\py);
\draw[->] (0.5+6,-1-\py) to (0.75+4.5+1,-2-\py);
\draw[->] (0.5+9,-1-\py) to (0.75+6.25+1.75,-2-\py);

\draw (1.5,-3-\py) ellipse (0.75 and 0.2);
\node (s1) at (0.75,-3-\py)[label=180:{$s_1$}, fill, scale=\usc, circle, draw]{};
\node (t1s2) at (2.25,-3-\py)[label=270:{$t_1=s_2$}, fill, scale=\usc, circle, draw]{};

\draw (3,-3-\py) ellipse (0.75 and 0.2);
\node (t2s3) at (3.75,-3-\py)[label=90:{$t_2=s_3$}, fill, scale=\usc, circle, draw]{};

\draw (4.5,-3-\py) ellipse (0.75 and 0.2);
\node (ld) at (5.5+0.75-0.5,-3-\py)[]{$\ldots$};
\node (ld) at (5.5+0.75+0.5,-3-\py)[]{$\ldots$};

\draw (6.5+2.00,-3-\py) ellipse (0.75 and 0.2);
\node (t2s3) at (5.75+2.00,-3-\py)[label=270:{$t_{\ell-1}=s_\ell$}, fill, scale=\usc, circle, draw]{};
\node (tll) at (7.25+2.00,-3-\py)[label=0:{$t_\ell$}, fill, scale=\usc, circle, draw]{};

\end{tikzpicture}
\caption{OR-cross-composition of $\ell$ instances of \AMSE into one instance of \MSE.}
\label{fig:ORCC}
\end{figure}
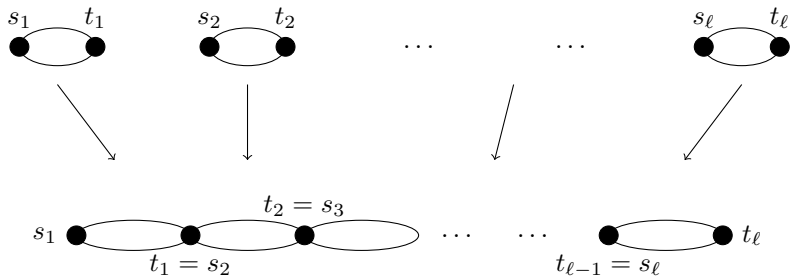
\fi{}
\noindent We prove \cref{thm:nopolyker} via an \emph{OR-cross-composition} \cite{BodlaenderJK14}, that is, given $\ell$ instances of an NP-hard problem~$Q$, all contained in one equivalence class of a polynomial-time computable relation~$\mathcal{R}$ of our choosing, we compute in polynomial-time an instance $(G, s, t, p, k)$ of \MSE such that%
\begin{inparaenum}[(i)]
\item $p$ is bounded by a polynomial function of the size of the largest input instance plus $\log(\ell)$ (\emph{boundedness}), and
\item $(G, s, t, p, k)$ is a yes-instance if and only if one of the input instances is a yes-instance (\emph{correctness}).
\end{inparaenum}
If this is possible, then \MSE{} does not admit a polynomial-size problem kernel with respect to~$p$ unless $\text{NP}\subseteq\text{coNP/poly}$~\cite{BodlaenderJK14}.

It is tempting to use \MSE itself as the problem $Q$, to assume that each of the instances asks for the same number of routes and same number of shared edges by virtue of $\mathcal{R}$, and to OR-cross-compose by simply gluing the graphs in a chain-like fashion on sinks and sources\iflong{} (see \cref{fig:ORCC})\fi{}. This fulfills the boundedness constraint, but not necessarily the correctness constraint, since the instances can share shared edges between them. 
That is, the shared edges in any instance can be as large as the number of edges of the graph in the instance. 
Hence, we use the following problem as the problem~$Q$ instead.
\decprob{\AMSEl (\AMSE)}{An undirected graph $G$, two distinct vertices $s,t\in V(G)$, and two integers $p,k\in \mathbb{N}$ such that $G$ has a $(p,s,t)$-routing with at most $k + 1$ shared edges.
}{Is there a $(p, s, t)$-routing in $G$ with at most $k$ shared edges?}

\begin{proposition}\label{lemma:amsenphard}
\AMSEl is NP-hard.
\end{proposition}

\cref{lemma:amsenphard} can be proven via a reduction from \MSE to \AMSE that introduces an additional path of length~$k + 1$ connecting $s$ and $t$.\iflong{} As a technical remark, since any instance of MSE with~$k=0$ is solvable in polynomial time, we assume here and in the following that~$k>0$.\fi{}

\iflong{}
\begin{proof}[Proof of \Cref{lemma:amsenphard}]

Let $(G,s,t,p,k)$ be an instance of \MSE. We describe the construction of an instance of \AMSE given $(G,s,t,p,k)$ and prove their equivalence.

\emph{Construction.}
We add a $P_{k + 2}$ with endpoints~$s$ and~$t$ to $G$ to obtain the graph~$G'$. Note that we can route any number $\geq 2$ of routes over the $P_{k+2}$ from $s$ to $t$ while sharing exactly $(k+1)$ edges. Thus, $(G',s,t,p+1,k+1)$ is a yes-instance of \MSE, and $(G',s,t,p+1,k)$ is an instance of \AMSE. We show that $(G,s,t,p,k)$ is a yes-instance of \MSE if and only if $(G',s,t,p+1,k)$ is a yes-instance of \AMSE.

\emph{Correctness.}
Let $(G,s,t, p ,k)$ be a yes-instance of \MSE. Let $\mathcal{P}$ be a $(p,s,t)$-routing in $G$ sharing at most $k$ edges. Since $G'[V(G)]=G$, $\mathcal{P}$ is also a $(p,s,t)$-routing in $G'$ sharing at most $k$ edges. Since the additional $P_{k+2}$ in $G'$ is not contained in an $s$-$t$~route in $\mathcal{P}$, we can construct an additional $s$-$t$~route $P$ in $G'$ using only this $(k+1)$-chain without sharing an additional edge. Thus, $\mathcal{P}\cup \{P\}$ is a $(p+1,s,t)$-routing in $G'$ sharing at most $k$ edges. That is, $(G',s,t,p+1,k)$ is a yes-instance of \AMSE.

Conversely, let $(G',s,t,p+1,k)$ be a yes-instance of \AMSE. Observe that at most one $s$-$t$~route appears on the $(k+1)$-chain in $G'$. Thus, at least $p$ $s$-$t$~routes share at most $k$ edges in $G'[V(G)]=G$. It follows that $(G,s,t,p,k)$ is a yes-instance of \MSE.
\end{proof}
\fi{}

If we OR-cross-compose $\ell$ instances of \AMSE instead, we know that if the resulting instance has a routing with $\ell(k + 1) - 1$ shared edges, then without loss of generality each of the original instances contributes at most~$k + 1$ shared edges. This means that at least one of the original instances is a yes-instance, giving the correctness of the OR-cross-composition.

\iflong{}
\begin{proof}[Proof of \Cref{thm:nopolyker}]
We describe an OR-cross-composition\iflong{} as sketched in \Cref{fig:ORCC}\fi{}. We cross-compose $\ell$~instances of \AMSE into one instance of \textsc{MSE($p$)}. We define the relation~$\mathcal{R}$ as follows: $(G,s,t,p,k)\equiv_\mathcal{R} (G',s',t',p',k')$ if $p=p'$, and $k=k'$. Obviously, to check whether two instances are equivalent with respect to~$\mathcal{R}$ can be done in constant time. Moreover, in any finite set of instances, the number of equivalence classes with respect to~$\mathcal{R}$ is upper-bounded by the product of the largest $p$ value and largest $k$ value over all instances in the finite set. Thus, $\mathcal{R}$ is a polynomial equivalence relation. 
We cross-compose $\ell$ $\mathcal{R}$-equivalent instances~$(G_i,s_i,t_i,p,k)_{i=1,\ldots,\ell}$ of \AMSE to an instance of \textsc{MSE($p$)} as follows. 
In the following, let~$I_j := (G_j,s_j,t_j,p,k)$ for all $j\in[\ell]$.

\emph{Construction.}
We join the graphs $G_1,\ldots,G_\ell$ in a chain-like fashion, that is, we identify $t_{i-1}$ and $s_i$ for each $i=2,\ldots, \ell$. %
Let $G^*$ be the obtained graph. Let $I^*:=(G^*, s_1,t_\ell,p^*,k^*)$ with $p^*=p$ and $k^*=\ell\cdot (k +1)-1$ be the instance of \MSE. Recall that $p^*$ is the parameter. We claim that $I^*$ is a yes-instance of \MSE if and only if at least one instance $I_i$ is a yes-instance of \AMSE.

\emph{Correctness.}
Let $I^*$ be a yes-instance of \MSE. Observe that each $t_{i-1}=s_i$ for $i=2,\ldots,\ell$ is a 1-separator in~$G^*$. Let $G^*_i$ be the graph induced by all vertices that are ``between'' $s_i$ and $t_i$: these are all vertices that are reachable in $G$ from $s_i$ and $t_i$ without touching the other. Let $S_i\subseteq G^*-\{t_i\}$ with $s_i\in V(S_i)$, that is $S_i$ is the connected component of $G-\{t_i\}$ that contains $s_i$. Analogously, let $T_i\subseteq G^*-\{s_i\}$ with $t_i\in V(T_i)$. We define $G^*_i:= G^*[\{s_i,t_i\}\cup (V(S_i)\cap V(T_i))]$. Observe that $G^*_i$ is isomorphic to $G_i$ for all $i=1,\ldots,\ell$, and $G^*_i$ and $G^*_j$ are edge-disjoint for all $i,j\in [\ell]$ with $i\neq j$. Moreover, $V(G^*)=\bigcup_{1\leq i\leq \ell} V(G^*_i)$ and $E(G^*)=\bigcup_{1\leq i\leq \ell} E(G^*_i)$. 

Since $I^*$ is a yes-instance of \MSE, there are at most $\ell\cdot (k+1)-1$ shared edges in any solution to $I^*$. Suppose one can find at least $(k+1)$ shared edges in each $G^*_i$. Since $G^*_i$ and $G^*_j$ are edge-disjoint for all $i,j\in [\ell]$ with $i\neq j$, it follows that there are at least $\ell\cdot (k+1)>\ell\cdot (k+1)-1 = k^*$ shared edges, contradicting the fact that $I^*$ is a yes-instance. Thus, there exists an index~$j\in[\ell]$ such that there are at most $k$ shared edges in $G^*_j$, or equivalently, $I_j$ is a yes-instance of \AMSE. 

Conversely, let $j\in[\ell]$ such that $I_j$ is a yes-instance of \AMSE. By construction, sharing exactly the same $k$ edges of a solution to $I_j$ allows $p$ $s_j$-$t_j$~routes in $G^*_j$. We know that for each $i\in[\ell]\backslash\{j\}$, instance~$I_i$ of \AMSE allows $p$ $s_i$-$t_i$~routes sharing at most $k+1$ edges. Thus, we can route $p$ $s_1$-$t_\ell$~routes through $G^*$ sharing at most 
\[ k + (\ell-1)\cdot (k+1) = \ell\cdot k + \ell -1 = \ell\cdot (k+1) -1 = k^* \]
edges. Thus, $I^*$ is a yes-instance of \MSE.
\end{proof}
\fi{}

\section{W[1]-hardness with Respect to Treewidth}\label{sec:whard-tw}

In this section, we present the following result.

\begin{theorem}\label{thm:whard-tw}
  \MSEl is W[1]-hard when parameterized by treewidth and the number~$k$ of shared edges combined.
\end{theorem}

To prove~\Cref{thm:whard-tw}, we give a parameterized reduction from the following problem. Herein, $\dot\cup$ denotes the disjoint union of sets.

\decprob{\MCCl (\MCC)}{An undirected, $k$-partite graph $G=(V=V_1\dot\cup\ldots\dot\cup V_k,E)$ with $k\in \mathbb{N}$.}{Is there a set $C\subseteq V$ of vertices such that $G[C]$ is a $k$-clique in $G$?}

\noindent \MCC is W[1]-complete when parameterized by~$k$~\cite{FellowsHRV09}. In the remainder of the section $(G, k)$ is an arbitrary but fixed instance of \MCC. We denote $|V_i|=:n_i$ and $V_i=:\{v^i_1,\ldots,v^i_{n_i}\}$ for all $i\in[k]$. We also say that $G$ has the \emph{color classes} $1,\ldots,k$, where each color class~$i$ is represented by the vertices in $V_i$. We write $E_{i,j}:=\{\{v,w\}\in E\mid v\in V_i, w\in V_j\}$ for the edges connecting vertices in $V_i$ and $V_j$, $i,j\in[k]$.

\looseness=-1 
The reduction is based on the following idea. The routes we are to allocate will be split evenly into contingents of routes for each color class by a simple gadget. For each of the color classes, we introduce a selection gadget, that contains vertices (\emph{outputs}) that correspond to the vertices in the \MCC instance. Each selection gadget will route almost all the routes in its contingent to exactly one of its outputs. The outputs will then disperse $(k-1)$-times a number of routes corresponding to the ID of the vertex that this output represents. In this way, the selection gadgets represent a choice of vertices, one for each color class. In order to verify that the choice represents a clique, we introduce validation gadgets, corresponding to the pairs of color classes. They will receive the routes from the outputs of the selection gadgets, that is, the ``input'' of the validation gadgets is a sum of two IDs. They induce a small number of shared edges only if the vertices according to the number of routes are connected. In order to achieve this, we ensure that the sum of two IDs uniquely identifies the vertices. We achieve this by using \Sidss{}.

\subparagraph*{Vertex IDs based on~\Sidss.}

A \emph{\Sids} is a set $S \subseteq \mathbb{N}$ that fulfills that for each $i,j,k,\ell\in S$ holds that if $i+k=j+\ell$ then $\{i, k\} = \{j, \ell\}$. That is, the sum of any two distinct elements in $S$ is unique. A \Sids~$S$ with $\max_{i\in S} i \in O(|S|^3)$ can be constructed on $O(|S|)$~time~\cite[page~42]{Dim02}. As mentioned, we use a \Sids to distinguish numbers of routes corresponding to vertices. For this purpose, we fix a \Sids~$S$ with $|S| = |V|$ and assign to each vertex~$v \in V$ an \emph{ID}~$g(v) \in S$ where $g$ is a bijection. For technical reasons, we need the following additional properties of $g$ (and~$S$):
\begin{compactenum}[(i)]
 \item $g(v)\geq n^3$ for all $v\in V$,
 \item $|g(v)-g(w)|\geq n^3$ for all $v,w\in V$, $v\neq w$, and 
 \item $|(g(v)+g(w))-(g(x)+g(y))|\geq n^3$ for all $v,w,x,y\in V$, $v\neq w$, $y\not\in \{v,w,x\}$.\label[prop]{ss:sumdiff}
\end{compactenum}
Clearly, by adding one to each integer in the \Sids~$S$ and then multiplying each integer by~$n^3$ we obtain a \Sids and a mapping $g$ that fulfill all of the above properties simultaneously. 

To enforce that only adjacent vertices are chosen in the selection gadgets, a part in a validation gadget that represents an edge must have the property that, if many routes are routed through it, then the number of routes corresponds to precisely the sum of IDs of the endpoints of the edge that is represented by this part. To do this, we have to enforce both upper and lower bounds on the sum of IDs. Upper bounds will be enforced by long parallel paths; for lower bounds, we use the notion of ``complement'' of an ID. For this, we define $\overline{g(v)}:=M-g(v)$ for all $v\in V$, where $M:=n^3+\max_{v\in V}g(v)$. Note that $g(v)+g(w)<g(x)+g(y)$ if and only if $\overline{g(v)}+\overline{g(w)}>\overline{g(x)}+\overline{g(y)}$ for $v,w,x,y\in V$.

\paragraph*{Construction.}

In the following, we describe the construction of the instance~$(G',s,t,p,k')$ of \MSE, given instance~$(G,k)$ of \MCC. Initially, $G'$ consists only of the two vertices~$s$ and~$t$, the source and the sink vertex, respectively. We describe the gadgets we use and their interconnections, which will fully describe the construction of $G'$. As mentioned, our gadgetry consists of two gadget types, selection gadgets on the one hand and validation gadgets on the other hand.

Before we proceed, we fix the following notation. An \emph{$m$-chain} is a $P_{m+1}$, i.e.\ a path of length $m$. A set of $\ell$ $m$-chains with common endpoints we call an \emph{$(\ell,m)$-bundle}. An  \emph{$(q,\ell,m)$-feather} is obtained by identifying one endpoint of an $(\ell,m)$-bundle with one endpoint of a $q$-chain.\iflong{} We also call the $q$-chain of the feather the \emph{$q$-shaft} and we call the $m$-chains the \emph{barbs} of the feather. Refer to \Cref{fig:subgraphs} for an illustration.\fi{} In the following, by attaching a chain, bundle, or feather~$H$ to a vertex~$v$, we mean to identify~$v$ with an endpoint of~$H$.
\iflong{}
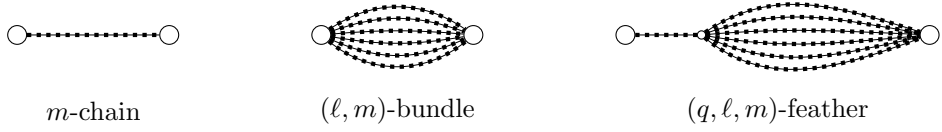
\begin{figure}[t]
\centering
  \begin{tikzpicture}

  \node (a) at (0,0)[scale=0.75,circle,draw]{};
  \node (b) at (2,0)[scale=0.75,circle,draw]{};
  \draw[very thin] (a) to (b);
  \draw[dotted, ultra thick] (a) to (b);

  \node (tx1) at (1,-1)[]{$m$-chain};

  \node (a) at (4,0)[scale=0.75,circle,draw]{};
  \node (b) at (6,0)[scale=0.75,circle,draw]{};

  \foreach \x in {-40,-25,...,40}{
  \draw[very thin] (a) to [out=\x,in=180-\x](b);
  \draw[dotted, ultra thick] (a) to [out=\x,in=180-\x](b);
  }

  \node (tx1) at (5,-1)[]{$(\ell,m)$-bundle};

  \node (s) at(8,0)[scale=0.75,circle, draw]{};
  \node (a) at (9,0)[scale=0.33,circle, draw]{};
  \node (b) at (12,0)[scale=0.75,circle,draw]{};

  \draw[very thin] (s)  to (a);
  \draw[dotted, ultra thick] (s)  to (a);

  \foreach \x in {-40,-25,...,40}{
  \draw[very thin] (a) to [out=\x,in=180-0.5*\x](b);
  \draw[dotted, ultra thick] (a) to [out=\x,in=180-0.5*\x](b);
  }

  \node (tx1) at (10,-1)[]{$(q,\ell,m)$-feather};

  \end{tikzpicture}
\caption{Illustration of a chain, bundle and feather.}\label{fig:subgraphs}
\end{figure}
\fi

\medskip We set the number of paths \ifshort{}$p=\left(|E|-\binom{k}{2}\right) + k\cdot \left((k-1)\cdot M+1\right) + n$\fi{}\iflong{}\[p=\left(|E|-\binom{k}{2}\right) + k\cdot \left((k-1)\cdot M+1\right) + n\]\fi{} and the number of shared edges\iflong{} (in the following also denoted by \emph{the budget}) \[k'=k\cdot k^{10}+ k\cdot(k+2(k-1))\cdot k^5 + \binom{k}{2}\cdot 3k.\]\fi{}\ifshort{} $k'=k\cdot k^{10}+ k\cdot(k+2(k-1))\cdot k^5 + \binom{k}{2}\cdot 3k$ .\fi{}

\subparagraph*{Selection gadgets.}
For each color class $i\in[k]$ in the instance $(G,k)$, we construct a selection gadget~\fbox{$i$} that selects exactly one vertex of~$V_i$ as follows.%
\iflong{} In \Cref{fig:choosegadget}, we illustrate an example of a selection gadget~\fbox{$i$} for color class~$i$. 
\begin{figure}[t]
\centering
\begin{tikzpicture}[scale=.95,x=0.8cm, y=0.95cm]

\node (C1) at (-1,0)[circle, label=90:{$c_i$}, draw]{};

\node (x11) at (2,2)[circle, label=90:{$x_1^i$}, draw]{};
\node (x12) at (2,-0.5)[circle, label=90:{$x_2^i$}, draw]{};
\node (x1m) at (2,-3)[circle, label=90:{$x_{n_i}^i$}, draw]{};

\node (x1i) at (2,-1.5)[]{$\vdots$};

\node (x1mi) at (3,-3)[]{$\ldots$};
\node (x1mi) at (4,-3)[]{$\ldots$};

\draw (C1) to node [above, color=red] {$k^{10}$} (x11);
\draw[dotted, very thick] (C1) to (x11);

\draw (C1) to (x12);
\draw[dotted, very thick] (C1) to (x12);

\draw (C1) to (x1m);
\draw[dotted, very thick] (C1) to (x1m);

\node (x111) at (4,2)[circle, label=90:{$x_{1,1}^i$}, draw]{};
\node (x112) at (6,2)[circle, draw]{};
\node (x11i) at (7.5,2)[]{$\ldots$};
\node (x11k) at (9,2)[circle, label=90:{$x_{1,k}^i$}, draw]{};

\draw (x11) to node [above, color=red] {$k^5$} (x111) ;
\draw (x111) to node [above, color=red] {$k^5$}(x112);
\draw (x112) to node [above, color=red] {$k^5$}(x11i);
\draw (x11i) to node [above, color=red] {$k^5$}(x11k);

\draw[dotted, very thick] (x11) to (x111);
\draw[dotted, very thick] (x111) to (x112);
\draw[dotted, very thick] (x112) to (x11i);
\draw[dotted, very thick] (x11i) to (x11k);

\node (x111+1) at (6,4)[circle, scale=1/2, draw]{};
\draw (x111) to node [above, color=red] {$k^5$} (x111+1);
\draw[dotted, very thick] (x111) to (x111+1);

\node (C1C2) at (12,3)[draw]{$ i,1$};

\node (C1C3) at (14,1)[color=gray, draw]{$i,2$};

\draw (x111+1) to [out=90, in=110](C1C2);
\draw (x111+1) to [out=70, in=115](C1C2);
\draw (x111+1) to [out=50, in=120](C1C2);

\draw[dotted, very thick] (x111+1) to [out=90, in=110](C1C2);
\draw[dotted, very thick] (x111+1) to [out=70, in=115](C1C2);
\draw[dotted, very thick] (x111+1) to [out=50, in=120](C1C2);

\node (x111-1) at (6,1)[circle, scale=1/2, draw]{};
\draw (x111) to node [above, color=red] {$k^5$}(x111-1);
\draw[dotted, very thick] (x111) to (x111-1);

\draw (x111-1) to [out=20, in=-90](C1C2);
\draw (x111-1) to [out=-0, in=-85](C1C2);
\draw (x111-1) to [out=-20, in=-80](C1C2);

\draw[dotted, very thick] (x111-1) to [out=20, in=-90](C1C2);
\draw[dotted, very thick] (x111-1) to [out=0, in=-85](C1C2);
\draw[dotted, very thick] (x111-1) to [out=-20, in=-80](C1C2);

\node (t) at (14,-2)[circle, label={$t$}, draw]{};
\draw (x11k) to (t);
\draw[dotted, very thick] (x11k) to (t);
\draw (x11k) to [out=0](t);
\draw[dotted, very thick] (x11k) to [out=0](t);

\node (s) at (-3,0)[circle, label={$s$}, draw]{};

\foreach \x in {-40,-35,...,40} {
\draw (s) to [out=\x, in=180-\x](C1);
\draw[dotted, very thick] (s) to [out=\x, in=180-\x](C1);
}

\node (x121) at (4,-0.5)[circle, draw]{};
\node (x122) at (6,-0.5)[circle, draw]{};
\node (x12i) at (7.5,-0.5)[]{$\ldots$};
\node (x12k) at (9,-0.5)[circle, draw]{};

\draw (x12) to (x121);
\draw (x121) to (x122);
\draw (x122) to (x12i);
\draw (x12i) to (x12k);
\draw[dotted, very thick] (x12) to (x121) ;
\draw[dotted, very thick] (x121) to (x122);
\draw[dotted, very thick] (x122) to (x12i);
\draw[dotted, very thick] (x12i) to (x12k);

\node (x121+) at (4.5,0.5)[scale=0.5,circle,color=lightgray,draw]{};\
\draw[color=lightgray, very thin] (x121) to (x121+);
\draw[dotted, thick, color=lightgray] (x121) to (x121+);
\foreach \x in {-20,-10,...,20} {
\draw[color=lightgray, very thin] (x121+) to [out=60+\x, in=210+\x](C1C2);
\draw[dotted, thick, color=lightgray] (x121+) to [out=60+\x, in=210+\x](C1C2);
}

\node (x121-) at (4.5,-1.5)[scale=0.5,circle,color=lightgray,draw]{};\
\draw[color=lightgray, very thin] (x121) to (x121-);
\draw[dotted, thick, color=lightgray] (x121) to (x121-);
\foreach \x in {-20,-10,...,10} {
\draw[color=lightgray, very thin] (x121-) to [out=290+\x, in=300+\x](C1C2);
\draw[dotted, thick, color=lightgray] (x121-) to [out=290+\x, in=300+\x](C1C2);
}

\node (x112+) at (6.5,3)[scale=0.5,circle,color=lightgray,draw]{};\
\draw[color=lightgray, very thin] (x112) to (x112+);
\draw[dotted, thick, color=lightgray] (x112) to (x112+);
\foreach \x in {-20,-10,...,20} {
\draw[color=lightgray, very thin] (x112+) to [out=60+\x, in=150+\x](C1C3);
\draw[dotted, thick, color=lightgray] (x112+) to [out=60+\x, in=150+\x](C1C3);
}

\node (x112-) at (6.5,1.5)[scale=0.5,circle,color=lightgray,draw]{};\
\draw[color=lightgray, very thin] (x112) to (x112-);
\draw[dotted, thick, color=lightgray] (x112) to (x112-);
\foreach \x in {-20,-10,...,10} {
\draw[color=lightgray, very thin] (x112-) to [out=290+\x, in=300+\x](C1C3);
\draw[dotted, thick, color=lightgray] (x112-) to [out=290+\x, in=300+\x](C1C3);
}

\node[text width = 4cm] (tx1) at (6,6){$g(x_1^i)$ paths};
\draw[->] (tx1) -- (6.2,5.2);

\node[text width = 4cm] (tx1) at (9,-3){$\overline{g(x_1^i)}$ paths};
\draw[->] (tx1) -- (10,0.5);

\node[text width = 3.2cm] (tx3) at (-1,3){$(k-1)\cdot M+n_i+1$ paths};
\draw[->] (tx3) -- (-2,0.75);
\end{tikzpicture}
\caption[]{Example for a selection gadget~\fbox{$i$} ($i>2$).}\label{fig:choosegadget}
\end{figure}
\fi{}%
We introduce vertex~$c_i$ corresponding to color class~$i$ in~$(G,k)$. We connect~$s$ with~$c_i$ via a~$((k-1)\cdot M+ n_i+1,k'+1)$-bundle. Each of the chains in the bundle will be in exactly one route later. We introduce the vertices~$x^i_1,\ldots, x^i_{n_i}$ in~$G'$, corresponding to the vertices~$v^i_1,\ldots, v^i_{n_i}\in V_i$, and we connect~$c_i$ to each of them by a $k^{10}$-chain. These vertices serve as hubs for the routes later; only one of them will carry almost all routes in any solution, representing the choice of a vertex into the clique.

In order to relay this choice to all the validation gadgets, we do the following. First, we attach a $k$-chain to each vertex $x^i_j$, $1\leq j \leq n_i$. Let $x^i_{j,1},\ldots,x^i_{j,k}$ denote the vertices on the chain attached to $x^i_j$, indexed by the distance on the chain to vertex~$x^i_j$; each vertex except~$x^i_{j,k}$ will make its own connection to the validation gadgets. We connect each~$x^i_{j,\ell}$, $\ell\in[k-1]$, with the vertex~$c_{i}c_{\ell'}$ in the validation gadget~\fbox{$i,\ell'$} (introduced below), where $\ell' = \ell$ if $\ell < i$ and $\ell' = \ell + 1$ otherwise. The connection is made by attaching a $(k^5,g(v^i_j),k'+1)$-feather to $x^i_{j,\ell}$ and $c_ic_{\ell'}$. 
Furthermore, to relay also the complement IDs, we connect each~$x^i_{j,\ell}$, $\ell\in[k-1]$, with the vertex $\overline{c_{i}c_{\ell'}}$ in the validation gadget~\fbox{$i,\ell'$} by attaching a $(k^5,\overline{g(v^i_j)},k'+1)$-feather to them. We $k^5$-subdivide each edge on the $k$-chain we attached to $x^i_j$, that is, we replace each edge by a $k^5$-chain. 
We apply this to all paths attached to~$x^i_1,\ldots,x^i_{n_i}$. 
This will ensure that in each color class, only the ``ID relay vertices''~$x^i_{j,\ell}$ corresponding to one ID will carry more than one route. Note that the only differences between the ID relay vertices are the second entries of the feathers, which depend on the corresponding values of the \Sids. Finally, we connect vertex $x^i_{j,k}$ with~$t$ via a $(2,k'+1)$-bundle; this vertex ensures that each $k^5$-chain between two vertices $x^i_{j,\ell}$ corresponding to the chosen ID is shared.

\subparagraph*{Validation gadgets.}
We need to check that the chosen vertices are adjacent using only their IDs. For this we encode the sums of IDs corresponding to two adjacent vertices into a bundle which has to be passed by the routes relayed from the selection gadgets. The budget will not allow to share any of the paths in this bundle. In this way, any sum of IDs has to be below a certain threshold. To get a lower bound, we also introduce bundles for sums of complement IDs of adjacent vertices. Finally, we ensure that an ``ID'' bundle and its ``complement ID'' bundle can be used simultaneously, only if they correspond to the same pair of vertices.

We now describe the construction of a validation gadget~\fbox{$i,j$}, $i,j\in[k]$, $i < j$\ifshort{}. \else{}, illustrated in \Cref{fig:valgadget}. %
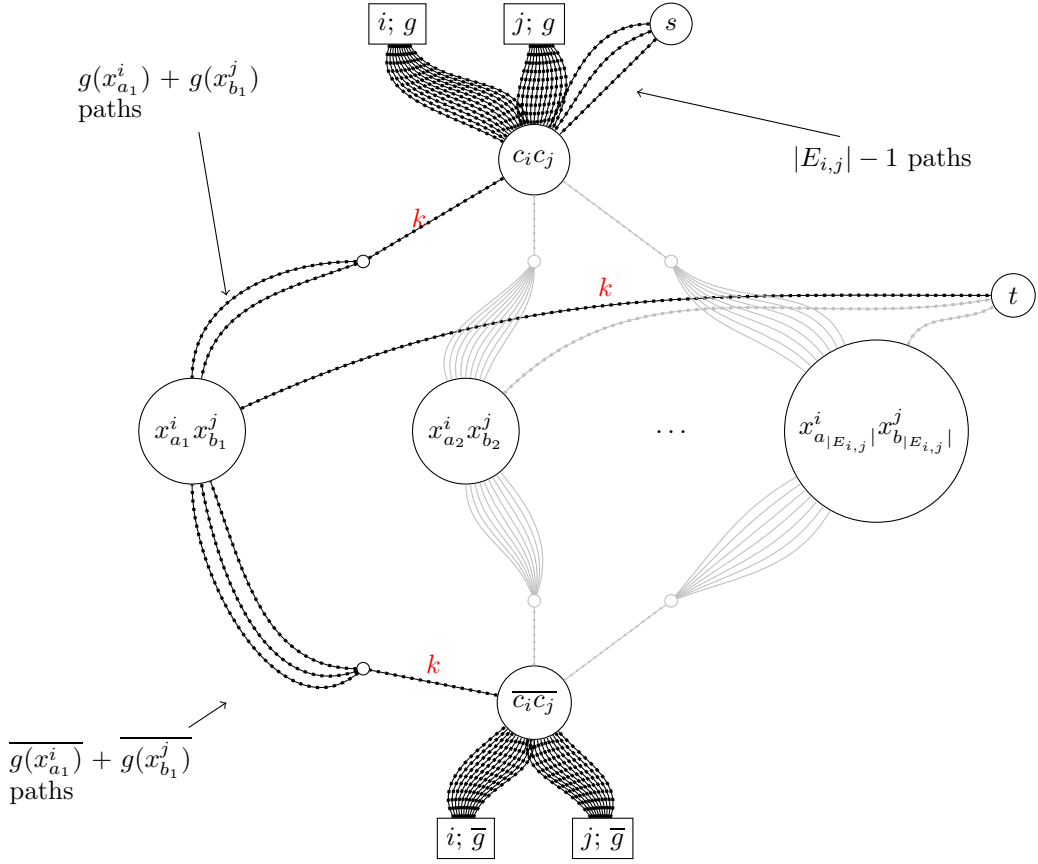
\begin{figure}[t]
\centering
\begin{tikzpicture}[scale=.9]

\node (c1c2) at (0,0)[circle,draw]{$c_ic_j$};

\node (c1c2q) at (0,-8)[circle,draw]{$\overline{c_ic_j}$};

\node (C1) at (-2,2)[draw]{$i$; $g$};
\node (C2) at (-0,2)[draw]{$j$; $g$};
\foreach \i in {60,65,...,110}{
\draw (C1) to [out=190+\i, in=210-\i](c1c2);
\draw[dotted, very thick] (C1) to [out=190+\i, in=210-\i](c1c2);
}
\foreach \i in {70,75,...,120}{
\draw (C2) to [out=190+\i, in=180-\i](c1c2);
\draw[dotted, very thick] (C2) to [out=190+\i, in=180-\i](c1c2);
}

\node (C1) at (-1,-8-2)[draw]{$i$; $\overline{g}$};
\node (C2) at (1,-8-2)[draw]{$j$; $\overline{g}$};
\foreach \i in {70,75,...,120}{
\draw (C1) to [out=\i, in=270+70-\i](c1c2q);
\draw[dotted, very thick] (C1) to [out=\i, in=270+70-\i](c1c2q);
}
\foreach \i in {70,75,...,120}{
\draw (C2) to [out=\i, in=360+20-\i](c1c2q);
\draw[dotted, very thick] (C2) to [out=\i, in=360+20-\i](c1c2q);
}

\node (x11x21) at (-5,-4)[circle, draw]{$x_{a_1}^i x_{b_1}^j$};
\node (x11x22) at (-1,-4)[circle, draw]{$x_{a_2}^i x_{b_2}^j$};
\node (xxdots) at (2,-4)[]{$\ldots$};
\node (x1ix2j) at (5,-4)[circle, draw]{$x_{a_{|E_{i,j}}|}^ix_{b_{|E_{i,j}}|}^j$};

\node (c1c2x11x21_2) at (-2.5,-1.5)[circle, scale=1/2, draw]{};
\draw (c1c2) to node [left, color=red] {$k$}(c1c2x11x21_2);
\draw[dotted, very thick] (c1c2) to (c1c2x11x21_2);

\draw (c1c2x11x21_2) to [out=180, in=90](x11x21);
\draw[dotted, very thick] (c1c2x11x21_2) to [out=180, in=90](x11x21);
\draw (c1c2x11x21_2) to [out=200, in=80](x11x21);
\draw[dotted, very thick] (c1c2x11x21_2) to [out=200, in=80](x11x21);

\node (c1c2qx11x21_2) at (-2.5,-10+2.5)[circle, scale=1/2, draw]{};
\draw (c1c2q) to node [above, color=red] {$k$}(c1c2qx11x21_2);
\draw[dotted, very thick] (c1c2q) to (c1c2qx11x21_2);

\draw (c1c2qx11x21_2) to [out=220, in=-90](x11x21);
\draw[dotted, very thick] (c1c2qx11x21_2) to [out=220, in=-90](x11x21);
\draw (c1c2qx11x21_2) to [out=200, in=-80](x11x21);
\draw[dotted, very thick] (c1c2qx11x21_2) to [out=200, in=-80](x11x21);
\draw (c1c2qx11x21_2) to [out=180, in=-70](x11x21);
\draw[dotted, very thick] (c1c2qx11x21_2) to [out=180, in=-70](x11x21);

\node (t) at (7,-2)[circle, draw]{$t$};

\draw (x11x21) to [out=25,in=180]node [above, color=red]{$k$}(t);
\draw[dotted, very thick] (x11x21) to [out=25,in=180](t);

\node (s) at (2,2)[circle, draw]{$s$};

\draw (s) to (c1c2);
\draw[dotted, very thick] (s) to (c1c2);
\draw (s) to [out=180, in=65](c1c2);
\draw[dotted, very thick] (s) to [out=180, in=65](c1c2);
\draw (s) to [out=200, in=55](c1c2);
\draw[dotted, very thick] (s) to [out=200, in=55](c1c2);

\node (x11x22+) at (0,-1.5)[scale=0.5,circle,color=lightgray, draw]{};
\draw[color=lightgray] (c1c2) to (x11x22+);
\draw[color=lightgray, dotted, thick] (c1c2) to (x11x22+);
\node (x11x22-) at (0,-6.5)[scale=0.5,circle,color=lightgray,draw]{};
\draw[color=lightgray] (c1c2q) to (x11x22-);
\draw[color=lightgray, dotted, thick] (c1c2q) to (x11x22-);
\foreach \x in {70,75,...,110} {
\draw[color=lightgray] (x11x22+) to [out=160+\x, in=180-\x](x11x22);
\draw[color=lightgray] (x11x22-) to [in=200+\x, out=180-\x](x11x22);
}

\node (x1ix2j+) at (2,-1.5)[scale=0.5,circle,color=lightgray,draw]{};
\draw[color=lightgray] (c1c2) to (x1ix2j+);
\draw[color=lightgray, dotted, thick] (c1c2) to (x1ix2j+);
\node (x1ix2j-) at (2,-6.5)[scale=0.5,circle,color=lightgray,draw]{};
\draw[color=lightgray] (c1c2q) to (x1ix2j-);
\draw[color=lightgray, dotted, thick] (c1c2q) to (x1ix2j-);
\foreach \x in {20,25,...,50} {
\draw[color=lightgray] (x1ix2j+) to [out=270+\x, in=160-\x](x1ix2j);
\draw[color=lightgray] (x1ix2j-) to [in=260-\x, out=\x](x1ix2j);
}

\draw[color=lightgray] (x11x22) to  [out=45, in=190](t);
\draw[color=lightgray, dotted, very thick] (x11x22) to [out=45, in=190](t);

\draw[color=lightgray] (x1ix2j) to  [out=70, in=210](t);
\draw[color=lightgray, dotted, very thick] (x1ix2j) to [out=70, in=210](t);

\node[text width=3cm] (t1) at (-5,1) {$g(x_{a_1}^i)+g(x_{b_1}^j)$ paths};
\draw[->] (t1) -- (-4.5,-2);

\node[text width=3cm] (tx2) at (-6,-9) {$\overline{g(x_{a_1}^i)}+\overline{g(x_{b_1}^j)}$ paths};
\draw[->] (tx2) -- (-4.5,-8);

\node[text width=4cm] (tx3) at (6,0) {$|E_{i,j}|-1$ paths};
\draw[->] (tx3) -- (1.5,1);

\end{tikzpicture}
\caption[]{Example for a validation gadget~\fbox{$i,j$}.}\label{fig:valgadget}
\end{figure}%
\fi{}%
We introduce exactly two vertices $c_{i}c_j$ and $\overline{c_{i}c_j}$ (recall that these vertices already appeared in the description of the selection gadgets). We introduce a vertex for each edge between $V_i$ and $V_j$, that is, if $\{v^i_y,v^j_z\}\in E_{i,j}$, then we introduce the vertex $x^i_yx^j_z$ in $G'$. We connect each~$x^i_yx^j_z$ to $c_ic_j$ by attaching a $(k,g(v^i_y)+g(v^j_z),k'+1)$-feather, we connect $x^i_yx^j_z$ to $\overline{c_ic_j}$ by attaching a $(k,\overline{g(v^i_y)}+\overline{g(v^j_z)},k'+1)$-feather, and we connect $x^i_yx^j_z$ to the sink vertex~$t$ by attaching a $k$-chain. Only one of the connections to the sink will carry more than one route; hence, it will be possible to use only one pair of complementary bundles (corresponding to a pair of adjacent vertices). 

For technical reasons, we need that each pair of bundles carries at least one route; this is achieved by also connecting $s$ with $c_ic_j$ via an $(|E_{i,j}|-1,k'+1)$-bundle.

\iflong{}
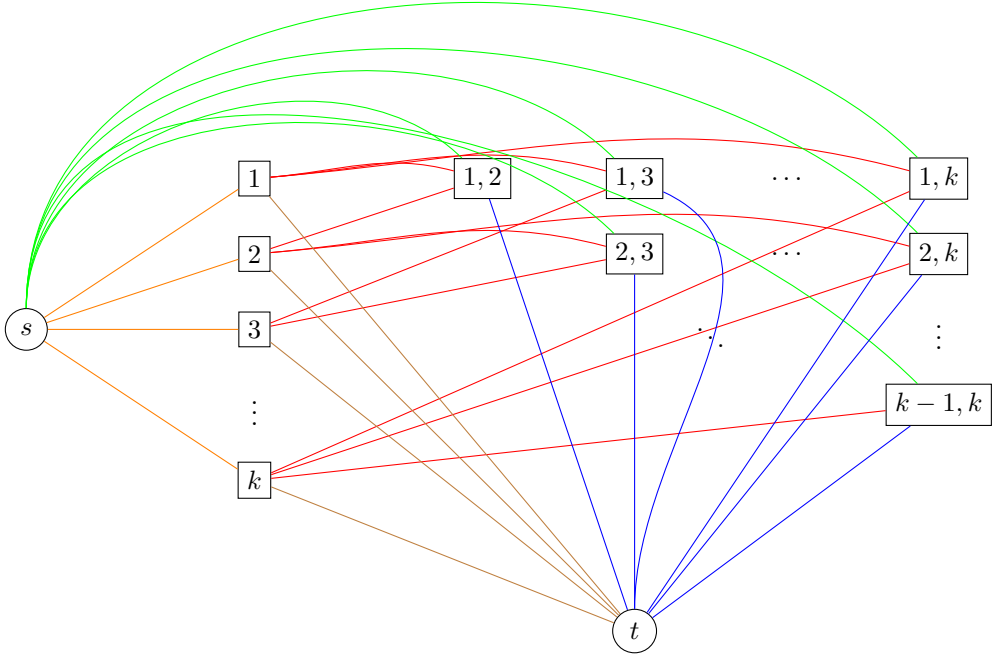
\begin{figure}[t]
\centering
\begin{tikzpicture}

\node (s) at (-1,0)[circle, draw]{$s$};

\node (t) at (7,-4)[circle, draw]{$t$};

\node (c1) at (2,2)[draw]{$1$};
\node (c2) at (2,1)[draw]{$2$};
\node (c3) at (2,0)[draw]{$3$};
\node (ci) at (2,-1)[]{$\vdots$};
\node (ck) at (2,-2)[draw]{$k$};

\node (c1c2) at (5,2)[draw]{$1,2$};
\node (c1c3) at (7,2)[draw]{$1,3$};
\node (c1ci) at (9,2)[]{$\ldots$};
\node (c1ck) at (11,2)[draw]{$1,k$};

\node (c2c3) at (7,1)[draw]{$2,3$};
\node (c2ci) at (9,1)[]{$\ldots$};
\node (c2ck) at (11,1)[draw]{$2,k$};

\node (cicip1) at (8,0)[]{$\ddots$};
\node (cick) at (11,0)[]{$\vdots$};

\node (ckm1ck) at (11,-1)[draw]{$k-1,k$};

\draw[color=red] (c1) to [out=5, in=165](c1c2);
\draw[color=red] (c1) to [out=5, in=165](c1c3);
\draw[color=red] (c1) to [out=5, in=165](c1ck);

\draw[color=red] (c2) to (c1c2);
\draw[color=red] (c3) to (c1c3);
\draw[color=red] (ck) to (c1ck);

\draw[color=red] (c2) to [out=5, in=165](c2c3);
\draw[color=red] (c2) to [out=5, in=165](c2ck);

\draw[color=red] (c3) to (c2c3);
\draw[color=red] (ck) to (c2ck);

\draw[color=red] (ck) to (ckm1ck);

\draw[color=orange] (s) to (c1);
\draw[color=orange] (s) to (c2);
\draw[color=orange] (s) to (c3);
\draw[color=orange] (s) to (ck);

\draw[color=green] (s) to [out=90](c1c2);
\draw[color=green] (s) to [out=90](c1c3);
\draw[color=green] (s) to [out=90](c1ck);
\draw[color=green] (s) to [out=90](c2c3);
\draw[color=green] (s) to [out=90](c2ck);
\draw[color=green] (s) to [out=90](ckm1ck);

\draw[color=brown] (c1) to (t);
\draw[color=brown] (c2) to (t);
\draw[color=brown] (c3) to (t);
\draw[color=brown] (ck) to (t);

\draw[color=blue] (c1c2) to (t);
\draw[color=blue] (c1c3) to [out=-25, in=90](t);
\draw[color=blue] (c1ck) to (t);
\draw[color=blue] (c2c3) to (t);
\draw[color=blue] (c2ck) to (t);
\draw[color=blue] (ckm1ck) to (t);

\end{tikzpicture}
\caption{High-level construction.}\label{fig:highlevel}
\end{figure}

In \Cref{fig:highlevel}, we give an overview of the interlinkage between the selection gadgets, the validations gadgets and $s$ and $t$.\fi{} 

\ifshort{}
\medskip
The correctness proof is deferred to the full version.
\else{}
\paragraph*{Correctness.}
\subparagraph*{\MCC \yes $\Rightarrow$ \MSE \yes.} Suppose that~$(G,k)$ is a yes-instance of \MCC, that is, $G$ contains a $k$-vertex clique with vertex set~$W$. We show that we can construct a $p$-routing in $G'$ such that at most~$k'$~edges are shared.

We construct the routes in parallel, where each one starts from $s$, traverses the selection gadgets and then the validation gadgets. In each step, we only give a partial description of the routes up to this point. The description is then successively completed.

First, we route one route over each $(k'+1)$-chain incident with $s$. In this way, each $c_i$ is incident with exactly $(k - 1) \cdot M + n_i + 1$ routes, and each $c_ic_j$ is incident with exactly $|E_{i, j}| - 1$ routes.

Next, we describe the routes within each selection gadget~\fbox{$i$}. From $c_i$ we route one route to each $x^i_\ell$. All the remaining $(k - 1) \cdot M + 1$ routes incident with~$c_i$ are routed to the vertex $x^i_{\ell'}$ corresponding to the vertex in $V_i \cap W$. So far, this induces $k \cdot k^{10}$ shared edges overall. The routes incident with $x^i_{\ell'}$ continue as follows. Note that the barbs in the feathers incident with the $x^i_{\ell', j}$, $j \in [k - 1]$, and the bundle incident with $x^i_{\ell', k}$ count exactly $2 + \sum_{j = 1}^{k - 1}g(v^i_{\ell'}) + \overline{g(v^i_{\ell'})} = 2 + (k - 1) \cdot M$ chains. This is also the number of routes incident with $x^i_{\ell'}$. We route one route incident with $x^i_{\ell'}$ over each of these chains. This induces $(k + 2(k - 1)) \cdot k^5$ further shared edges in each selection gadget. For each $x^i_\ell$, $\ell \neq \ell'$, we route the single route incident with this vertex to $t$ via $x^i_{\ell, k}$; no edge is shared in this way. Overall, we have used $k \cdot k^{10} + k \cdot (k + 2(k - 1)) \cdot k^5$ shared edges so far.

Finally, we describe the routes within each validation gadget~\fbox{$i, j$}. Note that, in the way we have defined the routes so far, each $c_ic_j$ is incident with $|E_{i, j}| - 1 + g(w_i) + g(w_j)$ routes, where $w_i$ is the vertex in $W \cap V_i$ and $w_j$ the vertex in $W \cap V_j$. Since $w_i$ and $w_j$ are adjacent, there is a $(k, g(w_i) + g(w_j), k' + 1)$-feather connecting $c_ic_j$ with some vertex~$x^i_yx^j_z$ that represents the edge $\{w_i, w_j\}$. We route $g(w_i) + g(w_j)$ routes from $c_ic_j$ to $x^i_yx^j_z$ and then to $t$, introducing $2k$ further shared edges. The remaining $|E_{i, j}| - 1$ routes incident with $c_ic_j$ are routed via the remaining, unused feathers and then to $t$, without introducing more shared edges. Similarly, $\overline{g(w_i)} + \overline{g(w_j)}$ routes are incident with $\overline{c_ic_j}$ and there is a $(k, \overline{g(w_i)} + \overline{g(w_j)}, k' + 1)$-feather connecting $\overline{c_ic_j}$ to $x^i_yx^j_z$. We route each of the routes incident with $\overline{c_ic_j}$ over this feather to $x^i_yx^j_z$ and then to~$t$, introducing $k$ more shared edges (only the shaft of the feather between $\overline{c_ic_j}$ and $x^i_yx^j_z$ is additionally shared). Doing this for all validation gadgets, we introduce $\binom{k}{2} \cdot 3k$ further shared edges. 

All $p$ routes connect~$s$ and~$t$, and exactly $k \cdot k^{10} + k \cdot (k + 2(k - 1)) \cdot k^5 + \binom{k}{2} \cdot 3k = k'$~edges are shared. Hence, we have demonstrated that a suitable $p$-routing exists.

\subparagraph*{\MSE \yes $\Rightarrow$ \MCC \yes.} Suppose that there is a $p$-routing \P\ with $k'$ shared edges in $G'$. We consider the routes in \P\ as directed from $s$ to~$t$. Observe that $s$ is incident with exactly $p$ $(k'+1)$-chains and thus, each $(k'+1)$-chain is in exactly one route. This leads to the following \cref{inv:1}.

\begin{observation}\label{inv:1}
  Each $(k'+1)$-chain incident with $s$ is in exactly one route in \P. Each vertex $c_i$, $i\in [k]$, appears in at least $(k-1)\cdot M + n_i + 1$ routes and each vertex $c_ic_j$, $1\leq i<j\leq k$, appears in at least $|E_{i,j}|-1$ routes.
\end{observation}

Each $c_i$ is connected to $n_i$ vertices corresponding to the vertices in $V_i$ via $k^{10}$-chains. By \cref{inv:1}, at each $c_i$, $(k-1)\cdot M + n_i + 1$ routes in~\P\ are distributed over $n_i$~vertices. Thus, at least one of the $k^{10}$-chains is shared by at least two routes in for each selection gadget~\fbox{$i$}. Since our budget~$k'$ does not allow for $(k + 1) \cdot k^{10}$ shared edges, also exactly one $k^{10}$-chain is shared in each selection gadget. Hence, in each selection gadget, there is exactly one vertex~$x^i_{\ell_i}$ that is incident with at least two routes in \P. Denote by $W$ the set of vertices in $G$ that these~$x^i_{\ell_i}$ correspond to. Clearly, $W$ is of size~$k$. We claim furthermore that $W$ is a clique.

To show the claim, we use the following. 
\begin{observation}\label{obs2}
  For each $i \in [k]$, vertex $x^i_{\ell_i}$ is incident with exactly $(k-1)\cdot M + 2$ routes in \P. Moreover, each of these routes traverses first $x^i_{\ell_i}$ and then exactly one barb in a feather incident with any $x^i_{\ell_i, j}$, $j \in [k - 1]$.
\end{observation}
To see this, denote by $\mathcal{C}_i$ the set containing the barbs in the feathers incident with the $x^i_{\ell_i, j}$, $j \in [k - 1]$, and the chains in the bundle incident with $x^i_{\ell_i, k}$. Note that
\[|\mathcal{C}_i| = 2 + \sum_{j = 1}^{k - 1}g(v^i_{\ell_i}) + \overline{g(v^i_{\ell_i})} = 2 + (k - 1) \cdot M.\]
Any route that contains~$x^i_{\ell_i}$ either first traverses $x^i_{\ell_i}$ and then a chain in $\mathcal{C}_i$ or vice versa, because the selection gadgets are trees except for the feathers they contain. Since each of these chains in $\mathcal{C}_i$ has length~$k' + 1$, no chain can carry two routes. Thus, $x^i_{\ell_i}$ is incident with at most $2 + (k - 1) \cdot M$~routes. Since each selection gadget gets at least $(k-1)\cdot M + n_i + 1$ routes via $c_i$, of which at most $n_i - 1$ can avoid $x^i_{\ell_i}$, each $x^i_{\ell_i}$ gets also exactly $2 + (k-1)\cdot M$ routes. Thus, indeed \cref{obs2} holds.

\cref{obs2} implies that each shaft of a feather incident with some $x^i_{\ell_i, j}$, $j \in [k - 1]$ is shared, and each $k^5$-chain connecting two $x^i_{\ell_i, j}$, $j \in [k - 1]$, is also shared. The shared edges within the selection gadgets thus amount to at least $k\cdot k^{10} + k \cdot (k + 2(k - 1)) \cdot k^5$, leaving a budget of at most $\binom{k}{2} \cdot 3k$.

Let $w_i$ be the vertex in $W$ corresponding to $x^i_{\ell_i}$, $i \in [k]$ (remember that $x^i_{\ell_i}$ is the vertex in selection gadget~\fbox{$i$} that carries at least two routes of \P). By \cref{obs2} each $c_ic_j$ appears in at least $g(w_i) + g(w_j)$ routes in \P\ and by \cref{inv:1}, $c_ic_j$ appears in $|E_{i,j}| - 1$ further routes in \P. We claim that out of these $g(w_i) + g(w_j) + |E_{i,j}| - 1$ routes, at most $n_i + n_j - 2$ routes traverse $c_ic_j$ and then, later on, traverse some vertex in a selection gadget. Call such routes \emph{unbehaved}. To see the claim, observe that each unbehaved route has to traverse a feather incident with some $x^i_{\ell,j}$ or $x^j_{\ell,i}$. However, it cannot traverse the feathers incident with $x^i_{\ell_i,j}$ or $x^j_{\ell_j,i}$ because that would mean that one of their barbs were shared by \cref{obs2}. Furthermore, at most $n_i + n_j - 2$ unbehaved routes can traverse a feather incident with some $x^i_{\ell, o}$, $\ell \neq \ell_i$, as otherwise a shaft of one of these feathers would be shared. This contradicts our remaining budget of $\binom{k}{2} \cdot 3k$. Thus, each $c_ic_j$ has at least $g(w_i) + g(w_j) + |E_{i,j}| - 1 - (n_i + n_j - 2)$ behaved routes. Denote their number by $r_{i,j}$ and observe that each behaved route containing $c_ic_j$ traverses only vertices of the validation gadget~\fbox{$i, j$}. 

By the same arguments as above, the number~$\overline{r_{i,j}}$ of routes that contain $\overline{c_ic_j}$ and then traverse only vertices of the validation gadget~\fbox{$i,j$} is at least $\overline{g(w_i)} + \overline{g(w_j)} - (n_i + n_j - 2)$. Note that, hence, at least $r_{i, j} + \overline{r_{i,j}}$ routes go from $c_ic_j$ and $\overline{c_ic_j}$ to $t$ \emph{within} the validation gadget~\fbox{$i, j$}.

By the properties of the mapping $g$, it holds that $r_{i, j}, \overline{r_{i,j}}, r_{i, j} + \overline{r_{i,j}}>n^2$. Since there are strictly less than $n^2$ edges in $E_{i,j}$, at least two $k$-shafts and at least one $k$-chain is thus shared in the validation gadget~\fbox{$i,j$}. Together with the remaining budget of $\binom{k}{2} \cdot 3k$, this observation implies that there are exactly two $k$-shafts and exactly one $k$-chain shared in the validation gadget~\fbox{$i,j$}. Hence, there is exactly one vertex corresponding to an edge in $E_{i,j}$ that appears in at least two routes. Denote this vertex by~$x^i_{y_i}x^j_{z_j}$.

Denote by $r'_{i,j}$ the number of routes that traverse $c_ic_j$ and then directly use the feather incident with~$x^i_{y_i}x^j_{z_j}$ and denote by $\overline{r'_{i,j}}$ the number of routes that traverse $\overline{c_ic_j}$ and then directly use the feather incident with~$x^i_{y_i}x^j_{z_j}$. Since $x^i_{y_i}x^j_{z_j}$ is incident with the only feathers in validation gadget \fbox{$i,j$} whose shafts are shared, $r'_{i, j} \geq r_{i ,j} - |E_{i, j}| + 1$. Similarly, $\overline{r'_{i,j}} \geq \overline{r_{i,j}} - |E_{i, j}| + 1$.

Since none of the barbs in the validation gadgets are shared, we have $r'_{i, j} \leq g(w'_i) + g(w'_j)$ and $\overline{r'_{i, j}} \leq \overline{g(w'_i)} + \overline{g(w'_j)}$ for two adjacent vertices $w_i'\in V_i$ and $w_j'\in V_j$. That is,
\begin{align*}
  g(w_i)+g(w_j) - (n_i+n_j-2) &\ \leq\  r'_{i,j} \ \leq\ g(w_i')+g(w_j'), \text{ and} \\
 \overline{g(w_i)}+\overline{g(w_j)}  - ( (n_i+n_j-2) + |E_{i,j}|-1) &\ \leq\ \overline{r'_{i,j}} \ \leq\ \overline{g(w_i')}+\overline{g(w_j')}, \text{ implying that}\\
 g(w'_i)+g(w'_j)  - ((n_i+n_j-2) + |E_{i,j}|-1) &\ \leq\  g(w_i)+g(w_j).
\end{align*}
We have $((n_i + n_j - 2) + |E_{i,j}| - 1) \leq n^2 + 2n < n^3$ for each $n > 2$ and hence, \[|(g(w_i') + g(w_j')) - (g(w_i) + g(w_j))| < n^3.\] Thus, by \cref{ss:sumdiff} of $g$ we have $\{w_i, w_j\} = \{w_i', w_j'\}$ implying that $\{w_i,w_j\}\in E_{i,j}$. It follows that there is an edge in~$G$ between any pair of vertices in $W$, and $W$ contains $k$ vertices of each color class. Thus, $G[W]$ is a $k$-vertex clique.
\fi{}

\paragraph*{Upper-Bound on the Treewidth.}
To construct a tree decomposition of small width, we start out with a single bag~$A$, where $
 A := \{s\} \cup \{t\} \cup \{c_i\mid i\in [k]\} \cup \{c_ic_j\mid 1\leq i<j\leq k\} \cup \{\overline{c_ic_j}\mid 1\leq i<j\leq k\}.$ Note that $|A|=2 + k + 2\binom{k}{2}$. Since all gadgets are interconnected via only vertices from the set~$A$, in order to construct a tree decomposition for $G'$, we can build a tree decomposition~$\mathbb{T}'$ of each gadget separately, then add $A$ to each of its bags, and then attach $\mathbb{T}'$ to the bag~$A$ we started with. Observe that each chain, bundle, and feather is a series-parallel graph. Since each gadget allows a tree-like structure\iflong{} (cf.\ \Cref{fig:treelikestrucs})\fi{} where each edge corresponds to a series-parallel graph and each leaf is contained in~$A$, we can find a tree decomposition of width at most $4$ for each gadget. Hence, the treewidth of the graph~$G'$ as constructed above is upper-bounded by $2\binom{k}{2}+ k+2+4$.
\iflong{}
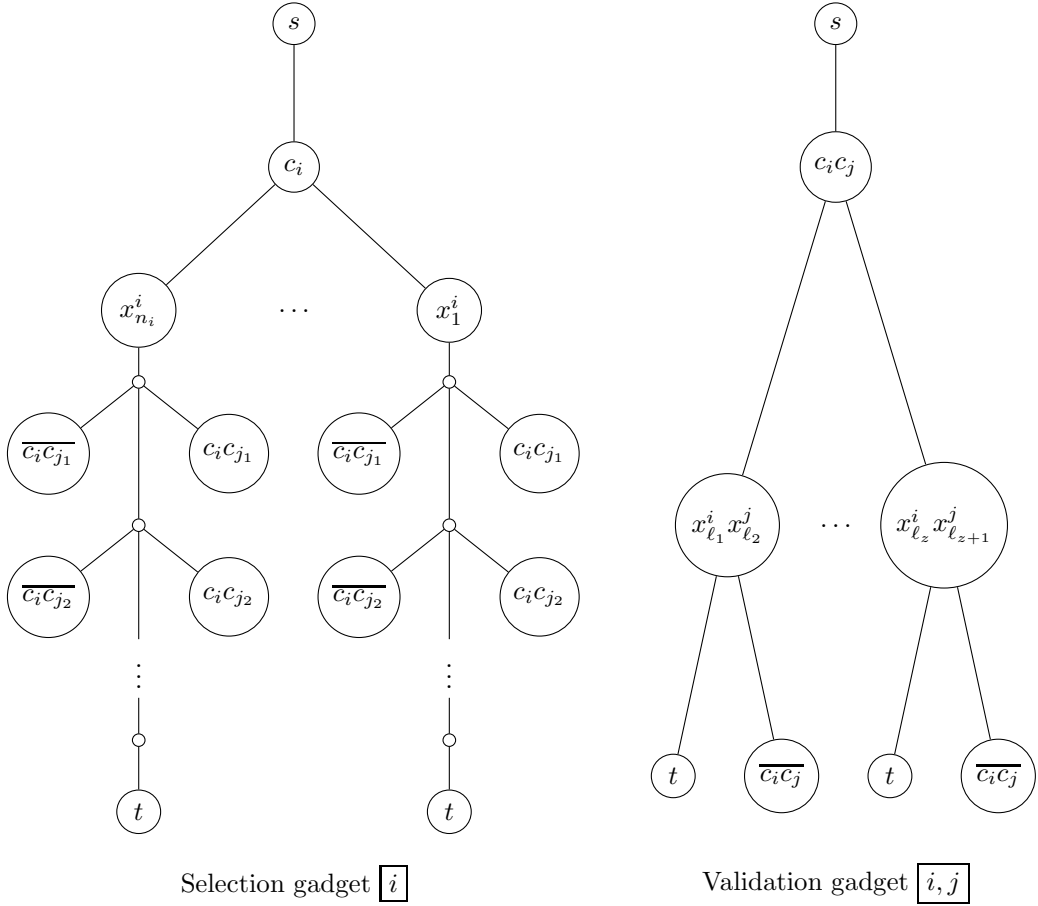
\begin{figure}[!t]
\centering
\begin{tikzpicture}[scale=.95]

\def\xwto{0.15}

\node (s) at (0,2)[circle,draw]{$s$};
\node (ci) at (0,0)[circle,draw]{$c_i$};
\draw (s) to (ci);
\node (xi1) at (2+\xwto,-2)[circle,draw]{$x^i_1$};
\node (xij) at (0,-2)[]{$\ldots$};
\node (xink) at (-2-\xwto,-2)[circle,draw]{$x^i_{n_i}$};

\draw (ci) to (xi1);
\draw (ci) to (xink);

\node  (a1) at (2+\xwto,-3)[scale=0.5, circle,draw]{};
\draw (xi1) to (a1);
\node  (a2) at (2+\xwto,-5)[scale=0.5, circle,draw]{};
\draw (a1) to (a2);
\node  (ai) at (2+\xwto,-7)[]{$\vdots$};
\draw (a2) to (ai);
\node  (a3) at (2+\xwto,-8)[scale=0.5, circle,draw]{};
\draw (ai) to (a3);
\node (t) at (2+\xwto,-9)[circle, draw]{$t$};
\draw (a3) to (t);

\node (cicj1) at (3.25+\xwto,-4)[circle,draw]{$c_ic_{j_1}$};
\draw (a1) to (cicj1);
\node (cicj1) at (0.75+\xwto,-4)[circle, draw]{$\overline{c_ic_{j_1}}$};
\draw (a1) to (cicj1);

\node (cicj2) at (3.25+\xwto,-6)[circle,draw]{$c_ic_{j_2}$};
\draw (a2) to (cicj2);
\node (cicj2) at (0.75+\xwto,-6)[circle, draw]{$\overline{c_ic_{j_2}}$};
\draw (a2) to (cicj2);

\node  (a1) at (-2-\xwto,-3)[scale=0.5, circle,draw]{};
\draw (xink) to (a1);
\node  (a2) at (-2-\xwto,-5)[scale=0.5, circle,draw]{};
\draw (a1) to (a2);
\node  (ai) at (-2-\xwto,-7)[]{$\vdots$};
\draw (a2) to (ai);
\node  (a3) at (-2-\xwto,-8)[scale=0.5, circle,draw]{};
\draw (ai) to (a3);
\node (t) at (-2-\xwto,-9)[circle, draw]{$t$};
\draw (a3) to (t);

\node (cicj1) at (-2-\xwto+1.25,-4)[circle,draw]{$c_ic_{j_1}$};
\draw (a1) to (cicj1);
\node (cicj1) at (-2-\xwto-1.25,-4)[circle, draw]{$\overline{c_ic_{j_1}}$};
\draw (a1) to (cicj1);

\node (cicj2) at (-2-\xwto+1.25,-6)[circle,draw]{$c_ic_{j_2}$};
\draw (a2) to (cicj2);
\node (cicj2) at (-2-\xwto-1.25,-6)[circle, draw]{$\overline{c_ic_{j_2}}$};
\draw (a2) to (cicj2);

\node (tx1) at(0,-10)[]{Selection gadget \fbox{$i$}};

\def\x{7.5};
\def\xwts{1.5};

\node (s) at (\x,2)[circle,draw]{$s$};
\node (cicj) at (0+\x,0)[circle,draw]{$c_ic_j$};
\draw (s) to (cicj);

\node (xi1xj2) at (0+\x-\xwts,-5)[circle,draw]{$x^i_{\ell_1}x^j_{\ell_2}$};
\draw (cicj) to (xi1xj2);
\node (cicjq) at  (0+\x-\xwts+0.75,-8.5)[circle,draw]{$\overline{c_ic_j}$};
\draw (xi1xj2) to (cicjq);
\node (t) at  (0+\x-\xwts-0.75,-8.5)[circle,draw]{$t$};
\draw (xi1xj2) to (t);

\node (xi3xj4) at (0+\x+\xwts,-5)[circle,draw]{$x^i_{\ell_z}x^j_{\ell_{z+1}}$};
\draw (cicj) to (xi3xj4);
\node (cicjq) at  (0+\x+\xwts+0.75,-8.5)[circle,draw]{$\overline{c_ic_j}$};
\draw (xi3xj4) to (cicjq);
\node (t) at  (0+\x+\xwts-0.75,-8.5)[circle,draw]{$t$};
\draw (xi3xj4) to (t);

\node (ldts) at (\x, -5)[]{$\ldots$};
\node (tx1) at(\x,-10)[]{Validation gadget \fbox{$i,j$}};

\end{tikzpicture}
\caption[]{Tree-like structures of the selection gadget~\fbox{$i$} and the validation gadget~\fbox{$i,j$}.}\label{fig:treelikestrucs}
\end{figure}
\fi{}
%
%


\section{Conclusion} \label{sec:concl}
\MSEl{} (\MSE) is a fundamental NP-hard network routing problem. 
We focused on exact solutions for the case of undirected, general graphs and 
provided several classification results concerning the parameterized
complexity of \MSE. 

It is fair to say that our fixed-parameter tractability results 
(based on tree decompositions and the treewidth reduction 
technique~\cite{MarxOR13}) are still far from practical relevance.
Our studies indicated, however, that 
\MSE is a natural candidate for performing a wider 
multivariate complexity analysis~\cite{FJR13,Nie10} as well 
as studying restrictions to special graph classes.
For instance, there is a simple search tree algorithm 
solving \MSE in $O((p-1)^k\cdot (m+n)^2)$ time which might 
be useful in some applications~\cite{Flu15}.
Moreover, it can be shown that on unbounded undirected grids (without holes),
due to combinatorial arguments, \MSE 
can be decided in constant time after reading the input~\cite{Flu15}.
In contrast, \MSE remains NP-hard when 
restricted to planar graphs of maximum degree four (which might be of particular relevance when 
studying street networks), and to directed planar graphs of maximum out- and indegree three~\cite{FluschnikSorge16}.
We consider it as interesting whether the running times of known FPT-algorithms (\Cref{sec:dp},\cite{AokiHHIKZ14,YeLLZ13}) for \MSEl
can be improved for the problem restricted to planar graphs.


In the known (pseudo) polynomial-time algorithms for graphs of bounded treewidth the exponents in the running time depend exponentially on the treewidth~\cite{AokiHHIKZ14,YeLLZ13}. It would be interesting to know whether a polynomial dependence is achievable.

\looseness=-1 A further line of future work is to study
closely related problems and natural variants of \MSE.
For instance, can the positive results
be transferred to the more general \textsc{Minimum Vulnerability} problem~\cite{AssadiENYZ12}
(see the introductory section)? 
There are also some preliminary investigations concerning the problem
\textsc{Short Minimum Shared Edges} (with  an additional 
upper bound on the maximum length of a route)~\cite{Flu15}. Finally,
it is natural to study ``time-sharing'' aspects for the shared edges,
yielding a further natural variant of \MSE.

Recently, Gutin et al.~\cite{GutinJW15} proved that the Mixed Chinese Postman Problem is W[1]-hard with respect to the treewidth of input graph, but fixed-parameter tractable with respect to the tree-depth of the input graph. 
Since we showed that \MSE is W[1]-hard with respect to the treewidth of the input graph, we consider tree-depth as an interesting parameter for a parameterized complexity analysis. 

%


\bibliographystyle{plain}
\bibliography{mse_long}


\end{document}